\documentclass{lmcs} 
\pdfoutput=1

\usepackage{lastpage}
\lmcsdoi{20}{4}{24}
\lmcsheading{}{\pageref{LastPage}}{}{}%
{Dec.~12,~2022}{Dec.~12,~2024}{}

\usepackage[utf8]{inputenc}%
\usepackage[british]{babel}%


\usepackage{amsmath}%
\usepackage{url}%
\usepackage{amssymb}%
\usepackage{mathpartir}%
\usepackage{alltt}%
\usepackage{graphicx}
\usepackage[sort]{cite}
\usepackage{tikz}%
\usetikzlibrary{calc,arrows,matrix}%

\newif\ifusecolour    
\usecolourtrue

\newif\ifrevision     
\revisiontrue

\newcommand{\ie}{\emph{i.e.}}
\newcommand{\eg}{\emph{e.g.}}

\newcommand{\CH}[1]{\mathbf{#1}} 
\newcommand{\GOTO}{\textsc{goto}}
\newcommand{\SECD}{SECD}

\ifusecolour
\definecolor{dimgray}{gray}{0.50}
\definecolor{myredcol}{rgb}{0.9,0.17,0.31}   
\definecolor{mygreencol}{rgb}{1.0,0.65,0.0}  
\definecolor{mybluecol}{rgb}{0.0,0.58,0.71}  
\newcommand{\myred}{myredcol}
\newcommand{\mygreen}{mygreencol}
\newcommand{\myblue}{mybluecol}
\newcommand{\R}[1]{\textcolor{\myred}{#1}}
\newcommand{\E}[1]{\textcolor{\mygreen}{#1}}
\newcommand{\B}[1]{\textcolor{\myblue}{#1}}
\newcommand{\D}[1]{\textcolor{dimgray}{#1}}
\else
\newcommand{\R}[1]{#1}
\newcommand{\E}[1]{#1}
\newcommand{\B}[1]{#1}
\newcommand{\D}[1]{#1}
\fi

\newcommand{\stgy}[1]{{\mathsf{#1}}}
\newcommand{\st}{\stgy{st}} 
\newcommand{\id}{\stgy{id}} 
\newcommand{\bn}{\stgy{bn}} 
\newcommand{\bv}{\stgy{bv}} 
\newcommand{\no}{\stgy{no}} 
\newcommand{\ao}{\stgy{ao}} 
\newcommand{\hr}{\stgy{hr}} 
\newcommand{\he}{\stgy{he}} 
\newcommand{\hn}{\stgy{hn}} 
\newcommand{\ha}{\stgy{ha}} 
\newcommand{\sn}{\stgy{sn}} 
\newcommand{\am}{\stgy{am}} 
\newcommand{\ho}{\stgy{ho}} 
\newcommand{\so}{\stgy{so}} 
\newcommand{\bs}{\stgy{bs}} 
\newcommand{\rn}{\stgy{rn}} 
\newcommand{\su}{\stgy{su}} 
\newcommand{\hy}{\stgy{hy}} 
\newcommand{\ea}{\stgy{ea}} 
\newcommand{\ev}{\stgy{ev}} 
\newcommand{\rb}{\stgy{rb}} 
\newcommand{\ar}{\stgy{args}}   
\newcommand{\bo}{\stgy{bodies}} 

\newcommand{\GL}[1]{\mathcal{#1}}

\newcommand{\la}{\R{{\stgy{la}}}}
\newcommand{\opone}{\stgy{op_1}}
\newcommand{\arone}{\E{\stgy{ar_1}}}
\newcommand{\optwo}{\stgy{op_2}}
\newcommand{\artwo}{\B{\stgy{ar_2}}}

\newcommand{\rel}{\to}

\newcommand{\cas}[3]{[#1/#2](#3)}
\newcommand{\OMEGA}{\mathbf{\Omega}}
\newcommand{\hole}{\scalebox{.6}[1]{\ensuremath\Box}}
\newcommand{\hyb}{\mathbin{\Diamond}}
\newcommand{\ctx}{\mathcal{C}}

\newcommand{\NT}[1]{\mathsf{#1}}
\newcommand{\Var}{\rotatebox[origin=c]{180}{$\Lambda$}}
\newcommand{\Neu}{\NT{Neu}}
\newcommand{\NF}{\NT{NF}}
\newcommand{\HNF}{\NT{HNF}}
\newcommand{\WHNF}{\NT{WHNF}}
\newcommand{\WNF}{\NT{WNF}}

\newcommand{\III}{\bn}
\newcommand{\IIS}{\stgy{IIS}}
\newcommand{\ISI}{\stgy{ISI}}
\newcommand{\ISS}{\bv}
\newcommand{\SII}{\he}
\newcommand{\SIS}{\stgy{SIS}}
\newcommand{\SSI}{\ho}
\newcommand{\SSS}{\ao}

\newcommand{\IIH}{\stgy{IIH}}
\newcommand{\SIH}{\stgy{SIH}}
\newcommand{\HII}{\stgy{HII}}
\newcommand{\HIS}{\stgy{HIS}}
\newcommand{\HIH}{\stgy{HIH}}

\newcommand{\ISH}{\stgy{ISH}}
\newcommand{\SSH}{\stgy{SSH}}
\newcommand{\HSI}{\stgy{HSI}}
\newcommand{\HSS}{\stgy{HSS}}
\newcommand{\HSH}{\stgy{HSH}}

\newcommand{\IHH}{\stgy{IHH}}
\newcommand{\SHH}{\stgy{SHH}}
\newcommand{\HHI}{\stgy{HHI}}
\newcommand{\HHS}{\stgy{HHS}}
\newcommand{\HHH}{\stgy{HHH}}

\newcommand{\RE}{\stgy{RE}}

\newcommand{\I}{\stgy{I}}

\newcommand{\Ev}{\stgy{E}}
\newcommand{\Rb}{\stgy{R}}

\newenvironment{code}%
{\vspace{4pt}\begin{alltt}\small}%
{\end{alltt}\vspace{4pt}}%

\newcounter{pi}
\newcounter{ai}

\begin{document}

\title[Equivalence of eval-readback and eval-apply evaluators]%
{Equivalence of eval-readback and eval-apply big-step evaluators by
  structuring the lambda-calculus's strategy space}%

\author[P.~Nogueira]{%
  Pablo Nogueira\lmcsorcid{0000-0002-8706-0027}%
}%
\address{UDIT University of Design, Innovation and Technology, Madrid, Spain.}%
\email{pablo.nogueira@udit.es}%

\author[\'A.~Garc\'ia-P\'erez]{\'Alvaro
  Garc\'ia-P\'erez\lmcsorcid{0000-0002-9558-6037}%
}%
\address{Université Paris-Saclay, CEA, List, F-91120, Palaiseau, France.}%
\email{alvaro.garciaperez@cea.fr}%

\begin{abstract}
  We study the equivalence between eval-readback and eval-apply big-step
  evaluators in the general setting of the pure lambda calculus. We study
  `one-step' equivalence (same strategy) and also discuss `big-step'
  equivalence (same final result). One-step equivalence extends for free to
  evaluators in other settings (calculi, programming languages, proof
  assistants, etc.) by restricting the terms (closed, convergent) while
  maintaining the strategy. We present a proof that one-step equivalence holds
  when (1) the `readback' stage satisfies straightforward well-formedness
  provisos, (2) the `eval' stage implements a `uniform' strategy, and (3) the
  eval-apply evaluator implements a `balanced hybrid' strategy that has `eval'
  as a subsidiary strategy. The proof proceeds by the `lightweight fusion by
  fixed-point promotion' program transformation on evaluator implementations
  to fuse readback and eval into the balanced hybrid. The proof can be
  followed with no previous knowledge of the transformation. We use Haskell
  2010 as the implementation language, with all evaluators written in monadic
  style to guarantee semantics (strategy) preservation, but the choice of
  implementation language is immaterial.

  To illustrate the large scope of the equivalence, we provide an extensive
  survey of the strategy space using canonical eval-apply evaluators in code
  and big-step `natural' operational semantics. We discuss the strategies'
  properties, some of their uses, and their abstract machines. We improve the
  formal definition of uniform and hybrid strategy, use it to structure the
  strategy space, and to obtain generic higher-order evaluators which are used
  in the equivalence proof. We introduce a systematic notation for both
  evaluator styles and use it to summarise strategy and evaluator
  equivalences, including (non-)equivalences within a style and
  (non-)equivalences between styles not proven by the transformation.
\end{abstract}
\maketitle
\newpage
\tableofcontents
\section*{Extended abstract}
Big-step evaluators are ubiquitous in formal semantics and implementations of
higher-order languages based on the lambda calculus, whether research
languages or core languages of functional programming languages and
proof-assistants. A big-step evaluator is an operational semantics device that
defines an evaluation strategy (an evaluation order of redexes) by specifying
how the redexes are located and contracted. It is `big-step' because it
delivers the final result of iterated small-step reduction. (Because an
evaluator defines a unique strategy, it is common to blur the distinction and
call the evaluator a strategy.) A big-step evaluator is typically written as a
recursive function in natural semantics or in code. In the latter, the
operational semantics of the implementation language must be taken into
account to guarantee semantics (strategy) preservation.

Whether in code or natural semantics, big-step evaluators are typically
written in two definitional styles, namely, \emph{eval-apply} and
\emph{eval-readback}. In eval-apply style, the evaluator is defined by a
single natural semantics or a single `eval' function in code. The `eval'
function can have a local `apply' function for application terms which can be
obviated with \texttt{let}-clauses or pattern-matching expressions. In
eval-readback style, the evaluator is defined by the composition of two
functions (two natural semantics): an `eval' function that evaluates terms up
to a certain intermediate form, and an `readback' function that pushes down
eval directly or through an eval-readback composition on specific subterms of
the intermediate form.

The problem we address is the \emph{one-step equivalence} (`define the same
strategy') between eval-readback and eval-apply big-step evaluators, and
address `big-step' equivalence (`deliver the same final result') in
passing. We use the pure lambda calculus as the most general setting with more
terms (closed, open, neutral, divergent, convergent, etc.) and strategies
(weak-reducing, head-reducing, full-reducing, strict, non-strict, etc.), where
the different evaluation criteria and their interaction give rise to different
kinds of final results (weak, head, weak head, value head, and plain normal
forms, etc.) and properties (completeness, spineness, etc.).  One-step
equivalence extends for free to evaluators in other settings (programming
languages, proof assistants, etc.) by restricting the set of terms (\eg\
type-erased closed terms for programming languages, convergent terms for proof
assistants) while the strategy is maintained.

In previous work, we showed that the one-step equivalence between one
eval-readback and one eval-apply evaluator can be proven \emph{in code} using
the `lightweight fusion by fixed-point promotion' program transformation (LWF)
which fuses the eval-readback composition into the single eval-apply evaluator
by means of strategy-preserving, simple and syntactic program-transformation
steps that are embeddable in a compiler's inlining optimisation. However, this
method has limitations. LWF has to be deployed for every pair of eval-readback
and eval-apply evaluators, some strategies have known evaluators only in one
of the two styles, and LWF is not invertible, that is, given an eval-apply
with no known eval-readback, the latter must be conjectured and must
LWF-transform to the eval-apply evaluator.  Whether an eval-readback evaluator
exists for a given eval-apply evaluator is precisely the open question of
expressive power derived from the open question of equivalence.

The main contribution of this paper is a single proof of one-step equivalence
by a novel application of LWF on an eval-readback and an eval-apply
\emph{arbitrary} evaluators which respectively call (instantiate) a
\emph{generic} readback evaluator and a \emph{generic} eval-apply evaluator
that implement as fixed-points all the concrete readback and eval-apply
evaluators for the standard (history-free, deterministic, non-parallel,
left-to-right) strategies in the literature and more. The pair of arbitrary
evaluators have type-checked but undefined (\ie\ existentially quantified)
parameters whose values are constrained by the well-formedness provisos of the
generic evaluators and by the LWF steps. The constraints are elaborated during
the transformation and collected at the end in a list of equations that
capture the conditions for one-step equivalence.

We provide an extensive overview of the eval-apply and eval-readback styles
using canonical examples from the literature. To illustrate the large scope of
the equivalence proof, we provide an extensive survey of the strategy space
using canonical eval-apply evaluators in natural semantics and code. We
discuss the strategies' diverse properties, abstract machines, and uses,
particularly in relation to evaluation of recursive functions in direct versus
continuation-passing style and thunking. We introduce new strategies in the
survey, in particular, an eager strategy that evaluates recursive functions in
delimited continuation-passing style with a non-strict fixed-point combinator
and thunking `protecting by variable' rather than the usual `by lambda'.

We motivate and arrive at an improved formal definition of uniform/hybrid
strategy that we use to orthogonally split the diverse strategy space. We
structure each uniform and hybrid strategy subspace according to the weakness,
strictness, and headness properties of strategies. The uniform space makes up
a lattice of eight strategies we call Gibbons's Beta Cube. The hybrid space is
larger and is further split by introducing the `balanced' concept (a hybrid
strategy is balanced when it is non-strict or when it uses its subsidiary
strategy, not itself, to evaluate operands of redexes). We introduce the
well-formedness provisos that make uniform and hybrid eval-apply evaluators
define uniform and hybrid strategies.

We introduce systematic notations for encoding uniform, hybrid, and
eval-readback evaluators, and use the notations to introduce new strategies,
study (non-)equivalences and other properties such as absorption, and provide
an extensive table summary of strategies, evaluators, and equivalences. We end
discussing the relevance of our contributions in the wider context of
correspondences among operational semantics devices.

We use Haskell 2010 as the implementation language, but the choice of language
is immaterial. We ignore efficiency issues because in the pure lambda calculus
there are no optimal strategies, and because efficiency can be studied
afterwards in a strategy space constrained by efficiency considerations.

We have addressed the paper to non-expert readers which may be unfamiliar with
background material. We are explanatory in the text (or in the appendices when
appropriate) for information or reminder. Expert readers can jump directly to
the sections of their interest. The introduction elaborates this extended
abstract with details, citations, and a list of contributions that serves as a
roadmap of the paper.

\section{Introduction}
\subsection{The background}
Big-step evaluators are ubiquitous in formal semantics and implementations of
higher-order languages based on the lambda calculus. A big-step evaluator is
an intensional function definition (a \emph{definiens}) for an evaluation
strategy (the unique \emph{definiendum}). An evaluation strategy (or just
`strategy') is synonymous with the evaluation order of the reducible
expressions (`redexes') of terms. A strategy is defined generically as a
function subset of the calculus's reduction relation. We are interested in
history-free sequential strategies where the choice of redex is based on fixed
positional criteria independent of previous choices. A big-step evaluator is
an operational semantics device that specifies how the redexes are located and
evaluated. It is `big-step' because it delivers the \emph{final} result of
iterated small-step reduction. A strategy can also be defined in small-step
operational semantics styles such as structural~\cite{Plo81},
context-based~\cite{Fel87,FH92,FFF09}, or abstract
machines~\cite{Lan64,DHS00}. Because an evaluator defines a unique strategy,
it is common to blur the distinction and call the evaluator a strategy. A
big-step evaluator is also called a `definitional interpreter'
\cite{Rey72}. Some authors reserve the latter for evaluators with shallow
(meta-level) embeddings, but `interpreter' (symbolic evaluation) predates the
deep versus shallow embedding dichotomy.

A big-step evaluator is typically written as a recursive function, in
code~\cite{McC60} or in natural semantics \cite{Kah87}.  In code, the
evaluator works on some data-type representation of the mathematical
terms. The operational semantics of the implementation language must be taken
into account to guarantee semantics (\ie~strategy) preservation
\cite{Rey72}. In natural semantics, the evaluator works on the mathematical
terms and is defined by a set of syntax-directed and deterministic inference
rules describing the function's behaviour recursively on subterms. A natural
semantics provides a compact definition, requires no knowledge of specific
programming languages, and serves as a specification of the evaluator's
implementation in code. In a partiality setting, a big-step evaluator may
terminate (converge) on a final result, or non-terminate (diverge) because it
cannot deliver a final result when it exists (which may be delivered by other
strategies) or because no final result exists (which may be delivered by no
strategy).

Whether in code or natural semantics, big-step evaluators are typically
written in two definitional \emph{styles}.

In \emph{eval-apply} style, the evaluator is defined by a \emph{single
  function}. This style appeared early in programming and formal semantics
\cite{McC60,Lan64,Rey72,Plo75,Sto77,SS78} and was the main style until the
1990s \cite{All78,ASS85,FH88,Hen87,Rea89,Mac90}. In its traditional code
definition, the \texttt{eval} function evaluates subterms recursively and
delegates the specific evaluation of application terms to a mutually-recursive
local \texttt{apply} function. This function can be obviated within
\texttt{eval} using \texttt{let}-clauses and pattern-matching expressions in
modern functional languages. In natural semantics, the eval-apply evaluator is
defined by a single set of inference rules for every case of term: variables,
abstractions (function bodies), and redex and non-redex applications. The
\texttt{apply} function corresponds to the inference rules for applications.
Section~\ref{sec:eval-apply} explains the eval-apply style with examples.

In \emph{eval-readback} style, the evaluator is defined by a \emph{composition
  of two functions} (also called stages), namely, \texttt{readback} after
\texttt{eval}. The name `eval-readback' comes from \cite{GL02} but the style
is older. Two classic evaluators in eval-readback style are \texttt{byName}
and \texttt{byValue} in \cite[p.\,390]{Pau96}.  The eval-readback style is
related to normalisation-by-evaluation \cite{BS91} but should not be confused
with it. In normalisation-by-evaluation, a term is evaluated by first
constructing its representative in a different target domain (evaluation
stage) and then reading the representative back to a canonical result term
(reification stage). In eval-readback there is only the domain of terms. The
\emph{less-reducing} \texttt{eval} function contracts redexes to an
intermediate result with unevaluated subterms. The \emph{further-reducing}
\texttt{readback} does not contract redexes, it `reads back' the intermediate
result to a further-reduced result by recursively distributing eval down the
unevaluated subterms to reach further redexes \cite[p.\,236]{GL02}. Readback
is a sort of `selective-iteration-of-eval' function that calls eval directly
or in an eval-readback composition on specific unevaluated subterms.
Section~\ref{sec:eval-readback} explains the eval-readback style with
examples. The separation of evaluation in two stages can be of help when
writing an evaluator. The separation also enables independent
optimisations. For instance, in \cite{GL02} the eval stage is implemented by
an abstract machine and its associated compilation scheme \cite{Ler91}.

\subsection{The problem}
\label{sec:the-problem}
We address the question of equivalence between eval-readback and eval-apply
evaluators, which partly answers the question of their expressive power,
\ie~whether a strategy can or cannot be defined in both styles. We are
primarily interested in whether the evaluators define the same strategy
(evaluate redexes in exactly the same order) and secondarily in whether the
evaluators do not define the same strategy but deliver identical final
results. We refer to the former equivalence as \emph{one-step equivalence} and
to the latter as \emph{big-step equivalence}. We aim for a \emph{single}
one-step equivalence proof for a large collection of strategies.

In previous work \cite{GPN14,GNS14}, we showed that given an eval-readback
evaluator \texttt{rb}\,$\circ$\,\texttt{ev}, its one-step equivalence with a
given eval-apply evaluator \texttt{ea} can be proven \emph{in code} using the
`lightweight fusion by fixed-point promotion' program transformation of
\cite{OS07}, which we hereafter abbreviate as LWF. (We also abbreviate
evaluator names using two letters (\eg\ \texttt{rb}, \texttt{ev}, etc.) like
\cite{Ses02}, to fit their code and natural semantics definitions in running
text and figures.)  The LWF transformation fuses
\texttt{rb}\,$\circ$\,\texttt{ev} into \texttt{ea} by means of
strategy-preserving, simple and syntactic program-transformation steps that
are embeddable in a compiler's inlining optimisation. The LWF transformation
had been previously used to prove equivalences of operational semantics
devices, but in so-called `syntactic-correspondence' proofs between small-step
evaluators and abstract machines, where LWF is used to fuse the decomposition
and recomposition steps, \eg\ \cite{Dan05,DM08,DJZ11}.

Using LWF to prove eval-apply and eval-readback equivalence has two
problems. First, using LWF for every possible strategy and pair of evaluators
is menial work.  Second, only one evaluator style is known for particular
strategies.  Third, LWF is not invertible. An eval-readback evaluator is not
obtained by `de-fusing' a given eval-apply evaluator. The eval-readback
evaluator must be conjectured, and then it must LWF-transform to the given
eval-apply evaluator. Whether an eval-readback evaluator exists for a given
eval-apply evaluator is precisely the open question of expressive power.


\subsection{The setting}
\label{sec:the-setting}
We address the problem in the setting of the \emph{pure lambda calculus} for
the following reasons. First, there are too many applied and typed lambda
calculi\footnote{\label{fn:applied-typed} An \emph{applied} extension of the
  pure lambda calculus adds primitive terms and reduction rules for them. A
  \emph{typed} extension adds type terms and type rules, and possibly modifies
  terms and reduction rules with type information. Most extensions combine
  both.} with significant differences among them. There is no principled
reason to pick one in particular. Second, the pure lambda calculus is the most
general calculus, with more terms, strategies, and concomitant big-step
evaluators. Third, we can restrict the set of terms (\eg\ closed terms,
convergent terms obtained by type-checking and type-erasure, etc.) as required
in extended settings such as programming languages and proof assistants.  We
define strategies and evaluators for arbitrary pure terms, prove equivalences
between evaluator styles, and then we can assume a restricted input term
subset so that the identical or one-step equivalent programming or
proof-assistant evaluator is \emph{de facto} included in the equivalence
(Section~\ref{sec:prelim:cbv}).

In the same line, we represent terms in code with a direct data-type
representation (a `deep embedding' \cite{BGGHHV92}), and use a direct
implementation of the capture-avoiding substitution operation of the
calculus. The results presented in this paper extend to other term
representations provided the evaluation order of redexes is preserved. The
results also extend to calculi where the strategies we study can be simulated,
\eg~\cite[Sec.\,6]{Plo75}.

We ignore efficiency issues because in the pure lambda calculus there are no
optimal strategies---there exist terms for which any strategy duplicates work
\cite{Lev78,Lev80}. Our goal is to structure the strategy space to prove
one-step equivalence. Efficiency can be studied afterwards, constraining the
strategy space by efficiency considerations. Also, efficient versions of
strategies may be undefinable in the pure lambda calculus. For example,
call-by-need is an efficient version of call-by-name that evaluates copies of
the same redex once but requires sharing and memoisation \cite{FH88}.  The
results in this paper can inform further work for such strategy spaces (see
contribution \ref{C:survey}). For the wide-range of issues regarding
efficiency see \cite{Lev80,FH88,Lam89,GAL92,AG98,RP04,AdL16,AdLV21}.

The pure lambda calculus allows free variables and hence \emph{open} terms,
with the particular case of \emph{neutral terms} (\emph{neutrals}, for
short).\footnote{We thank Noam Zeilberger for pointing us to the `neutral'
  terminology.} A free variable is not a parameter of an abstraction. The
notion is relative to a term: $x$ may be free in $B$ but bound (not free) in
$\lambda x.B$. Neutral terms are applications of the form $x\,N_1\cdots N_n$
where $n\geq1$. A neutral is not a redex, and cannot evaluate to a redex
unless a suitable term can be substituted for the leftmost variable operator
(called the `head variable'). When the head variable is free, it may be seen
as a data-type value constructor that is applied to the $n$ operands.

The pure lambda calculus allows strategies with diverse evaluation properties,
in particular:
\begin{itemize}
\item Weak vs non-weak evaluation: respectively, the non-evaluation or
  evaluation of abstractions (function bodies) as in programming
  languages. Mind the opposites: weak = non-evaluation of abstractions;
  non-weak = evaluation of abstractions.\footnote{Weak evaluation is called
    `lazy evaluation' by several authors by association with the `lazy lambda
    calculus' \cite{Abr90}. But that calculus is weak \emph{and} non-strict,
    two properties of lazy functional programming languages.}

\item Strict vs non-strict evaluation: respectively, the evaluation or
  non-evaluation of operands of redexes before substitution. Technically,
  strictness is a denotational property of functions but it can be `abused' as
  an evaluation property.

\item Head vs non-head evaluation: respectively, the non-evaluation or
  evaluation of (operands of) neutrals, which gives the name to strategies
  such as `head reduction' or `head spine'
  (Section~\ref{sec:strategies}). Mind the opposites: head = non-evaluation of
  neutrals; non-head = evaluation of neutrals.

\item Full evaluation (hereafter, `full-reduction'): evaluation to a result
  term with no redexes. Full-reduction with open terms is fundamental in
  optimisation by partial evaluation and type-checking in proof-assistants and
  automated reasoning, \eg~\cite{GL02,Cre07,ACPPW08,SR15,ABM15}. Many authors
  use `strong reduction' rather than `full-reduction'.  However, `strong' is
  the antonym of `weak' which only refers to the non-evaluation of
  abstractions.  Full-reduction `goes under lambda' (non-weak) and also under
  neutrals (non-head). Also, strong reduction may be confused with `strong
  normalisation' which is a property of terms.\footnote{A term is
    strongly-normalising when all its evaluation sequences converge on a
    normal form \cite{Bar84}.} Thus, we use `full-reduction' \cite{SR15} with
  a dash to make it a technical name.
\end{itemize}
These evaluation criteria \emph{and their interaction} give rise to different
kinds of final results and properties. For example, a weak and strict strategy
does not evaluate the abstractions within the operands of redexes it does
evaluate. A non-weak and head strategy does not evaluate the neutrals within
the abstractions it does evaluate.

\subsection{The contributions}
\label{sec:contribs}
The last contribution of this paper is a single proof of the one-step
equivalence of eval-readback and eval-apply evaluators. The proof proceeds by
a novel application of LWF on a plain but arbitrary eval-readback evaluator
and eval-apply evaluator. They are `plain' in that they respectively call
(instantiate) a \emph{generic} readback evaluator and a \emph{generic}
eval-apply evaluator that respectively implement as fixed-points all the plain
and concrete readback and eval-apply evaluators for the standard
(history-free, deterministic, non-parallel, left-to-right) strategies in the
literature and more. They are `arbitrary' in that some parameters to the
generic evaluators have undefined (but type-checked) values, which are thus
`existentially quantified' in the sense of logic programming. The LWF steps
are applied verbatim on the two plain and arbitrary evaluators. The concrete
values substitutable for the undefined parameters are constrained by
well-formedness provisos of the generic evaluators and by the LWF steps, and
are collected at the end in a list of equations that capture the conditions
for one-step equivalence.

This last contribution strongly depends on other significant contributions
that we list below in increasing order of importance and of presentation in
the paper. The list of contributions below also serves as a roadmap summary of
the paper. In Section~\ref{sec:conclusions} we provide a detailed technical
summary of the contributions that can be understood after reading the paper.
\begin{enumerate}[label=\textbf{C\arabic*}]
\item \label{C:survey} \textbf{Eval-apply style and strategy survey.} We
  provide an extensive survey of pure lambda calculus strategies defined by
  their canonical eval-apply evaluators in natural semantics
  (Section~\ref{sec:strategies}). We previously overview the eval-apply style
  in code (Section~\ref{sec:eval-apply-monadic}) and natural semantics
  (Section~\ref{sec:eval-apply-natural-sem}) using as a first example the
  call-by-value (or call-by-weak-normal-form) strategy of the pure lambda
  calculus. We discuss this strategy in detail, compare it to other
  call-by-value strategies, and justify why we can call it `call-by-value'
  (Section~\ref{sec:prelim:cbv}).  We use the Haskell 2010 standard as the
  implementation language, with all evaluators (plain and generic) written in
  monadic style to guarantee semantics (strategy) preservation. The reasons
  for choosing Haskell are discussed in Section~\ref{sec:eval-apply-monadic},
  but the choice of implementation language is ultimately immaterial

  We discuss in the survey strategies that are well-known, less known, and
  interesting novel ones. We recall their properties, their abstract machines,
  and some of their uses, particularly in the evaluation of general recursive
  functions, where we use simple and well-known Church numerals as data for
  examples. We focus on this particular use for some strategies because it
  illustrates their evaluation properties in relation to thunking
  \cite{Ing61,DH92} and to vanilla and delimited continuation-passing style
  \cite{HD94,BBD05}. We want to draw attention to the novel strategy discussed
  in contribution~\ref{C:so}.

  The survey excludes history-informed, parallel, non-deterministic, and
  right-to-left\footnote{A right-to-left strategy evaluates the operand of a
    redex before the abstraction when it is non-weak, and the operands of a
    neutral from right-to-left when it is non-head.} strategies. The survey
  also excludes strategies that cannot be defined in the pure lambda calculus
  (\eg\ call-by-need which requires sharing and memoisation).  Finally, the
  survey excludes simulations of strategies that belong to other calculi (\eg\
  Section~\ref{sec:prelim:cbv}). The survey is nonetheless extensive and
  informs future work on other strategy spaces. Strategy equivalence must be
  studied under equivalent conditions: in the same calculus, and with either
  the same term representation or with the same order of redexes (strategy)
  under different term representations.

\item \label{C:so} \textbf{Introduction of `spine applicative order'.} We
  introduce this strategy in the survey which evaluates general recursive
  functions in delimited continuation-passing style, \emph{eagerly} but with a
  \emph{non-strict} fixed-point combinator, and with thunking
  `\emph{protecting by variable}' instead of the common `protecting by lambda'
  in the literature on thunking. This and the other strategies illustrate the
  diversity of the strategy space.

\item \textbf{Characterisation of eval-readback style.} We overview the
  eval-readback style in code (Section~\ref{sec:monadic-eval-readback}) and
  natural semantics (Section~\ref{sec:natural-sem-eval-readback}) using
  well-known examples from the literature, and relate them to the eval-apply
  evaluators in the survey. We characterise the style precisely by means of
  two well-formedness provisos (Section~\ref{sec:eval-readback:provisos}), and
  discuss the limit on the amount of evaluation that can be moved between eval
  and readback, which determines equivalences within the style
  (Section~\ref{sec:eval-readback:equiv}).

\item \textbf{Generic templates and evaluators.} We present the generic
  eval-apply and the generic readback evaluators by means of natural semantics
  templates and higher-order parametric evaluator implementations in code. The
  generic evaluators deliver plain evaluators as fixed-points by
  instantiation. The generic evaluators are obtained by parametrising on all
  the \emph{variability points} of their respective style. The generic
  eval-apply evaluator delivers all the evaluators in the survey. The generic
  readback evaluator delivers all the readbacks that satisfy the provisos. The
  generic evaluators let us express the evaluator space in a single
  definition. We employ them in the equivalence proof by LWF.

\item \label{C:hybrid} \textbf{Improved definition of uniform/hybrid
    strategy.} We improve the formal definition of uniform vs hybrid strategy
  that we presented in \cite[Sec.\,4]{GPN14} and that we use in this paper to
  structure the strategy space in contribution \ref{C:struct}. Intuitively, a
  hybrid (or layered, or stratified) strategy inextricably \emph{depends} on
  other less-reducing subsidiary strategies to evaluate particular subterms
  such as operators or operands.  When there is no dependency, the strategy is
  uniform. The dependency shows in the evaluation sequences which include
  whole evaluation sequences of the subsidiaries for those subterms. The
  dependency also usually shows in the \emph{definitional device} used for
  defining the hybrid strategy (code or operational semantics, whether
  big-step, small-step, or abstract machines) which explicitly includes the
  definitional device of one or more subsidiary strategies. In this latter
  case we say the hybrid \emph{strategy} is defined by a definitional device
  in hybrid \emph{style}.  For instance, a hybrid strategy's big-step
  evaluator is in hybrid style when it calls a subsidiary strategy's big-step
  evaluator. As explained in \cite{GPN14}, and as we elaborate in
  Section~\ref{sec:uniform-vs-hybrid}, the distinction is essential because,
  depending on the type of device, it may be possible to write uniform-style
  devices for hybrid strategies by hiding the subsidiary devices, or to write
  hybrid-style devices for uniform strategies by concocting spurious
  subsidiary devices.

  As we detail in Sections~\ref{sec:hybrid-normal-order}
  and~\ref{sec:hybrid-applicative-order}, the uniform/hybrid notion was
  introduced by \cite{Ses02} in the context of big-step evaluators in natural
  semantics, to respectively indicate the in/dependence on subsidiary big-step
  evaluators---what we call \emph{style}.  We relied on the style concept in
  \cite{GNG10,GPN13,GNMN13,GNS14} and distinguished it from the
  style-independent property of a strategy in \cite{GPN14}, redundantly
  referring to a uniform/hybrid strategy as `uniform/hybrid in its
  nature'. Hereafter, we use `style' for definitional devices and use
  `uniform/hybrid' by itself for a property of a strategy.

  In \cite{GPN14} we formalised uniform/hybrid strategy using small-step
  context-based reduction semantics \cite{Fel87,FH92,FFF09}. Since we wish to
  improve on \cite{GPN14}, in Section~\ref{sec:uniform-vs-hybrid} we switch
  temporarily to that small-step definitional device and arrive at an improved
  definition from illustrative examples, rather than define the concept first
  and then illustrate it with examples.

  We introduce in this paper the concept of a \emph{balanced} (opposite,
  \emph{unbalanced}) hybrid strategy (Definition~\ref{def:balanced}). This
  concept is required for the equivalence proof by LWF. A hybrid strategy
  always uses a subsidiary to evaluate the operator $M$ in applications $MN$
  to obtain a redex $(\lambda x.B)N$. A balanced hybrid is either a non-strict
  hybrid (does not evaluate $N$) or is a strict hybrid that uses the
  subsidiary to evaluate $N$.

\item \label{C:struct} \textbf{Structuring of the strategy space.} We
  structure the strategy space using the uniform/hybrid dichotomy as an
  orthogonal criterion, and structure the uniform and hybrid subspaces
  according to the weakness, strictness, and headness properties of strategies
  considered as a triple.

  The uniform strategy space makes up a lattice of eight strategies we call
  \emph{Gibbons's Beta Cube}.\footnote{We name the cube with permission after
    Jeremy Gibbons who suggested to us the use of booleans to construct a cube
    lattice (Section~\ref{sec:sub:cube}). Other unrelated named cubes in the
    lambda calculus literature are J.\,J.\,L{\'e}vy's cube of residuals
    \cite{Lev78} and H.\,Barendregt's `Lambda Cube' of typed lambda calculi
    \cite{Bar91}.}  The cube contains the well-known uniform strategies
  (call-by-value, call-by-name, applicative order, head spine) and four novel
  strategies. We introduce the well-formedness provisos that characterise
  uniform eval-apply evaluators and show that they define uniform
  strategies. We also introduce a systematic notation for encoding uniform
  evaluators in a triple, and study the property of absorption among uniform
  strategies.

  The hybrid space is larger. We introduce the well-formedness provisos that
  characterise hybrid eval-apply evaluators and show they are required to
  define hybrid strategies.  The hybrid evaluators rely on \emph{one} uniform
  subsidiary evaluator to satisfy the desired properties of the strategy
  (completeness, spineness, etc.). In particular, the subsidiary is always
  used to evaluate operators. The hybrid evaluator evaluates more redexes than
  the subsidiary evaluator. We introduce a systematic notation for encoding
  hybrid evaluators by means of two triples: the Beta Cube triple of the
  uniform evaluator used as subsidiary by the hybrid evaluator, and the triple
  encoding the hybrid's weakness, strictness, and headness. We use the
  notation to introduce novel strategies and discuss \mbox{(non-)}equivalences
  among hybrid evaluators.

\item \label{C:equiv} \textbf{Equivalence proof by LWF.} Finally, we present a
  proof of equivalence between eval-readback evaluators and \emph{balanced
    hybrid} eval-apply evaluators. The proof assumes the provisos for both
  styles and proceeds by application of the LWF transformation. We start from
  a \emph{plain but arbitrary} eval-readback evaluator
  \texttt{rb}\,$\circ$\,\texttt{ev} where \texttt{ev} is a uniform evaluator
  defined as a fixed-point of the generic eval-apply evaluator, and
  \texttt{rb} is defined as a fixed-point of the generic readback
  evaluator. Some parameters of the generic evaluators are type-checked but
  undefined and thus \emph{arbitrary}. The LWF transformation fuses
  \texttt{rb} and \texttt{ev} into the plain but arbitrary balanced hybrid
  eval-apply evaluator \texttt{hy} that relies on a uniform subsidiary
  evaluator \texttt{su}, where both \texttt{hy} and \texttt{su} are
  fixed-points of the generic eval-apply evaluator. The type-checked but
  undefined parameters to the generic evaluators have their substitutable
  values constrained by the provisos and the LWF steps. The possible
  substitutions are collected in equations. In particular, some of the
  equations state that \texttt{su} = \texttt{ev}. When the equations hold then
  one-step equivalence holds and thus
  \texttt{rb}\,$\circ$\,\texttt{ev}\,$=$\,\texttt{hy}.
  \begin{center}\small
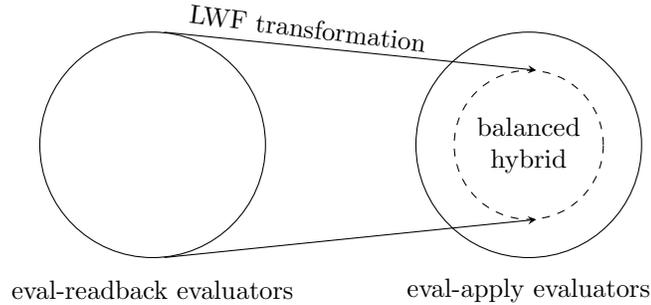
    \begin{tikzpicture}
    \node[draw,circle,minimum size=3cm] (left) {};
    \node[draw,circle,xshift=5cm,minimum size=3cm] (right) {};
    \node[draw,dashed,circle,xshift=5cm,minimum size=1cm, text width=1.5cm,
    align=center]
          (small) {balanced hybrid};
    \node[below of=left,node distance=55] {eval-readback evaluators};
    \node[below of=right,node distance=55] {eval-apply evaluators};
    \draw[-stealth]
    (tangent cs:node=left,point={(small.south)},solution=2) --
    (tangent cs:node=small,point={(left.south)});
    \draw[-stealth]
    (tangent cs:node=left,point={(small.north)},solution=1) --
    (tangent cs:node=small,point={(left.north)},solution=2)
    node[pos=0.38,above,rotate=-7.1] {{\small LWF transformation}};
    \end{tikzpicture}
  \end{center}
  An expected corollary of the proof is that \texttt{hy} absorbs
  \texttt{su}/\texttt{ev}, that is,
  \texttt{hy}\,$\circ$\,\texttt{su}\,$=$\,\texttt{hy}. Absorption can be used
  to fuse an evaluator composition or to defuse an evaluator into a
  composition.

  The LWF proof does not work for \emph{unbalanced} hybrid strategies. These
  can be defined in eval-apply style but not in an LWF-equivalent
  eval-readback style. Recall from contribution \ref{C:hybrid} above that an
  unbalanced hybrid \texttt{hy} is such that its subsidiary \texttt{su}
  evaluates applications that evaluate to a redex more than \texttt{hy}.
  Thus, \texttt{hy} cannot absorb \texttt{su}. It might be possible to obtain
  an equivalence between an eval-readback evaluator and an unbalanced hybrid
  eval-apply evaluator, but LWF cannot find that equivalence.

  For a few strategies the equivalence by LWF holds \emph{modulo commuting
    redexes} due to the particular steps of LWF. Commuting (\ie\
  non-overlapping) redexes occur in the operands $N_i$ of a neutral $x
  N_1\cdots N_n$. An equivalence modulo commuting redexes means the strategies
  are big-step equivalent but differ on the point of divergence, \ie~one
  strategy evaluates $N_j$ first and diverges and the other evaluates $N_k$
  first and diverges, with $j < k$.

\item \textbf{Summary and discussion of (non-)equivalences.} We provide an
  extensive summary of strategies, evaluators, and equivalences using the
  systematic notation, and also discuss the relevance of our contributions in
  the wider context of correspondences among operational semantics devices
  (Section~\ref{sec:conclusions}).
\end{enumerate}
These contributions are listed in increasing order of importance, but all are
necessary and interdependent, and justify the content and length of the
paper. The LWF proof goes through thanks to the generic evaluators and natural
semantics templates. The latter are obtained after structuring the strategy
space according to the improved definition of uniform/hybrid strategy. And the
strategy space is presented in a survey which justifies the importance of each
strategy and shows the scope of the equivalence proof. The LWF proof confirms
the structuring is adequate. We also provide an extensive summary and discuss
the relevance of the results in the wider context of operational semantics.

The paper is divided into four main sections (Sections~\ref{sec:eval-apply} to
\ref{sec:hybrid-evalreadback}) which can be read at separate times before the
conclusions, and the related and future work sections. The survey
(Section~\ref{sec:strategies}) can be read cursorily and its details fetched
on demand when reading about the structuring
(Section~\ref{sec:regimentation}).  The equivalence proof by LWF
(Section~\ref{sec:hybrid-evalreadback}) can be followed with no previous
knowledge of LWF, but we provide an overview of LWF in Appendix~\ref{app:LWF}.
The code has been type-checked and tested, and can be obtained on request.

We have addressed the paper to non-expert readers which may be unfamiliar with
some background material and thus we are explanatory in the text (or in the
appendices when appropriate) for information or reminder. Expert readers can
jump directly to the sections of their interest, with particular attention to
Sections~\ref{sec:regimentation} to \ref{sec:conclusions}.

\section{Technical preliminaries}
\label{sec:prelim:lambda}
We overview the pure lambda calculus notation and concepts required for this
paper in Appendix~\ref{app:prelim:lambda}, namely, syntax and meta-notation of
terms, reduction, evaluation sequence, evaluation strategy, completeness, and
`completeness for' a final result other than normal form. For the most part we
adhere to the standard references \cite{CF58,Bar84,HS08}. We use the notation
for capture-avoiding substitution of \cite{Plo75} (and justify why in
Appendix~\ref{app:prelim:lambda}). We use grammars in Extended Backus-Naur
Form to define sets of terms.  Figure~\ref{fig:lam-sets} lists the main sets
of terms we use in this paper, which play a special role in evaluation. They
are proper subsets of the set $\Lambda$ of pure lambda calculus terms.

\begin{figure}[htb]
  \begin{tabular}[t]{lllcll}
    $\Lambda$ & ::= & \multicolumn{3}{l}{$\Var\ |\ \lambda \Var.\Lambda\ |\
       \Lambda \; \Lambda$}
       & Terms over a set of variables $\Var$. \\
    $\Neu$ & ::= & \multicolumn{3}{l}{$\Var \ \Lambda \ \{\Lambda\}^*$}
       & $\NT{Neu}$trals: applications $x\,N_1\cdots N_n$ where $n\geq1$. \\
    $\NF$ & ::= &  $\lambda \Var.\NF$ & $|$ & $\Var \; \{\NF\}^*$
       & $\NT{N}$ormal $\NT{F}$orms. \\
    $\WNF$ & ::= & $\lambda \Var.\Lambda$ & $|$ & $\Var \; \{\WNF\}^*$
       & $\NT{W}$eak $\NT{N}$ormal $\NT{F}$orms. \\
    $\HNF$ & ::= & $\lambda \Var.\HNF$ & $|$ & $\Var \; \{\Lambda\}^*$
       & $\NT{H}$ead $\NT{N}$ormal $\NT{F}$orms. \\
    $\WHNF$ & ::= & $\lambda \Var.\Lambda$ & $|$ & $\Var \; \{\Lambda\}^*$
       & $\NT{W}$eak $\NT{H}$ead $\NT{N}$ormal $\NT{F}$orms. \\
    $\NT{VHNF}$ & ::= & $\lambda \Var.\NT{VHNF}$ & $|$ & $\Var \; \{\WNF\}^*$
       & $\NT{V}$alue $\NT{H}$ead $\NT{N}$ormal $\NT{F}$orms.
  \end{tabular}
  \caption{Main sets of terms in Extended BNF where $\{\NT{X}\}^*$ denotes
    zero or more occurrences of non-terminal $\NT{X}$, that is, $\NT{X}_1
    \cdots \NT{X}_n$ with $n\geq0$.}
\label{fig:lam-sets}
\end{figure}

Appendix~\ref{app:LWF} provides an overview of LWF.
Appendix~\ref{app:haskell-semantics} provides an overview of Haskell's
operational semantics in relation to
LWF. Appendix~\ref{app:completeness-leftmost-spine} explains why strict
strategies are incomplete and when non-strict strategies are complete,
particularly the leftmost and spine strategies discussed in the survey. For
the moment we assume the intuitive definition of uniform/hybrid strategy and
of un/balanced hybrid strategy provided in contribution \ref{C:hybrid}.
(Alternatively, the reader may assume the formal definition of uniform/hybrid
strategy in \cite{GPN14}.)  We also assume that all the uniform-style and
hybrid-style eval-apply evaluators in the survey respectively define uniform
and hybrid strategies. The validity of this assumption will be proven in
Section~\ref{sec:regimentation} where we give the improved formal definition
of uniform/hybrid strategy. The definition of one-step equivalence `modulo
commuting redexes' has been introduced in contribution \ref{C:equiv}.

\section{Call-by-value policy, strategies, and calculi}
\label{sec:prelim:cbv}
To discuss an evaluator style we must first pick some strategy. In this
section we define a strategy in words to avoid the circularity of using an
evaluator to define the strategy.

\paragraph{Definition of call-by-\texorpdfstring{$\WNF$}{WNF}}
We pick the `call-by-weak-normal-form' strategy that evaluates terms to weak
normal form ($\WNF$ in Figure~\ref{fig:lam-sets}) by choosing the
leftmost-innermost redex not inside an abstraction. Recall that $\WNF$
consists of arbitrary (unevaluated) abstractions, variables, and neutrals with
operands in $\WNF$.

The strategy is weak (does not evaluate abstractions), strict (evaluates
operands of redexes to $\WNF$), and non-head (evaluates operands of neutrals
to $\WNF$).  Like all strategies, it is an identity on variables. It is also
an identity on abstractions because it is weak. And it evaluates applications
$MN$ by first evaluating $M$:
\begin{itemize}
\item If $M$ evaluates to an abstraction $\lambda x.B$, then the application
  $(\lambda x.B)N$ is a redex. The strategy evaluates the operand $N$ to $N'$
  and then evaluates the contractum $\cas{N'}{x}{B}$.

\item If $M$ does not evaluate to an abstraction, then it either evaluates to
  a variable $x$ or to a neutral $M'$. In both cases, the application ($xN$ or
  $M'N$) is a neutral and the strategy evaluates the operand $N'$ as if the
  head variable (either $x$ or the head variable of $M'$) were a strict
  data-type value constructor.
\end{itemize}
The strategy is incomplete for $\WNF$ (and therefore for $\NF$) because it is
strict (Appendix~\ref{app:completeness-leftmost-spine}). For example, it
diverges on $(\lambda x.y)\OMEGA$ whereas other strategies converge on $y$
which is a $\WNF$ (and also a $\NF$).

Note that the strategy contracts redexes when the operand is a neutral in
$\WNF$.  For example, given a redex $(\lambda x.B) (z N)$, if the evaluation
of $N$ converges on $N'$, then the strategy evaluates the contractum $\cas{(z
  N')}{x}{B}$, in accordance with the beta-rule of the pure lambda calculus
which is unrestricted about what can be the operand of a redex.

\paragraph{Call-by-$\WNF$ and the call-by-value policy}
The name `call-by-value' refers to the parameter-passing \emph{policy} in
programming languages for applications of functions and value constructors to
operands. The policy decrees that functions are not evaluated, and that
operands are evaluated to a `value' whose definition depends on the
programming language. The value is then passed to the function or attached to
the value constructor. In the pure lambda calculus the notion of `value'
cannot be the expected one of normal form because the abstractions within the
operand would be evaluated to normal form, but abstractions (functions) must
not be evaluated. (The strategy that embodies call-by-normal-form is
`applicative order' discussed in the survey.)

In the pure lambda calculus the call-by-value strategy is call-by-$\WNF$ which
is weak, strict, and non-head. However, this strategy does not uphold
\emph{operational equivalence} when terms are open (concretely, when redexes
have neutrals in $\WNF$ as operands) \emph{and also} divergent. Operational
equivalence is also called `contextual equivalence' because it is defined
conveniently using closing contexts, namely, terms with a placeholder hole for
placing another term and closing all its free variables (see
Section~\ref{sec:context-based-red-semantics} for a precise definition of
context).  Two terms $M$ and $N$ are operationally equivalent when in any
closing context $\ctx$ then $\ctx(M)$ and $\ctx(N)$ both converge to the same
result or both diverge. Operational equivalence is expected when one term
evaluates to the other. However, take the redex $R\equiv(\lambda x.\lambda
y.y)(z z)$, which is an instantiation of the aforementioned redex form
$(\lambda x.B)(z N)$.  Call-by-$\WNF$ evaluates $R$ to the normal form
$\lambda y.y$ because the neutral $(z z)$ is a $\WNF$. But $R$ and $\lambda
y.y$ are not operationally equivalent under call-by-$\WNF$: with the context
$\ctx = (\lambda z.\hole)(\lambda x.x x)$, the call-by-$\WNF$ evaluation of
$\ctx(R)$ diverges whereas the evaluation of $\ctx(\lambda y.y)$
converges.\footnote{$\ctx(R) = (\lambda z.(\lambda x.\lambda y.y)(z
  z))(\lambda x.x x)$ evaluates to $(\lambda x.\lambda y.y)\OMEGA$ which
  diverges. $\ctx(\lambda y.y) = (\lambda z.\lambda y. y)(\lambda x.x x)$
  converges to $\lambda y.y$.}

A statement of the problem, its explanation and a solution, were notoriously
introduced in \cite{Plo75}.  The explanation is that the neutral must be
\emph{given the opportunity to diverge before contraction}, as is the case in
$\ctx(R)$. The solution is to block the contraction in case the head variable
might be eventually substituted, and thus to restrict the notion of `value' to
variables and abstractions so that $R$ is no longer a redex.

The problem, however, is not with the strategy itself but with the pure lambda
calculus and its beta-reduction relation on which call-by-$\WNF$ is based. The
blocking must occur in the beta rule for a beta-reduction with blocking to be
confluent. The confluent and blocking `lambda-value' calculus was proposed
where evaluation to a value upholds operational equivalence
\cite[pp.\,135--136,\,142--145]{Plo75}. The lambda-value calculus is the basis
of other improved call-by-value calculi
\cite{FF86,Mog91,Wad03,CH00,KMRDR21,AG22}. Readers interested in further
technical observations about lambda-value are referred to
Appendix~\ref{app:lambda-value}.

\paragraph{Call-by-$\WNF$ as a call-by-value strategy}
In this paper we refer to call-by-$\WNF$ as `call-by-value' for a formal and a
pragmatic reason. The formal reason is that it is the call-by-value strategy
of the \emph{pure} lambda calculus.  To study blocking strategies with open
and divergent terms requires a lambda-value or related calculus with blocking
in the beta rule. Proving evaluator equivalences for such strict strategies
(and for the non-strict strategies that can be simulated in that calculi
\cite[Sec.\,6]{Plo75}) is future work
(Section~\ref{sec:future-work}).

The pragmatic reason is that operational equivalence holds for \emph{closed}
terms under weak evaluation and for \emph{convergent} terms under weak or full
evaluation. As stated in Section~\ref{sec:the-setting}, studying strategies
and evaluators for arbitrary pure terms allows us to then restrict to closed
and/or convergent subsets to include strategies and evaluators in programming
and proof-assistant settings.

With closed terms and weak evaluation, neutrals are given the opportunity to
diverge because evaluation does not go under lambda and every neutral is given
the binding for its free head variable. (In applied calculi, applications
$c\,N_1 \cdots N_n$ with value constructor $c$ are permissible operands.) Many
so-called `call-by-value' evaluators and abstract machines in the literature
assume closed input terms: the \SECD\ abstract machine \cite{Lan64}
\cite{Lan65} \cite[Chp.\,10]{FH88} (its evaluator defined in
\cite[p.\,316]{Lan64}), \cite{FF86} \cite[Chp.\,9]{FH88}
\cite[Chp.\,13]{Rea89} \cite[Chps.\,11--12]{Mac90}
\cite[pp.\,421,\,427]{Ses02} \cite[p.\,390]{Pau96} (named \texttt{eval}), and
\cite[p.\,17]{PR99} (named `inner machine'). These evaluators implement
call-by-$\WNF$ for closed input terms. Even the lambda-value big-step
evaluator `$\mathrm{eval}_V$' of \cite[p.\,130]{Plo75} assumes closed input
terms and implements call-by-$\WNF$ for closed input terms
(Appendix~\ref{app:lambda-value}).

With convergent terms, there is no opportunity to diverge, terms can be open
or closed, and evaluation under lambda and full-reduction are unproblematic.
Convergent pure terms are obtainable by strong type-checking and type-erasure
\cite{Pie02,GL02}.\footnote{Convergent terms are not exclusively
  strongly-normalising terms. A simply-typed lambda calculus extended with
  terminating structural recursion primitives has convergent terms because
  self-application does not type-check \cite{Pie02}.} Many programming
languages have weak evaluation and type systems based on System~F
\cite{Car97,Pie02}. However, these type systems are sound for full-reduction
\cite[p.\,3]{SR15}. Call-by-$\WNF$ and its full-reducing relatives we discuss
in the survey (Section~\ref{sec:strategies}) work verbatim with convergent
terms and uphold operational equivalence. In
Sections~\ref{sec:strict-normalisation} and \ref{sec:eval-readback} we discuss
the case in point of the weak-reducing (and referred to as `call-by-value')
evaluator $\GL{V}$ \cite[p.\,236]{GL02} that implements call-by-$\WNF$ for
convergent terms. We also study the full-reducing call-by-value eval-readback
evaluator named `strong reduction' (denoted by $\GL{N}$) \cite[p.\,237]{GL02}
that uses $\GL{V}$ for its eval stage.\footnote{To be precise, $\GL{V}$
  evaluates operands of neutrals right-to-left, so the equivalence with
  call-by-$\WNF$ is modulo commuting redexes. The result delivered by $\GL{V}$
  is a $\WNF$, not a $\HNF$ as stated in \cite[p.\,236]{GL02}. Both $\GL{V}$
  and $\GL{N}$ are proven correct for strongly-normalising terms
  \cite[Lems.\,5\,\&\,6,\,pp.\,237--9]{GL02}.}

\section{The eval-apply evaluator style}
\label{sec:eval-apply}
An evaluator in eval-apply style is defined by an \texttt{eval} function that
recursively evaluates subterms. In its traditional definition, \texttt{eval}
relies on a mutually-recursive local \texttt{apply} function for the specific
evaluation of application terms. In modern functional programming languages
the local \texttt{apply} can be obviated within \texttt{eval} using
\texttt{let}-clauses and pattern-matching expressions.

As explained in Section~\ref{sec:the-setting}, we use a deep embedding of
lambda-calculus terms into Haskell where lambda-calculus variables are
represented by strings.
\begin{code}
  data Term = Var String | Lam String Term | App Term Term
\end{code}
We implement capture-avoiding substitution directly (and thus inefficiently)
according to the textbook definition \cite{HS08}. The implementation details
are unimportant.\footnote{For interested readers, our string variables have
  trailing numbers (\eg\ \texttt{"x1"}, \texttt{"y5"}, etc.). Function
  \texttt{subst} pattern-matches on its third argument and obtains fresh
  variables by incrementing the largest trailing number of all the string
  variables of all its inputs.}
\begin{code}
  subst :: Term -> String -> Term -> Term
  subst n x b = ...  -- substitute n for x in b.
\end{code}
The following evaluator, named \texttt{bv}, implements in traditional
eval-apply style the call-by-value (call-by-$\WNF$) strategy described in
words in Section~\ref{sec:prelim:cbv}.
\begin{code}
  bv :: Term -> Term
  bv v@(Var _)   = v
  bv l@(Lam _ _) = l
  bv (App m n)   = apply (bv m) (bv n)
    where apply :: Term -> Term -> Term
          apply (Lam x b) n' = bv (subst n' x b)
          apply m'        n' = App m' n'
\end{code}
The evaluator pattern-matches on the input term. For variables (first clause)
and abstractions (second clause) it behaves as an identity. We use Haskell's
`as-patterns' to simplify the code. For applications (third clause), the local
\texttt{apply} function is called with recursive calls to \texttt{bv} on the
operator and the operand. The local \texttt{apply} pattern-matches on the
operator to evaluate it. Its first clause evaluates the redex and its second
clause evaluates the neutral. The local \texttt{apply} can be obviated within
\texttt{eval} using Haskell's \texttt{let} and \texttt{case} expressions:
\begin{code}
  bv :: Term -> Term
  bv v@(Var _)   = v
  bv l@(Lam _ _) = l
  bv (App m n)   = let m' = bv m in case m' of
                     (Lam x b) -> let n' = bv n in bv (subst n' x b)
                     _         -> let n' = bv n in App m' n'
\end{code}
The two \texttt{let} expressions for the operand \texttt{n'} could be
simplified or omitted (substituting \texttt{bv}~\texttt{n} for \texttt{n'})
but they will prove handy when we compare evaluators and generalise in
Section~\ref{sec:regimentation}.

Readers familiar with Haskell will notice that \texttt{bv} actually does
\emph{not} implement call-by-value due to the classic `semantics-preservation'
problem of evaluators \cite{Rey72}: the evaluation order of function and
value-constructor application (non-strict call-by-need) in the implementation
language (Haskell) does not match the intended evaluation order of function
and term application (strict) for the object language (pure lambda calculus).
Concretely, the operand \texttt{n'} (bound to the call \texttt{bv}~\texttt{n}
in the \texttt{let} expression) is passed \emph{unevaluated} by Haskell to
\texttt{subst}, which evaluates \texttt{n'} only if \texttt{x} occurs in
\texttt{b}.  Thus, for example, if we pass to \texttt{bv} the input term
$(\lambda x.z)\OMEGA$ as an expression of type \texttt{Term}, then the
evaluator wrongly delivers $z$ (as \texttt{Var\,"z"}).

In a strict implementation language, such as Standard ML or OCaml, where
function and value-constructor application are strict, the \texttt{n'}
parameter would be evaluated before being passed to \texttt{subst}. Such
strict evaluation would also work fine with evaluators that implement
non-strict strategies such as, say, call-by-name
(Section~\ref{sec:call-by-name}), because in that case \texttt{n'} would be
bound to a term (a data-type value) rather than to a call expression like
\texttt{bv}~\texttt{n}.

The problem with the strictness of \texttt{subst} is compounded by the
\emph{partially-defined} values of type \texttt{Term}. This type represents
terms of $\Lambda$ with the addition of the totally-undefined value $\bot$ and
of partially-defined values such as \texttt{Lam}~\texttt{"x"}~$\bot$,
\texttt{App}~$\bot~\bot$, etc. The totally-undefined value is necessary
because we want to consider non-terminating strategies. It is the case that
every evaluator \texttt{ev} is strict, that is, \texttt{ev}~$\bot$ = $\bot$,
because it pattern-matches on its input term. It is also the case that an
evaluator may not terminate for an input term \texttt{t}, namely,
\texttt{ev}~\texttt{t} = $\bot$. However, partially-defined values have no
counterpart in $\Lambda$. We must avoid \texttt{ev}~\texttt{t} = \texttt{p}
for a partially-defined term \texttt{p}. This can be done using strict
function application during evaluation.

\subsection{Monadic style eval-apply}
\label{sec:eval-apply-monadic}
Strict function application can be simulated in Haskell using the principled
and flexible monadic style. We use the Haskell 2010 standard which has a
well-defined static and operational semantics
(Appendix~\ref{app:haskell-semantics}). Evaluators are one of the motivating
applications of monads \cite{Wad92a,Wad92b}, and monadic evaluators can be
program-transformed to other operational semantics devices
\cite{ADM05}. Monads are provided in the Haskell 2010 library by means of the
\texttt{Monad} type class and its datatype instances. Type classes have a
trivial dictionary operational semantics \cite{HHPJW96}. The \texttt{Monad}
type class provides operators for monadic identity (\verb!return!), monadic
function application (\verb!=<<!), and monadic function composition
(\verb!<=<!), called Kleisli composition. The first two operators are defined
by programmers for the particular monad instance. Monadic function composition
is defined in the library in terms of monadic function application. For
comparison, we show in Figure~\ref{fig:monadic-ops} the type signatures
alongside ordinary identity, function application, and function composition.

\begin{figure}[htb]
  \centering
  \begin{code}
  -- Identity
  id     ::            a ->   a
  return :: Monad m => a -> m a

  -- Application
  (\$)   ::            (a ->   b) ->   a ->   b
  (=<<) :: Monad m => (a -> m b) -> m a -> m b

  -- Composition
  (.)   ::            (b ->   c) -> (a ->   b) -> (a ->   c)
  (<=<) :: Monad m => (b -> m c) -> (a -> m b) -> (a -> m c)

  (g . f) x   = let y = f x
                in g y

  (g <=< f) x = do y <- f x
                   g y
 \end{code}
 \caption{Type signatures of monadic operators alongside their ordinary
   counterparts. Monadic function composition is defined in terms of monadic
   function application. The \texttt{do}-notation is syntactic sugar for the
   latter.}
  \label{fig:monadic-ops}
\end{figure}

Strict function application can be simulated by monadic application
\cite[Sec.\,3.2]{Wad92b} when $\mathtt{return} \: \bot = \bot$ and $\mathtt{f}
\: \verb!=<<! \: \bot = \bot$ for function \texttt{f}, which is the case for
some strict monads.\footnote{The statement holds in Haskell 2010. However, at
  the time of writing, the latest \texttt{Monad} class is a subclass of
  \texttt{Functor} and \texttt{Applicative} which are not in the 2010
  standard, and the \texttt{Functor} laws use non-strict composition
  \cite[Sec.\,3.2]{Wad92b}. To use Haskell in this setting we would have to
  define an alternative strict monadic composition.} In particular, we will
use the \texttt{IO} monad because we are interested in printing the results of
evaluation. (We will use monadic composition later on when we discuss
eval-readback evaluators in Section~\ref{sec:eval-readback}).

The monadic style has several benefits. First, we can keep Haskell which is
our preferred choice of implementation language for, among other things, its
convenient syntax for writing the generic evaluators and their fixed-points
(Sections~\ref{sec:eval-apply-template} and~\ref{sec:eval-readback-template}).
Second, we can keep \texttt{subst} non-monadic. Third, we can toggle between
strict and non-strict Haskell evaluation by instantiating the evaluator's type
signature with strict and non-strict monad instances. The non-strict Haskell
evaluation of strict strategies like \texttt{bv} can be helpful in iterative
tracing up to particular depths \cite{Ses02} to locate divergence points,
known in vernacular as `chasing bottoms' \cite{DJ04}. The parametrisation on a
monad type parameter requires the `arbitrary rank polymorphism' type-checking
extension of the Glorious Haskell Compiler (GHC) which falls outside the
Haskell 2010 standard. Nevertheless, the extension has no run-time effect, as
it simply allows GHC to overcome the limitations of the default type-inference
algorithm by letting programmers provide type signature annotations for
perfectly type-checkable code. Hereafter, we write evaluators in monadic style
with a monad type parameter, but the parameter is optional and immaterial.

The transcription of the naive \texttt{bv} evaluator to monadic style is
immediate with Haskell's \texttt{do}-notation:
\begin{code}
  bv :: Monad m => Term -> m  Term
  -- may assume    Term -> IO Term
  bv v@(Var _)   = return v
  bv l@(Lam _ _) = return l
  bv (App m n)   = do m' <- bv m
                      case m' of
                        (Lam x b) -> do n' <- bv n
                                        bv (subst n' x b)
                        _         -> do n' <- bv n
                                        return (App m' n')
\end{code}
The type of \texttt{bv} is bounded-quantified on the monad type parameter
\texttt{m}, which can be instantiated to a strict or non-strict monad instance
using a type annotation when executing the evaluator, \eg\ \texttt{(bv}
\texttt{::} \texttt{IO} \texttt{Term)} \texttt{t}. Readers uninterested in the
parameter may assume \texttt{m} already instantiated to the \texttt{IO} monad,
as shown in the comment. Notice that \texttt{m} also occurs as a value
variable (operator \texttt{m}) on the third clause for applications. This is
usual Haskell coding style, as type and value variables have independent
scopes.

The use of Haskell, monadic style, and strict monad instances are not required
for the proof by LWF. Haskell 2010 is used as a vehicle for communication and,
at any rate, we could treat all Haskell code as syntactic sugar for a strict
implementation language. LWF is a simple, syntactic, and non-invertible
transformation that only requires evaluators to be strict on their first
parameter (Appendix~\ref{app:LWF}), which is the case for every evaluator
because the first parameter is the input term. LWF does not involve equational
reasoning, although monadic equational reasoning in Haskell is perfectly
possible, \eg\ \cite{GH11}. A LWF proof can also be obtained for non-monadic
or monadic code in a strict implementation language, as demonstrated in
\cite{GPN14} for two plain evaluators.

\subsection{Natural semantics style eval-apply}
\label{sec:eval-apply-natural-sem}
A natural semantics is a proof-theory of syntax-directed and deterministic
(exhaustive and non-overlapping) inference rules \cite{Kah87}. Every rule has
zero or more antecedent formulas above a line which are ordered from left to
right, and only one consequent formula shown below the line.  Every rule
states that the formula below is the case if all formulas above are the
case. A rule with no formulas above is an axiom. A formula holds iff a
stacking up of rules (an evaluation tree) ends in axioms.

An eval-apply evaluator in natural semantics consists of a set of inference
rules that specify the evaluator's behaviour recursively on the subterms. The
call-by-value evaluator $\bv$ below is defined in natural semantics by four
rules: one for evaluating \textsc{var}iables, one for evaluating
\textsc{abs}tractions, one for evaluating applications which unfold to a redex
that is \textsc{con}tracted, and one for evaluating applications which unfold
to \textsc{neu}trals. %
{\small\begin{mathpar}%
    \inferrule*[left=var] %
    { } %
    {\bv(x) = x} %
    \and %
    \inferrule*[left=abs] %
    { } %
    {\bv(\lambda x.B) = \lambda x.B} %
    \and %
    \inferrule*[left=con] %
    {\bv(M) = \lambda x.B \quad \bv(N) = N' \quad \bv(\cas{N'}{x}{B}) = B'} %
    {\bv(MN) = B'} %
    \and %
    \inferrule*[left=neu] %
    {\bv(M) = M' \quad M' \not\equiv \lambda x.B \quad \bv(N) = N'} {\bv(MN) =
      M'N'}%
  \end{mathpar}}%
Every formula of the form $\bv(M) = N$ reads `$N$ is the result in the finite
evaluation sequence of $M$ under $\bv$'. The formula is proven by the finite
evaluation tree above it. When $\bv(M)$ has an infinite evaluation sequence
($\bv$ diverges on $M$) then it has an infinite evaluation tree and we write
$\bv(M)=\bot$ to indicate it. (The symbol `$\bot$' is meta-notation because
$\bot\not\in\Lambda$.)

Let us explain the rules. The \textsc{var} axiom states that $\bv$ is an
identity on variables (they are normal forms). The \textsc{abs} axiom states
that $\bv$ is also an identity on abstractions ($\bv$ is weak). The
\textsc{con} rule states that the application $MN$ evaluates by $\bv$ to a
term $B'$ if $\bv$ evaluates the operator $M$ to an abstraction $\lambda x.B$,
and also evaluates the operand $N$ to a term $N'$, and also evaluates the
contractum of the redex $(\lambda x.B)N'$ to a result $B'$. The \textsc{neu}
rule states that the application $MN$ evaluates by $\bv$ to a neutral $M'N'$
if $\bv$ evaluates the operator $M$ to a non-abstraction $M$' (\ie~a variable
or a neutral), and also evaluates the operand $N$ to a result $N'$.

The eval-apply natural semantics corresponds with the eval-apply evaluator in
code, where the \textsc{con} and \textsc{neu} rules are implemented by
\texttt{apply}. In rule \textsc{con}, the contractum's occurrence within the
$\bv$ function is a notational convenience. The contractum can be placed as a
premise right before the last premise that evaluates it.
{\small\begin{mathpar}%
    \inferrule*[left=con]{%
      \inferrule*{\vdots}{\bv(M) = \lambda x.B} %
      \quad \inferrule*{\vdots}{\bv(N) = N'}%
      \quad \cas{N'}{x}{B} \equiv B'' \quad \inferrule*{\vdots}{\bv(B'') =
        B'}} %
    {\bv(MN) = B'} %
  \end{mathpar}%
}%
With this modification it is easier to illustrate the connection between an
evaluation tree and an evaluation sequence. The latter can be obtained by the
in-order traversal of the \textsc{con} rules in an evaluation tree. The
in-order traversal of the \textsc{con} rule is as follows: in-order traversal
of the stacked evaluation tree (vertical dots) above the first premise and
then above the second premise, followed by the contractum premise in the
middle, followed by the in-order traversal of the stacked evaluation tree
above the last premise.  Figure~\ref{fig:example-nat-sem} illustrates with an
example.

\begin{figure}[htb]
  \centering\small%
  \flushleft{\textbf{Derivation tree:}}
  {\scriptsize\begin{mathpar}%
    \inferrule*{%
      \bv(\lambda x.\CH{I}z) \equiv \lambda x.\CH{I}z
      \quad
      \inferrule*{%
        \bv ( \CH{I} ) = \CH{I} \quad \bv (z) = z \quad
        \boxed{\cas{z}{x}{x}}^1 \equiv z \quad \bv(z) = z
      }%
      { \bv(\CH{I}z) = z }%
      \quad %
      \boxed{\cas{z}{x}{\CH{I}z}}^2 \equiv \CH{I}z \quad
      \inferrule* %
      {\ldots \boxed{\cas{z}{x}{x}}^3 \ldots}%
      { \bv( \CH{I}z ) = z} }%
    { \bv((\lambda x.\CH{I}z)(\CH{I}z)) = z }%
  \end{mathpar}}%
  \flushleft{\textbf{Evaluation sequence:}}
  \begin{displaymath}
    (\lambda x.\CH{I}z)(\underline{\CH{I}z}) %
    \rel %
    (\lambda x.\CH{I}z)(\ \boxed{\cas{z}{x}{x}}^1 ) %
    \equiv %
    \underline{(\lambda x.\CH{I}z)z} %
    \rel %
    \boxed{\cas{z}{x}{\CH{I}z}}^2 %
    \equiv %
    \underline{\CH{I}z} %
    \rel %
    \boxed{\cas{z}{x}{x}}^3 %
    \equiv %
    z
  \end{displaymath}
  \caption{Fragment of the evaluation tree of the term $(\lambda
    x.\CH{I}z)(\CH{I}z)$ where $\CH{I}\equiv \lambda x.x$. The evaluation
    sequence (with redexes underlined) is obtained by collecting the redexes
    and contracta in in-order traversal.}
  \label{fig:example-nat-sem}
\end{figure}

\section{A survey of the pure lambda calculus's strategy space}
\label{sec:strategies}
This section provides a survey of pure lambda calculus strategies, their
properties (some of which are features of their definition and serve as
classification criteria), their abstract machines, and some uses that
illustrate the properties. The strategies are defined through their canonical
eval-apply evaluators in natural semantics. For ease of reference, all the
evaluators are collected in two figures. Figure~\ref{fig:strategies1}
(page~\pageref{fig:strategies1}) shows the evaluators for the main well-known
strategies in the literature.  Figure~\ref{fig:strategies2}
(page~\pageref{fig:strategies2}) shows the evaluators for less known but
equally important strategies. The last three are novel and of particular
interest. We discuss each evaluator/strategy in a separate section. With the
exception of the novel strategies, the details discussed in the survey are
known and can be found in the cited literature.

The survey can be read cursorily and particulars fetched on demand when
reading about the structuring of the strategy space
(Section~\ref{sec:regimentation}). The role of the survey is to illustrate a
varied strategy space that is nonetheless determined by a few variability
points. The survey is thus wide in strategies but not too deep in strategy
details beyond those needed to structure the space. A summary can be found at
the end in Section~\ref{sec:summary-survey}.

All evaluators are defined in natural semantics by four rules named
\textsc{var}, \textsc{abs}, \textsc{con} and \textsc{neu} which have been
explained in Section~\ref{sec:eval-apply-natural-sem}.  We omit the
\textsc{var} rule in the figures because every evaluator is an identity on
variables. We also omit the rule names because we always show them in order
(\textsc{abs}, \textsc{con}, and \textsc{neu}) from left to right, top to
bottom. The \textsc{con} and \textsc{neu} rules are related by their first
premises. In all the \textsc{con} rules the first premise has the form $\ea(M)
= \lambda x.B$ (here $\ea$ is a generic name for \texttt{e}val-\texttt{a}pply
evaluator) which states that the operator $M$ evaluates to an abstraction. In
all the \textsc{neu} rules the first and second premises respectively have the
form $\ea(M) = M'$ and $M' \not\equiv \lambda x.B$ which state that the
operator $M$ does not evaluate to an abstraction.

The first six evaluators in Figure~\ref{fig:strategies1} and the first two in
Figure~\ref{fig:strategies2} are discussed in \cite{Ses02} which classifies
strategies according to three criteria: weakness, strictness, and kind of
final result. The present survey can be seen as an extension of that work
where we add more strategies, properties, classification criteria, and
background. In particular, we indicate whether the evaluators are in uniform
or hybrid \emph{style}, and whether they are \emph{balanced} hybrid (recall
the latter's definition from contribution \ref{C:hybrid}).  In
Section~\ref{sec:regimentation} we will confirm that all the uniform-style and
hybrid-style evaluators in the survey respectively define uniform and hybrid
strategies.

\begin{figure}[p]
  \small
  \begin{flushleft}
    \textbf{Call-by-value} ($\bv$) or call-by-$\WNF$:
  \end{flushleft}
  \vspace{4pt}
  \begin{mathpar}
    \inferrule %
    { } %
    {\bv(\lambda x.B) = \lambda x.B} %
    \and %
    \inferrule %
    {\bv(M) = \lambda x.B \quad \bv(N) = N' \quad \bv(\cas{N'}{x}{B}) = B'} %
    {\bv(MN) = B'} %
    \and %
    \inferrule %
    {\bv(M) = M' \quad M' \not\equiv \lambda x.B \quad \bv(N) = N'} %
    {\bv(MN) = M'N'} %
  \end{mathpar}

  \vspace{5pt}

  \begin{flushleft}
    \textbf{Call-by-name} ($\bn$):
  \end{flushleft}
  \vspace{4pt}
  \begin{mathpar}
    \inferrule %
    { }{\bn(\lambda x.B) = \lambda x.B} %
    \and %
    \inferrule %
    {\bn(M)  = \lambda x.B \quad \bn(\cas{N}{x}{B}) = B'} %
    {\bn(MN) = B'} %
    \and %
    \inferrule %
    {\bn(M) = M' \quad M' \not\equiv \lambda x.B} %
    {\bn(MN) = M'N} %
  \end{mathpar}

  \vspace{10pt}

  \begin{flushleft}
    \textbf{Applicative order} ($\ao$) or call-by-$\NF$:
  \end{flushleft}
  \vspace{4pt}
  \begin{mathpar}
    \inferrule %
    {\ao(B) = B'}{\ao(\lambda x.B) = \lambda x.B'} %
    \and %
    \inferrule %
    {\ao(M) = \lambda x.B \quad \ao(N) = N' \quad \ao(\cas{N'}{x}{B}) = B'} %
    {\ao(MN) = B'} %
    \and %
    \inferrule %
    {\ao(M) = M' \quad M' \not\equiv \lambda x.B \quad \ao(N) = N'} %
    {\ao(MN) = M'N'} %
  \end{mathpar}

  \vspace{4pt}

  \begin{flushleft}
    \textbf{Normal order} ($\no$):
  \end{flushleft}
  \vspace{4pt}
  \begin{mathpar}
    \inferrule %
    {\no(B) = B'}{\no(\lambda x.B) = \lambda x.B'} %
    \and %
    \inferrule %
    {\bn(M) = \lambda x.B \quad \no(\cas{N}{x}{B}) = B'} %
    {\no(MN) = B'} %
    \and %
    \inferrule %
    {\bn(M) = M' \quad M' \not\equiv \lambda x.B \quad \no(M') = M'' \quad
      \no(N) = N'} %
    {\no(MN) = M'' N'} %
  \end{mathpar}

  \vspace{8pt}

  \begin{flushleft}
    \textbf{Head reduction} ($\hr$):
  \end{flushleft}
  \vspace{4pt}
  \begin{mathpar}
    \inferrule
    {\hr(B) = B'}{\hr(\lambda x.B) = \lambda x.B'} %
    \and %
    \inferrule %
    {\bn(M) = \lambda x.B \quad \hr(\cas{N}{x}{B}) = B'} %
    {\hr(MN) = B'} %
    \and %
    \inferrule %
    {\bn(M) = M' \quad M' \not\equiv \lambda x.B \quad \hr(M') = M''} %
    {\hr(MN) = M''N} %
  \end{mathpar}

  \vspace{4pt}

  \begin{flushleft}
    \textbf{Head spine} ($\he$):
  \end{flushleft}
  \vspace{8pt}
  \begin{mathpar}
    \inferrule %
    {\he(B) = B'}{\he(\lambda x.B) = \lambda x.B'} %
    \and %
    \inferrule %
    {\he(M) = \lambda x.B \quad \he(\cas{N}{x}{B}) = B'} %
    {\he(MN) = B'} %
    \and %
    \inferrule %
    {\he(M) = M' \quad M' \not\equiv \lambda x.B} %
    {\he(MN) = M'N} %
  \end{mathpar}

  \vspace{10pt}

  \begin{flushleft}
    \textbf{Strict normalisation} ($\sn$):
  \end{flushleft}
  \vspace{6pt}
  \begin{mathpar}
    \inferrule %
    {\sn(B) = B'} %
    {\sn(\lambda x.B) = \lambda x.B'} %
    \and %
    \inferrule %
    {\bv(M) = \lambda x.B \quad \bv(N) = N' \quad \sn(\cas{N'}{x}{B}) = B'} %
    {\sn(MN) = B'} %
    \and %
    \inferrule %
    {\bv(M) = M' \quad M' \not\equiv \lambda x.B \quad \sn(M') = M'' \quad
      \sn(N) = N'} %
    {\sn(MN) = M''N'} %
  \end{mathpar}
  \caption{Canonical eval-apply evaluators in natural semantics for the
    foremost strategies in the literature.}
  \label{fig:strategies1}
\end{figure}

\subsection{Call-by-value \texorpdfstring{$(\bv)$}{(bv)}}
\label{sec:call-by-value}
We have already discussed this strategy and its eval-apply evaluator $\bv$ in
Sections~\ref{sec:prelim:cbv} and~\ref{sec:eval-apply}. The strategy is
call-by-value where `value' means `weak normal form' (call-by-$\WNF$). It
chooses the leftmost-innermost redex not inside an abstraction.  It is weak,
strict, non-head, and incomplete (because it is strict,
Appendix~\ref{app:completeness-leftmost-spine}). It delivers weak normal forms
as results.  The call-by-value evaluator ($\bv$) is defined in uniform
style. The strategy is uniform.

As discussed in Section~\ref{sec:prelim:cbv}, the \SECD\ abstract machine
\cite{Lan64}\cite[Chp.\,10]{FH88} implements an applied version of
call-by-$\WNF$ when input terms are closed \cite[p\,421]{Ses02}. The CAM
abstract machine \cite{CCM84}\cite[Chp.\,13]{FH88} ultimately implements
call-by-$\WNF$ like \SECD. The ZAM abstract machine of the ZINC project
\cite{Ler91} implements the weak evaluator $\GL{V}$ mentioned in
Section~\ref{sec:prelim:cbv}, which evaluates neutral operands right-to-left
and is thus one-step equivalent to $\bv$ modulo commuting redexes. The CEK
abstract machine \cite{FF86,FH92} is related to the call-by-value strategy
defined by the `$\mathrm{eval}_V$' big-step evaluator of the lambda-value
calculus \cite{Plo75} (Appendix~\ref{app:lambda-value}). However, as pointed
out in \cite[p\,421]{Ses02}, the call-by-value strategy described in
\cite[Sec.\,2]{FH92} does not evaluate $(x y) ((\lambda z.z)v)$ with value $v$
to $(x y) v$ because it has no evaluation context for it.

The seminal framework for environment-based call-by-value abstract machines
for a calculus of closures can be found in \cite{Cur91,BD07}. In \cite{BD07},
the name `call-by-value' is reserved for the abstract machine whereas the
strategy is called `applicative-order'.  However, the name `applicative order'
is usually reserved for the full-reducing call-by-$\NF$ strategy of the pure
lambda calculus (Section~\ref{sec:applicative-order}). Traditionally,
environment-based abstract machines correspond with calculi with closures or
explicit substitutions \cite{Cur91,ACCL91}. In \cite{ABM14}, the authors use
the linear substitution calculus to vanish some machine transitions and
`distill' abstract machines from the calculus while preserving the time
complexity of the execution.

\subsection{Call-by-name \texorpdfstring{$(\bn)$}{(bn)}}
\label{sec:call-by-name}
This strategy chooses the leftmost-outermost redex not inside an
abstraction. It is weak, non-strict, head, and delivers (and is complete for)
weak head normal forms ($\WHNF$ in Figure~\ref{fig:lam-sets}). These consist
of variables and non-evaluated abstractions and neutrals.

The \textsc{abs} rule does not evaluate abstractions (weak).  The \textsc{con}
rule evaluates the operator to an abstraction (goes left) and then the operand
$N$ is passed unevaluated (non-strict `by name') to the substitution. The
\textsc{neu} rule also leaves the neutral's operand $N$ unevaluated (head), as
if the head variable were a non-strict data-type value constructor. The
call-by-name evaluator ($\bn$) is defined in uniform style. The strategy is
uniform.

The presence of rule \textsc{neu} marks the difference with call-by-name in
programming languages where open terms are disallowed. As pointed out in
\cite[p.\,425]{Ses02}, it is well-known that, with closed input terms,
call-by-name ($\bn$) is identical to those strategies and to the big-step
evaluator `$\mathrm{Eval}_N$' of \cite{Plo75}. The latter does not consider
neutrals. Interestingly, the small-step call-by-name evaluator `$\rel_N$' does
evaluate neutrals \cite[p.\,146,\,4th\,rule]{Plo75} but to a limited extent,
as also pointed out in \cite[p.\,421]{Ses02}. For example, $x \, ((\lambda
z.z)\,v)$ evaluates to $x\,v$ but $x \, y \, ((\lambda z.z)\,v)$ is a stuck
term. The `$\mathrm{Eval}_N$' and `$\rel_N$' evaluators are one-step
equivalent (proven on paper \cite[Thm.\,2,\,p.\,146]{Plo75}, not by program
transformation) only for closed input terms.

Call-by-name is the basis for the call-by-need evaluation mechanism of
non-strict functional programming languages (broadly, call-by-name with
sharing and memoisation \cite[Chp.\,11]{FH88}). The seminal framework for
environment-based call-by-name abstract machines for a calculus of closures
can be found in \cite{Cur91,BD07}. The Krivine abstract machine (KAM)
\cite{Kri07} implements call-by-name. It has been the inspiration for abstract
machines of proof-assistants that add capabilities for full-reduction on KAM.
The KN abstract machine is a well-known full-reducing improvement of KAM
(Section~\ref{sec:normal-order}).

\subsection{Applicative order \texorpdfstring{$(\ao)$}{(ao)}}
\label{sec:applicative-order}
This strategy chooses the leftmost-innermost redex (whether or not inside an
abstraction).  It is non-weak, strict, non-head, full-reducing, and
incomplete. It delivers normal forms as results ($\NF$ in
Figure~\ref{fig:lam-sets}).

The \textsc{abs} rule evaluates abstractions (non-weak). The \textsc{con} rule
evaluates the operator to an abstraction in $\NF$ (goes left), then the
operand to $\NF$ (strict), and then the contractum to $\NF$. The \textsc{neu}
rule evaluates the operator to $\NF$ (goes left), and then the operand to
$\NF$ (non-head). The applicative order evaluator ($\ao$) is defined in
uniform style. The strategy is uniform.

Because applicative order evaluates operands to normal form before
contraction, it is sometimes referred to as `call-by-value' where value means
`normal form' (call-by-$\NF$). It should not be called `call-by-value' in
applied versions of the pure lambda calculus for modelling programming
languages because functions are not evaluated and the corresponding strategy
is call-by-value as call-by-$\WNF$ (Section~\ref{sec:call-by-value}). As a
historical curiosity, applicative order is called `standard order' in
\cite[Chp.\,6]{FH88} perhaps because evaluating operands before contraction
was somehow considered standard. However, `standard' has a technical meaning
in reduction (Section~\ref{sec:normal-order}) which does not apply to
applicative order. An early abstract machine for applicative order based on a
modification of the \SECD\ abstract machine can be found in \cite{McG70}.

Applicative order illustrates that weakness and strictness are separate
classification criteria, and that the notion of `value' is determined by the
properties of the strategy that determine how operands of redexes are
evaluated.

\subsection{Normal order \texorpdfstring{$(\no)$}{(no)}}
\label{sec:normal-order}
This strategy chooses the leftmost-outermost redex. It is non-weak,
non-strict, non-head, full-reducing, and delivers (and is complete for) normal
forms ($\NF$). In the lambda-calculus literature, a leftmost-outermost redex
is simply called `leftmost' and normal order is called `leftmost reduction'
\cite{CF58,HS08,Bar84}. Note that call-by-name ($\bn$) is another leftmost
strategy, but does not evaluate leftmost redexes inside abstractions.

The \textsc{abs} rule evaluates abstractions (non-weak). The \textsc{con} rule
is more sophisticated. Given an application $MN$, normal order evaluates the
operator $M$ first (goes left), but when the operator evaluates to an
abstraction $\lambda x.B$ then the application $(\lambda x.B)N$ is the
leftmost-outermost redex to be contracted. Normal order must not evaluate the
abstraction body $B$. That would be going innermost, like applicative order.
Normal order contracts the outermost redex $(\lambda x.B)N$ passing $N$
unevaluated to the substitution (non-strict). To avoid the full evaluation of
abstractions in operator position, normal order relies on the weak
call-by-name strategy ($\bn$) as subsidiary. Finally, the \textsc{neu} rule
evaluates the operator further to $\NF$ and also evaluates the operand
(non-head) further to $\NF$, as the aim is to deliver a final $\NF$.

The normal order evaluator ($\no$) is hybrid in style, and it is balanced, as
it evaluates applications $MN$ in rule \textsc{con} before contraction as much
as its subsidiary call-by-name ($\bn$). The strategy is balanced hybrid.

Normal order is the `standard' full-reducing and complete strategy of the pure
lambda calculus, both in the colloquial sense (it is the \emph{de facto}
strategy used for full evaluation in the calculus) and in the technical sense
(roughly, it does not evaluate a redex to the left of the one just contracted
\cite{Bar84,HS08}). Due to the possibility of non-termination, normal order is
a semi-decision procedure for the `has-a-normal-form' property of terms.
Because normal order is non-weak, it is inaccurate to refer to it as
call-by-name in applied versions of the pure lambda calculus for modelling
programming languages where functions are unevaluated. To our knowledge, no
programming language has normal order as its evaluation strategy. It would
evaluate functions non-strictly and value constructors strictly (no
partially-defined terms). The opposite, strict function evaluation with
non-strict value constructors was implemented by the Hope programming language
\cite{FH88}.

The abstract machine KN \cite{Cre07} implements normal order in lockstep
\cite{GPN19}. KN is the full-reducing version of the KAM machine mentioned in
Section~\ref{sec:call-by-name}. In \cite{ABM15}, the authors present a
KN-based full-reducing abstract machine with a global environment (the Strong
Milner Abstract Machine) using the `distillery' approach from the linear
substitution calculus mentioned in Section~\ref{sec:call-by-value}. The
underlying strategy, `linear leftmost-outermost reduction' is a refinement of
normal order and an extension of linear head reduction \cite{DR04}.

As motivated in contribution \ref{C:survey}, we illustrate normal order's
evaluation properties in the evaluation of general recursive functions on
Church numerals. Because normal order is full-reducing and complete, it
delivers Church numeral results in their original representation as normal
forms.  For illustration we use the typical factorial example.\footnote{The
  factorial function is actually \emph{primitive} recursive, but the
  discussion applies to any \emph{general} recursive function.} The boldface
terms denote the Church-encoded intended arithmetic functions and constants.
\begin{displaymath}
  \CH{F} \equiv \lambda f.\lambda n.\, \CH{Cond} \, (\CH{IsZero}\,n) \,
    \CH{One} \, (\CH{Mult} \, n \, ( f \, (\CH{Pred} \, n)))
\end{displaymath}
Given a Church numeral $\CH{n}$, the evaluation of $(\CH{Y}\,\CH{F}\,\CH{n})$
under $\no$ delivers the Church numeral for the factorial of $\CH{n}$. The
term $\CH{Y}$ is the standard non-strict fixed-point combinator
(Appendix~\ref{app:prelim:lambda}). The weak-reducing strategies we have
discussed can also be used but they deliver terms with unevaluated `recursive'
calls. Applicative order cannot be used because it evaluates every subterm
fully to normal form, including the fixed-point combinators which have no
normal form. Strict strategies require other combinators and encodings. We
defer this discussion to Section~\ref{sec:strict-normalisation}.

\subsection{Head reduction \texorpdfstring{$(\hr)$}{(hr)}}
\label{sec:head-reduction}
This strategy, defined in \cite[p.\,504]{Wad76}, chooses the head redex, \ie\
the redex that occurs in head position (Appendix~\ref{app:prelim:lambda}). It
is non-weak, non-strict, and head, and delivers (and is complete for) head
normal forms ($\HNF$ in Figure~\ref{fig:lam-sets}). These are approximate
normal forms where neutrals are unevaluated. The head redex is in left
position and thus head reduction is also a `leftmost' strategy, concretely, an
approximation of leftmost normal order.

The \textsc{abs} rule evaluates abstractions (non-weak).  Like normal order,
head reduction relies on weak call-by-name ($\bn$) as subsidiary in rules
\textsc{con} and \textsc{neu} to avoid the evaluation of abstractions in
operator position.  Unlike normal order, head reduction does not evaluate
neutrals (head). Indeed, normal order and head reduction differ only in rule
\textsc{neu}. The difference can be expressed in jest: unlike head reduction,
normal order lost is head in pursuit of the normal form.  Like normal order
($\no$), head reduction ($\hr$) is a balanced hybrid evaluator and strategy.

Head reduction is a semi-decision procedure for the `useful-as-a-function'
property of terms called `solvability'~\cite{Wad76,Bar84}.  Solvable terms
consist of terms that evaluate to $\NF$ as well as terms that do not evaluate
to $\NF$ but may deliver a $\NF$ when used as functions on suitable operands.

\subsection{Head spine \texorpdfstring{$(\he)$}{(he)}}
\label{sec:head-spine}
This strategy, defined in \cite[p.\,209]{BKKS87}, is the first \emph{spine}
strategy in the survey. Spine strategies are non-strict and go under lambda
just enough to uphold completeness
(Appendix~\ref{app:completeness-leftmost-spine}). They evaluate the term's
spine when depicting the term as a tree. They are big-step equivalent to a
leftmost strategy but evaluate operators more eagerly to (and yet are complete
for) $\HNF$. Concretely, spine strategies evaluate operators to $\HNF$ whereas
their big-step equivalent leftmost strategies evaluate operators to less
evaluated $\WHNF$. The leftmost counterpart of head spine ($\he$) is head
reduction ($\hr$). Both evaluate only needed redexes in the same number of
different steps.

Head spine is non-weak, non-strict, and head, and delivers (and is complete
for) $\HNF$s. The head spine evaluator ($\he$) is uniform in style. The
strategy is uniform. The classic eval-apply evaluator \texttt{headNF} in
\cite[p.\,390]{Pau96} literally implements the head spine $\he$ evaluator in
code (Section~\ref{sec:eval-readback}).

The more eager evaluation of an abstraction operator to $\HNF$ may in some
cases eliminate copies of the bound variable in the body and consequently
require less copies of the operand after contraction. The trivial example is
$(\lambda x.(\lambda y.\lambda z.y) \, x \, x) N$. Head reduction copies the
operand $N$ twice to the two occurrences of $x$. Head spine evaluates the
operator to $(\lambda x.x)$ and consequently copies $N$ once.  There are terms
for which copying is the same, \eg~ $(\lambda x.\CH{Y} x) N$.  The cost of
copying is implementation dependent (Section~\ref{sec:summary-survey}).

\subsection{Strict normalisation \texorpdfstring{$(\sn)$}{(sn)}}
\label{sec:strict-normalisation}
This strategy chooses the leftmost-innermost redex but it evaluates operators
to $\WNF$ using call-by-value ($\bv$) as subsidiary. If the result is an
abstraction (rule \textsc{con}) then it evaluates the redex's operand to
$\WNF$ also using call-by-value. The contractum is evaluated fully to normal
form.  If the result is not an abstraction (rule \textsc{neu}) then it
evaluates the neutral fully to normal form. Finally, it evaluates abstractions
fully to normal form (rule \textsc{abs}).

Strict normalisation is non-weak, strict, non-head, full-reducing, and
incomplete for arbitrary terms. Like normal order ($\no$), strict
normalisation ($\sn$) is a balanced hybrid evaluator and strategy.

We introduced this strategy and its natural semantics evaluator in
\cite[pp.\,195-6]{GPN14} but under the name \texttt{byValue} because it is
equivalent to the classic eval-readback evaluator \texttt{byValue} in
\cite[p.\,390]{Pau96}. In turn, the \texttt{byValue} evaluator antedates and
is identical to the `strong reduction' ($\GL{N}$) eval-readback evaluator in
\cite[p.\,237]{GL02}, save for the evaluation order of operands of neutrals
(left-to-right in \texttt{byValue}, right-to-left in $\GL{N}$). We discuss and
compare \texttt{byValue} and $\GL{N}$ in Section~\ref{sec:eval-readback}.
Recall from Section~\ref{sec:prelim:cbv} that $\GL{N}$ is complete for
convergent terms. Thus, \texttt{byValue} and strict normalisation ($\sn$) are
also complete for convergent terms. The LWF proof
(Section~\ref{sec:hybrid-evalreadback}) will show that strict normalisation
($\sn$) and \texttt{byValue} are one-step equivalent modulo commuting redexes.

Abstract machines for full-reducing call-by-value are scarce. The \SECD\
\cite{Lan64}\cite[Chp.\,10]{FH88}, CAM \cite{CCM84}\cite[Chp.\,13]{FH88} and
ZAM \cite{GL02} abstract machines are weak-reducing. An abstract machine for a
full-reducing call-by-value strategy of the pure lambda calculus with
right-to-left evaluation of operands of neutrals \emph{and} of redexes is
presented in \cite{BBCD20}, with its performance improved in
\cite{BCD21}. Interestingly, the abstract machine is obtained by
deconstructing the normal order KN machine using program transformation and
the concept of hybrid style.  The machine uses a right-to-left weak
call-by-value subsidiary. The relation with right-to-left versions of strict
normalisation ($\sn$), \texttt{byValue}, and $\GL{N}$ is not explored.

A seminal full-reducing call-by-value strategy for the lambda-value calculus
of \cite{Plo75} is defined by the abstract machine obtained as a instance of
the `principal reduction machine' of \cite[p.\,70]{RP04}. A related
full-reducing call-by-value strategy, called `value normal order', is
presented in \cite[Sec.\,7.1]{GPN16}. The latter differs from the former in
that blocks are evaluated left-to-right.

Strict normalisation can be used to evaluate general recursive functions on
Church numerals using the strict fixed-point combinator $\CH{Z}$ and thunking
`protecting by lambda'. Thunking \cite{Ing61,DH92} consists in delaying the
evaluation of subterms that have recursive calls which may diverge. These
subterms are `serious' in opposition to the `trivial' subterms that converge
\cite[p.\,382]{Rey72}. The serious subterms are delayed by placing them inside
an abstraction, hence `protecting by lambda'. To illustrate, take the
following thunked definition of the factorial function:
\begin{displaymath}
  \CH{F} \equiv \lambda f.\lambda n.\, \CH{Cond} \, (\CH{IsZero}\,n) \,
  (\lambda v.\, \CH{One}) \, (\lambda v.\, \CH{Mult} \, n \, ( f \,
  (\CH{Pred}\,n) ) )
\end{displaymath}
Given a Church numeral $\CH{n}$, the term $(\CH{Z}\,\CH{F}\,\CH{n}\,\CH{T})$
evaluates under $\sn$ to the Church numeral for the factorial of $\CH{n}$. An
arbitrary closed term $\CH{T}$, \eg~the identity $\CH{I}$, can be passed to
unthunk. The serious subterm $(\CH{Z}\,\CH{F})\,(\CH{Pred}\, n)$ is protected
by the binding $\lambda v$.

A continuation-passing style encoding is also possible using the identity as
initial continuation:
\begin{displaymath}
  \CH{F} \equiv \lambda f.\lambda n.\, \CH{Cond} \,
  (\CH{IsZero}\,n) \, (\lambda k. \, k\, \CH{One}) \,
  (\lambda k. \,f \, (\CH{Pred}\,n) \, (\lambda x. k \, (\CH{Mult} \, n \, x)))
\end{displaymath}
The serious subterm $(\CH{Z}\,\CH{F})(\CH{Pred}\,n)(\lambda
x.k\,(\CH{Mult}\,n\,x))$ is protected by $\lambda k$. At each step the
passed continuation $k$ is composed with the multiplication which is passed
as a new continuation $\lambda x.k\,(\CH{Mult}\,n\,x)$.


\begin{figure}[p]
  \small
  \begin{flushleft}
    \textbf{Hybrid normal order} ($\hn$):
  \end{flushleft}
  \vspace{4pt}
  \begin{mathpar}
    \inferrule %
    {\hn(B) = B'}{\hn(\lambda x.B) = \lambda x.B'} %
    \and %
    \inferrule %
    {\he(M) = \lambda x.B \quad \hn(\cas{N}{x}{B}) = B'} %
    {\hn(MN) = B'} %
    \and %
    \inferrule %
    {\he(M) = M' \quad M' \not\equiv \lambda x.B \quad \hn(M') = M'' \quad
      \hn(N) = N'} %
    {\hn(MN) = M''N'} %
  \end{mathpar}

  \vspace{4pt}

  \begin{flushleft}
    \textbf{Hybrid applicative order} ($\ha$):
  \end{flushleft}
  \vspace{4pt}
  \begin{mathpar}
    \inferrule %
    {\ha(B) = B'}{\ha(\lambda x.B) = \lambda x.B'} %
    \and %
    \inferrule %
    {\bv(M) = \lambda x.B \quad \ha(N) = N' \quad \ha(\cas{N'}{x}{B}) = B'} %
    {\ha(MN) = B'} %
    \and %
    \inferrule %
    {\bv(M) = M' \quad M' \not\equiv \lambda x.B \quad \ha(M') = M'' \quad
      \ha(N) = N'} %
    {\ha(MN) =  M''N'} %
  \end{mathpar}

  \vspace{4pt}

  \begin{flushleft}
    \textbf{Ahead machine} ($\am$):
  \end{flushleft}
  \vspace{4pt}
  \begin{mathpar}
    \inferrule %
    {\am(B) = B'}{\am(\lambda x.B) = \lambda x.B'} %
    \and %
    \inferrule %
    {\bv(M) = \lambda x.B \quad \bv(N) = N' \quad \am(\cas{N'}{x}{B}) = B'} %
    {\am(MN) = B'} %
    \and %
    \inferrule %
    {\bv(M) = M' \quad M' \not\equiv \lambda x.B \quad \am(M') =  M'' \quad
      \bv(N) = N'} %
    {\am(MN) = M''N'} %
  \end{mathpar}

  \vspace{4pt}

  \begin{flushleft}
    \textbf{Head applicative order} ($\ho$):
  \end{flushleft}
  \vspace{4pt}
  \begin{mathpar}
    \inferrule %
    {\ho(B) = B'}{\ho(\lambda x.B) = \lambda x.B'} %
    \and %
    \inferrule %
    {\ho(M) = \lambda x.B \quad \ho(N) = N' \quad \ho(\cas{N'}{x}{B}) = B'} %
    {\ho(MN) = B'} %
    \and %
    \inferrule %
    {\ho(M) = M' \quad M' \not\equiv \lambda x.B} %
    {\ho(MN) = M'N} %
  \end{mathpar}

  \vspace{4pt}

  \begin{flushleft}
    \textbf{Spine applicative order} ($\so$):
  \end{flushleft}
  \vspace{4pt}
  \begin{mathpar}
    \inferrule %
    {\so(B) = B'}{\so(\lambda x.B) = \lambda x.B'} %
    \and %
    \inferrule %
    {\ho(M) = \lambda x.B \quad \so(N) = N' \quad \so(\cas{N'}{x}{B}) = B'} %
    {\so(MN) = B'} %
    \and %
    \inferrule %
    {\ho(M) = M' \quad M' \not\equiv \lambda x.B \quad \so(M') = M'' \quad
      \so(N) = N'} %
    {\so(MN) = M''N'} %
  \end{mathpar}
  \begin{flushleft}
    \textbf{Balanced spine applicative order} ($\bs$):
  \end{flushleft}
  \vspace{4pt}
  \begin{mathpar}
    \inferrule %
    {\bs(B) = B'}{\bs(\lambda x.B) = \lambda x.B'} %
    \and %
    \inferrule %
    {\ho(M) = \lambda x.B \quad \ho(N) = N' \quad \bs(\cas{N'}{x}{B}) = B'} %
    {\bs(MN) = B'} %
    \and %
    \inferrule %
    {\ho(M) = M' \quad M' \not\equiv \lambda x.B \quad \bs(M') = M'' \quad
      \bs(N) = N'} %
    {\bs(MN) = M''N'} %
  \end{mathpar}
  \caption{Canonical eval-apply evaluators in natural semantics for less known
    strategies.}
  \label{fig:strategies2}
\end{figure}

\subsection{Hybrid normal order \texorpdfstring{$(\hn)$}{(hn)}}
\label{sec:hybrid-normal-order}
Hereafter we discuss the evaluators in Figure \ref{fig:strategies2}
(page~\pageref{fig:strategies2}). Hybrid normal order is non-weak, non-strict,
non-head, full-reducing, and delivers (and is complete for) normal forms
($\NF$). This strategy is presented in \cite[p.\,429]{Ses02} and given its
name because it `is a hybrid of head spine reduction and normal order
reduction'. Normal order is itself hybrid, so `hybrid normal order' can be
confusing, but it is the original name given in \cite[p.\,429]{Ses02}. Here
`hybrid' is used informally but closely to the concept of dependence on a
subsidiary evaluator, what we call hybrid \emph{style} (recall contribution
\ref{C:hybrid}). What the phrase means is that hybrid normal order ($\hn$) is
obtained from normal order ($\no$) by replacing from the latter the subsidiary
call-by-name ($\bn$) by head spine ($\he$). A more precise description is that
hybrid normal order is the spine counterpart of leftmost normal order because
it relies on a spine subsidiary to evaluate operators up to $\HNF$, whereas
normal order relies on a leftmost subsidiary to evaluate operators to $\WHNF$
(Section~\ref{sec:head-spine}). Hybrid normal order ($\hn$) is a balanced
hybrid evaluator and strategy.

The eval-apply hybrid normal order evaluator ($\hn$) is one-step equivalent to
the classic eval-readback evaluator \texttt{byName} in \cite[p.\,390]{Pau96}
(Section~\ref{sec:eval-readback}). This is proven by the LWF proof
(Section~\ref{sec:hybrid-evalreadback}). Recall from
Section~\ref{sec:head-spine} that hybrid normal order's subsidiary, head spine
($\he$), is literally implemented in code by \texttt{headNF} which is the eval
stage of the eval-readback evaluator \texttt{byName}
(Section~\ref{sec:eval-readback}).

\subsection{Hybrid applicative order \texorpdfstring{$(\ha)$}{(ha)}}
\label{sec:hybrid-applicative-order}
Hybrid applicative order is non-weak, strict, non-head, full-reducing,
incomplete, and delivers normal forms ($\NF$). This strategy is presented in
\cite[p.\,428]{Ses02} and given its name because it `is a hybrid of
call-by-value and applicative order'. The word `hybrid' is used again
informally and closely to the concept of dependence on a subsidiary evaluator,
what we call hybrid \emph{style}. Hybrid applicative order ($\ha$) relies on
weak call-by-value ($\bv$) as subsidiary to evaluate operators, just like
normal order ($\no$) relies on weak call-by-name ($\bn$) as subsidiary for
that. In short, hybrid applicative order ($\ha$) differs from applicative
order ($\ao$) in that it does not evaluate abstractions in operator position.
Although hybrid applicative order ($\ha$) is incomplete, it can find more
normal forms than applicative order ($\ao$) because it is less eager (less
undue divergence) and in fewer steps than normal order ($\no$) because it is
strict \cite{Ses02}.

Hybrid applicative order ($\ha$) is a hybrid evaluator and strategy, \emph{but
  it is unbalanced}: it evaluates operands of redexes \emph{more} than its
subsidiary $\bv$ because it calls itself recursively on the operand rather
than calling $\bv$. The assertion in \cite[p.\,428]{Ses02} that hybrid
applicative order is equivalent to the classic eval-readback \texttt{byValue}
evaluator is thus incorrect: hybrid applicative order does not evaluate
operands of redexes using $\bv$ (Section~\ref{sec:eval-readback}).

Hybrid applicative order can be used to evaluate general recursive functions
on natural numbers with the same fixed-point combinator and thunking style
than strict normalisation ($\sn$). But hybrid applicative order is more eager
on operands.

\subsection{Ahead machine \texorpdfstring{$(\am)$}{(am)}}
\label{sec:ahead-machine}
This strategy is defined in \cite[Def.\,5.1]{PR99}. It uses call-by-value
($\bv$) as subsidiary to evaluate operators and operands in redexes, like
strict normalisation ($\sn$), but it also uses call-by-value to evaluate
neutrals. (Recall from Section~\ref{sec:prelim:cbv} that in
\cite[p.\,17]{PR99} call-by-value ($\bv$) is called `inner machine'.) Ahead
machine evaluates abstractions using itself recursively.

Ahead machine is non-weak, strict, non-head, incomplete, and delivers
so-called `value head normal forms' ($\NT{VHNF}$ in Figure~\ref{fig:lam-sets})
which are terms in $\WNF\cap\HNF$. Ahead machine ($\am$) is a balanced hybrid
evaluator and strategy. Although it is incomplete in the pure lambda calculus,
it is a semi-decision procedure for an approximate notion of solvability in
the pure lambda-value calculus \cite{PR99}.

\subsection{Head applicative order \texorpdfstring{$(\ho)$}{(ho)}}
\label{sec:ho}
This is a novel strategy. It is the natural head variation of applicative
order ($\ao$). Its rule structure is similar to $\ao$'s save for the
\textsc{neu} rule where neutrals are not evaluated. Head applicative order
($\ho$) is non-weak, strict, head, incomplete, and delivers head normal forms
($\HNF$). It is a uniform evaluator and strategy.

\subsection{Spine applicative order  \texorpdfstring{$(\so)$}{(so)}}
\label{sec:so}
This is a novel strategy. It is non-weak, strict, non-head, full-reducing, and
incomplete, and delivers $\NF$s. It is the spine-ish counterpart of hybrid
applicative order ($\ha$). Recall from Section~\ref{sec:head-spine} that spine
strategies, which live in the non-strict subspace, evaluate the term's spine
and operators to $\HNF$. A spine strategy is big-step equivalent to a leftmost
strategy that evaluates operators to $\WHNF$. In the strict space we cannot
speak of spineness literally, because evaluated operands are not in the term's
spine, but we think the spine analogy is pertinent when operators are
evaluated to $\HNF$ by a strict subsidiary. Spine applicative order ($\so$)
evaluates operators to $\HNF$ using head applicative order ($\ho$), whereas
hybrid applicative order ($\ha$) evaluates operators to $\WNF$ using
call-by-value ($\bv$). Spine applicative order ($\so$) is a hybrid strategy
\emph{but it is unbalanced}: it evaluates operands of redexes \emph{more} than
its subsidiary by calling itself recursively.

Spine applicative order ($\so$) is more eager than hybrid applicative order
($\ha$).  It has several interesting features. First, it is the \emph{most
  eager} hybrid strategy in the survey. Second, it evaluates general recursive
functions on Church numerals with thunking `protecting by variable' (rather
than `by lambda'), with delimited CPS (continuation-passing style)
\cite{BBD05}, and with the \emph{non-strict} fixed-point combinator $\CH{Y}$.
(Notice the contrast: `most eager' and `non-strict fixed-point'.) To
illustrate, take the delimited-CPS definition of the factorial function:
\begin{displaymath}
  \CH{F} \equiv \lambda f.\lambda n.\lambda
  k.\, \CH{Cond} \, (\CH{IsZero}\,n) \, (k\,\CH{One}) \, (k\, (f \,
  (\CH{Pred}\,n) \, (\CH{Mult}\,n)))
\end{displaymath}
Given a Church numeral $\CH{n}$, the term $(\CH{Y}\,\CH{F}\,\CH{n}\, \CH{I})$
evaluates under $\so$ to the Church numeral for the factorial of $\CH{n}$. The
identity $\CH{I}$ is passed as the initial continuation. The variable $k$,
which stands for the delimited continuation, protects the `serious' subterm
$(\CH{Y}\,\CH{F})\,(\CH{Pred}\,n)\, (\CH{Mult}\,n)$. In the evaluation under
$\so$ there are no administrative evaluation steps resulting from a strict
fixed-point combinator and from the composition of continuations because they
are used at each call, passing the multiplication $(\CH{Mult}\,n)$ as a
continuation to each serious subterm.

\subsection{Balanced spine applicative order %
  \texorpdfstring{$(\bs)$}{(bs)}}
\label{sec:bs}
This is a novel strategy. It is a balanced hybrid variation of spine
applicative order ($\so$) where the subsidiary, head applicative order
($\ho$), is employed to evaluate operands of redexes to recover the balanced
property. It is non-weak, strict, non-head, full-reducing, incomplete, and
delivers $\NF$s. Although it is less eager than spine applicative order
($\so$), balanced spine applicative order ($\bs$) can be used to evaluate
general recursive functions in the same fashion as $\so$.

\subsection{Summary of the survey}
\label{sec:summary-survey}
Figure~\ref{fig:summary-survey} (page~\pageref{fig:summary-survey}) summarises
the strategies, their properties, and their use in evaluating general
recursive functions. Applicative order ($\ao$) is missing from the bottom
table because, as explained in Section~\ref{sec:normal-order}, it evaluates
every subterm fully to normal form, including the fixed-point combinators
which have no normal form, and thus it defeats any thunking mechanism.

The four criteria of weakness, strictness, headness, and uniform/hybrid are
the few variability points in the natural semantics. Their configurations
determine the strategy space. The staged evaluation of the operator in rule
\textsc{neu} is not a variability point as it is subsumed by the
uniform/hybrid configuration: only hybrid strategies use the staged evaluation
because they use the subsidiary to evaluate operators.

\begin{figure}[htb]\small
  \begin{tabular}[t]{llllllll}
    \hline
    \multicolumn{8}{c}{Strategies and properties} \\
    \hline

    $\bv$ & weak & strict & non-head & to $\WNF$ & incomplete & &
    uniform \\

    $\bn$ & weak & non-strict & head & to $\WHNF$ & complete-for &
    leftmost & uniform \\

    $\ao$ & non-weak & strict & non-head & to $\NF$ & incomplete & & uniform \\

    $\no$ & non-weak & non-strict & non-head & to $\NF$ & complete-for &
    leftmost & hybrid bal. \\

    $\hr$ & non-weak & non-strict & head & to $\HNF$ & complete-for &
    leftmost & hybrid bal.\\

    $\he$ & non-weak & non-strict & head & to $\HNF$ & complete-for &
    spine & uniform \\

    $\sn$ & non-weak & strict & non-head & to $\NF$ & incomplete & & hybrid bal.\\

    $\hn$ & non-weak & non-strict & non-head & to $\NF$ & complete-for &
    spine & hybrid bal.\\

    $\ha$ & non-weak & strict & non-head & to $\NF$ & incomplete & & hybrid \\

    $\am$ & non-weak & strict & non-head & to $\NT{VHNF}$ & incomplete &
    & hybrid bal. \\

    $\ho$ & non-weak & strict & head & to $\HNF$ & incomplete &  & uniform \\

    $\so$ & non-weak & strict & non-head & to $\NF$ & incomplete & & hybrid \\

    $\bs$ & non-weak & strict & non-head & to $\NF$ & incomplete & & hybrid
    bal. \\

    \hline
  \end{tabular}

  \vspace{7pt}
  \begin{tabular}[t]{cc}
    \hline
    \multicolumn{2}{c}{Big-step equivalence} \\
    \hline
    Leftmost & Spine \\
    \hline
    $\no$ & $\hn$ \\
    $\hr$ & $\he$ \\
    \hline
  \end{tabular}

  \vspace{7pt}

  \begin{tabular}[t]{clcll}
    \hline
    \multicolumn{5}{c}{Usage in evaluating general recursive functions} \\
    \hline
    Strategy & Result & Combinator & Style & Thunking \\
    \hline
    $\bn$ & $\WHNF$ & $\CH{Y}$ & direct or CPS & none \\
    $\hr$, $\he$ & $\HNF$ & $\CH{Y}$ & direct or CPS & none \\
    $\no$, $\hn$ & $\NF$ & $\CH{Y}$ & direct or CPS & none \\

    $\bv$ & $\WNF$ & $\CH{Z}$ & direct or CPS & protecting by lambda \\
    $\am$ & $\NT{VHNF}$ & $\CH{Z}$ & direct or CPS & protecting by lambda \\
    $\sn$, $\ha$ & $\NF$ & $\CH{Z}$ & direct or CPS & protecting by lambda \\

    $\ho$ & $\HNF$ & $\CH{Y}$ & delimited CPS & protecting by variable \\
    $\so$, $\bs$ & $\NF$ & $\CH{Y}$ & delimited CPS & protecting by variable \\
    \hline
  \end{tabular}
  \caption{Summary of strategies, properties, and use in evaluating general
    recursive functions.}
  \label{fig:summary-survey}
\end{figure}

Each strategy has its uses. As mentioned in Section~\ref{sec:the-setting},
there is no optimal strategy in terms of time (\eg\ contraction counts) or
space (\eg\ term size) because there are terms for which any strategy
duplicates work \cite{Lev78,Lev80,Lam89,AG98}. Optimisations are
implementation dependent, particularly on term representation. For example, we
implement evaluators in Haskell where we can share our deeply-embedded terms:
\begin{code}
    m :: Term
    m = App n n
    where n :: Term
          n = ... -- some large term
\end{code}
There is only a copy of \texttt{n} in \texttt{m}. The evaluation of
\texttt{ev} \texttt{m} does not overwrite \texttt{m} because they are
different Haskell expressions. The first may be memoised but the second
remains intact. We have left efficiency considerations aside because the
particulars constrain the strategy space.
\begin{rem}
  We end this summary drawing the reader's attention to a symmetry which may
  be related to the notion of duality \cite{CH00}. Call-by-name ($\bn$) and
  call-by-value ($\bv$) are uniform and weak-reducing. These strategies are
  respectively used as subsidiaries by hybrid and full-reducing normal order
  ($\no$) and strict normalisation ($\sn$). Normal order uses call-by-name as
  subsidiary on operators. Strict normalisation uses call-by-value on
  operators and on operands. Call-by-name is dual to lambda-value's
  call-by-value \cite{Wad03}. We conjecture that, with convergent terms,
  call-by-value as call-to-$\WNF$ is dual to call-by-name in the weak-reducing
  space, and strict normalisation is dual to normal order in the full-reducing
  space.
\end{rem}

\section{The eval-readback evaluator style}
\label{sec:eval-readback}
An evaluator in eval-readback style is defined as the composition of two
functions: a less-reducing `eval' function that contracts redexes to an
intermediate result with no redex in outermost position, and a
further-reducing `readback' function that `reads back' the intermediate result
to distribute `eval' down the appropriate unevaluated subterms of the
intermediate result to reach further redexes. Readback is a sort of
`selective-iteration-of-eval' function that calls eval directly or in an
eval-readback composition on specific unevaluated subterms.

The style is related to the staged evaluation of normalisation-by-evaluation
\cite{BS91} but without recourse to an external domain. The separation of
evaluation in eval and readback stages can be of help in writing
evaluators. For example, the eval stage can be simply weak or non-head and
written in uniform eval-apply style. The readback stage can then distribute
eval under lambda, or within neutral operands, or both. The separation also
enables independent optimisations, as mentioned in the introduction.

We illustrate the style using first the two classic eval-readback evaluators
\texttt{byValue} and \texttt{byName} in \cite[p.\,390]{Pau96}. We discuss more
eval-readback evaluators later in the section. The \texttt{byValue} and
\texttt{byName} evaluators are both full-reducing and deliver normal forms.
The first is strict and incomplete. The second is non-strict and complete for
$\NF$. We first show and discuss them in their original non-monadic
definitions in the strict Standard ML programming language, but adapted to our
data-type \texttt{Term} and \texttt{subst} function (in the originals, terms
have strings only for free variables and de~Bruijn indices for bound
variables). We want readers to recognise the originals rather than show only
the monadic Haskell versions (shown in
Section~\ref{sec:monadic-eval-readback}). Furthermore, the original
non-monadic definition of \texttt{byValue} can be compared directly with the
also non-monadic definition of the `strong reduction' evaluator
\cite[p.\,237]{GL02} which we have claimed in
Section~\ref{sec:strict-normalisation} to be identical to \texttt{byValue}
modulo commuting redexes.

The \texttt{byValue} evaluator is the composition of an eval evaluator, simply
called \texttt{eval}, and a readback evaluator called \texttt{bodies}. The
\texttt{eval} evaluator is literally the Standard ML version of the
non-monadic call-by-value evaluator \texttt{bv} discussed in
Section~\ref{sec:call-by-value}, which works fine in strict Standard ML.
\begin{code}
  fun eval (App(m,n)) =
    (case eval m of
               Lam(x,b) => eval (subst x (eval n) b)
             | m'       => App(m', eval n))
    | eval t          = t

  fun byValue t = bodies (eval t)
  and bodies (Lam(x,b)) = Lam(x, byValue b)
    | bodies (App(m,n)) = App(bodies m, bodies n)
    | bodies t          = t
\end{code}
The readback evaluator \texttt{bodies} delivers a $\NF$ by fully evaluating
the abstraction bodies (hence its name `\texttt{bodies}') in the $\WNF$
delivered by \texttt{eval}. As the shape of a $\WNF$ indicates, the
unevaluated abstractions occur at the top level or within the operands of
neutrals.

In its first clause, \texttt{bodies} evaluates an abstraction to $\NF$ by
calling the eval-readback composition (\texttt{byValue}) that recursively
distributes \texttt{eval} down the body. In its second clause, \texttt{bodies}
calls itself recursively on the operator of the neutral to distribute itself
down nested applications, and on the operand to reach the unevaluated
abstractions on which it will eventually call the eval-readback composition
(\texttt{byValue}) of the first clause. Finally, \texttt{bodies} is the
identity on variables. It has no clause for redexes because it is a readback
and merely distributes evaluation.

As anticipated in Section~\ref{sec:strict-normalisation}, \texttt{byValue} is
the precursor of the eval-readback `strong reduction' evaluator $\GL{N}$ in
\cite[p.\,237]{GL02}.  The latter only differs from the former in that
neutrals have their operands flattened and evaluated right-to-left. This is
illustrated by the following definition of $\GL{N}$ that mirrors the original
in \cite[p.\,237]{GL02} but with left-to-right evaluation. The $W_i$ terms
stand for $\WNF$s.
\begin{displaymath}
  \begin{array}{rcl}
    \GL{N} & = &\GL{R}\circ\GL{V} \\
    \GL{R}(\lambda x.B) & = & \lambda x.\GL{N}(B) \\
    \GL{R}(x\, W_1 \cdots W_n) & = & x\,\GL{R}(W_1)\cdots\GL{R}(W_n)
      \qquad n \geq 0
  \end{array}
\end{displaymath}
Substitute \texttt{byValue} for $\GL{N}$, \texttt{bodies} for $\GL{R}$, and
\texttt{eval} for $\GL{V}$, and the result is \texttt{byValue}.

As discussed in Sections~\ref{sec:call-by-value}
and~\ref{sec:strict-normalisation}, strong reduction is incomplete for
arbitrary terms but complete for convergent terms, as would any strategy in
that setting.

The \texttt{byName} evaluator is the composition of an eval evaluator called
\texttt{headNF} and a readback evaluator called \texttt{args}. The eval
evaluator \texttt{headNF} literally implements in code the head spine ($\he$)
eval-apply evaluator of Section~\ref{sec:head-spine}.
\begin{code}
  fun headNF (Lam(x,b)) = Lam(x, headNF b)
    | headNF (App(m,n)) =
        (case headNF m of
            Lam(x,b) => headNF (subst x n b)
          | m'       => App(m', n))
    | headNF t          = t

  fun byName t = args (headNF t)
  and args (Lam(x,b))   = Lam(x, args b)
    | args (App(m,n))   = App(args m, byName n)
    | args t            = t
\end{code}
The readback evaluator \texttt{args} delivers a $\NF$ by fully evaluating the
unevaluated (operands of) neutrals (`arguments', hence its name
`\texttt{args}') in the $\HNF$ delivered by \texttt{headNF}. As the shape of a
$\HNF$ indicates, the unevaluated neutrals occur at the top level or within
abstraction bodies.

In its first clause, \texttt{args} evaluates an abstraction to $\NF$ by
calling itself recursively to reach the unevaluated neutrals on whose operands
it will eventually call the eval-readback composition (\texttt{byName}) of the
second clause. In the second clause, \texttt{args} calls itself recursively on
the operator of the neutral to distribute itself down nested applications, and
calls the eval-readback composition (\texttt{byName}) that recursively
distributes \texttt{headNF} down the operand. Finally, \textrm{args} is the
identity on variables. It has no clause for redexes because it is a readback
and merely distributes evaluation.

Recall from Section~\ref{sec:hybrid-normal-order} that eval-readback
\texttt{byName} is one-step equivalent to the eval-apply hybrid normal order
evaluator ($\hn$) discussed in that section. This can be gleaned from their
respective definitions and is proven by the LWF proof
(Section~\ref{sec:hybrid-evalreadback}).

\subsection{Monadic style eval-readback}
\label{sec:monadic-eval-readback}
We show the monadic definitions of \texttt{byValue} and \texttt{byName} in
Haskell. Recall from the previous section that \texttt{eval} is identical to
\texttt{bv} and \texttt{headNF} is identical to \texttt{he}. Hereafter we use
\texttt{bv} and \texttt{he} and assume their eval-apply definitions given in
Sections~\ref{sec:eval-apply-monadic} and~\ref{sec:head-spine} respectively.

In monadic style, \texttt{byValue} and \texttt{byName} are defined by a
monadic (Kleisli) composition. We have respectively replaced the calls to
\texttt{byValue} and \texttt{byName} in \texttt{bodies} and \texttt{args} by
their respective definitions as Kleisli compositions.
\begin{code}
  byValue :: Monad m => Term -> m Term
  byValue =  bodies <=< bv

  bodies :: Monad m => Term -> m Term
  bodies v@(Var _)   = return v
  bodies (Lam x b)   = do b' <- (bodies <=< bv) b
                          return (Lam x b')
  bodies (App m n)   = do m' <- bodies m
                          n' <- bodies n
                          return (App m' n')

  byName :: Monad m => Term -> m Term
  byName =  args <=< he

  args :: Monad m => Term -> m Term
  args v@(Var _)   = return v
  args (Lam x b)   = do b' <- args b
                        return (Lam x b')
  args (App m n)   = do m' <- args m
                        n' <- (args <=< he) n
                        return (App m' n')
\end{code}

\bigskip

\noindent We turn to a different example, namely, the eval-readback version of
normal order presented in \cite[p.\,183]{GPN14}. This evaluator is the
composition of eval-apply call-by-name \texttt{bn}
(Section~\ref{sec:call-by-name}) with the following readback evaluator
\texttt{rn} (from \underline{r}eadback \underline{n}ormal order):
\begin{code}
  no :: Monad m => Term -> m Term
  no =  rn <=< bn

  rn :: Monad m => Term -> m Term
  rn v@(Var _)   = return v
  rn (Lam x b)   = do b' <- (rn <=< bn) b
                      return (Lam x b')
  rn (App m n)   = do m' <- rn m
                      n' <- (rn <=< bn) n
                      return (App m' n')
\end{code}
As a proper readback, \texttt{rn} has no clause for redexes because it simply
distributes evaluation (the Kleisli composition) on the unevaluated
abstractions and neutrals of the $\WHNF$ delivered by \texttt{bn}. A $\WHNF$
has no redex in outermost position.

\bigskip

\noindent The eval-readback examples discussed so far implement full-reducing
strategies. The style can also be used for less-reducing strategies. We show
two examples which are variations of \texttt{byValue} obtained by modifying
the readback \texttt{bodies}. The first example is the `ahead machine' ($\am$)
discussed in Section~\ref{sec:ahead-machine}. In its eval-apply definition the
strategy uses weak, strict and non-head call-by-value \texttt{bv} to evaluate
operators and operands. Its eval-readback version uses \texttt{bv} for eval,
like \texttt{byValue}, but its readback does not distribute evaluation on
operands. We call this readback \texttt{bodies2} for easy comparison with
\texttt{bodies}. We use the identity strategy \texttt{id} for the
non-evaluation of operands.
\begin{code}
  am :: Monad m => Term -> m Term
  am =  bodies2 <=< bv

  bodies2 :: Monad m => Term -> m Term
  bodies2 v@(Var _) = return v
  bodies2 (Lam x b) = do b' <- (bodies2 <=< bv) b
                         return (Lam x b')
  bodies2 (App m n) = do m' <- bodies2 m
                         n' <- id n
                         return (App m' n')
\end{code}
The second example does not appear in the literature. It uses \texttt{bv} for
eval but the readback calls eval (\texttt{bv}) directly on abstraction bodies
rather than via the Kleisli composition. The strategy delivers a $\WNF$ where
the abstractions in the spine are evaluated. We refrain from naming the
strategy and will use the systematic notation that we introduce in
Section~\ref{sec:regimentation-eval-readback} to identify it as
$\Ev\Rb\circ\bv$. We call the readback \texttt{bodies3} for easy comparison
with \texttt{bodies}.
\begin{code}
  \textit{unnamed} :: Monad m => Term -> m Term
  \textit{unnamed} =  bodies3 <=< bv

  bodies3 :: Monad m => Term -> m Term
  bodies3 v@(Var _) = return v
  bodies3 (Lam x b) = do b' <- bv b
                         return (Lam x b')
  bodies3 (App m n) = do m' <- bodies3 m
                         n' <- bodies3 n
                         return (App m' n')
\end{code}

\subsection{Natural semantics style eval-readback}
\label{sec:natural-sem-eval-readback}
In natural semantics, an eval-readback evaluator requires two separate sets of
rules, one for eval and one for readback. Readback \emph{has no} \textsc{con}
\emph{rule} as it simply distributes evaluation in its \textsc{abs} and
\textsc{neu} rules to reach the redexes within the intermediate form which are
evaluated by eval (which does have a \textsc{con} rule).
Figure~\ref{fig:nat-sem-readback} (page~\pageref{fig:nat-sem-readback}) shows
the natural semantics for the readbacks discussed in code in
Section~\ref{sec:monadic-eval-readback}. We omit \textsc{var} rules because
readbacks are identities on variables. The compositions are meta-notation to
abbreviate a sequence of two premises, \eg~$(\bo \circ \bv)(B) = B''$
abbreviates the sequence of premise $\bv(B) = B'$ followed by premise $\bo(B')
= B''$. A proper natural semantics needs separate premises on which to stack
evaluation trees, but we use compositions to abbreviate and to match the
compositions in code. Compositions will be also handy when we generalise in
Section~\ref{sec:generic-reducer}. The identity strategy $\id$ is used for
non-evaluation instead of an omitted premise.

\begin{figure}[htbp]
 \small
  \begin{flushleft}
    \texttt{byValue}'s readback ($\bo$):%
  \end{flushleft}
  \vspace{4pt}
  \begin{mathpar}
    \inferrule*[left=abs] %
    {(\bo \circ \bv)(B) = B''}%
    {\bo(\lambda x.B) = \lambda x.B''} %
    \and %
    \inferrule*[left=neu] %
    {\bo(M) = M' \quad \bo(N) = N'} %
    {\bo(MN) = M'N'} %
  \end{mathpar}

  \vspace{5pt}

  \begin{flushleft}
    \texttt{byName}'s readback ($\ar$):%
  \end{flushleft}
  \vspace{4pt}
  \begin{mathpar}
    \inferrule*[left=abs] %
    {\ar(B') = B'} %
    {\ar(\lambda x.B) = \lambda x.B'} %
    \and %
    \inferrule*[left=neu] %
    {\ar(M) = M' \quad (\ar \circ \he)(N) = N''} %
    {\ar(MN) = M'N''} %
  \end{mathpar}

  \vspace{5pt}

  \begin{flushleft}
    Normal order's readback ($\rn$):%
  \end{flushleft}
  \vspace{4pt}
  \begin{mathpar}
    \inferrule*[left=abs] %
    {(\rn\circ\bn)(B) = B''} %
    {\rn(\lambda x.B) = \lambda x.B''} %
    \and %
    \inferrule*[left=neu] %
    {\rn(M) = M' \quad (\rn\circ\bn)(N) = N''} {\rn(MN) = M'N''} %
  \end{mathpar}

  \vspace{5pt}

  \begin{flushleft}
    Ahead machine's readback ($\bo_2$):
  \end{flushleft}
  \vspace{4pt}
  \begin{mathpar}
    \inferrule*[left=abs] %
    {(\bo_2 \circ \bv)(B) = B''} %
    {\bo_2(\lambda x.B) = \lambda x.B''} %
    \and %
    \inferrule*[left=neu] %
    {\bo_2(M) = M' \quad \id(N) = N} %
    {\bo_2(MN) = M' N} %
  \end{mathpar}

  \vspace{5pt}

  \begin{flushleft}
    A novel strategy's readback ($\bo_3$); the strategy is $\Ev\Rb\circ\bv$
    (Section~\ref{sec:regimentation-eval-readback}):
  \end{flushleft}
  \vspace{4pt}
  \begin{mathpar}
    \inferrule*[left=abs] %
    {\bv(B) = B'} %
    {\bo_3(\lambda x.B) = \lambda x.B'} %
    \and %
    \inferrule*[left=neu] %
    {\bo_3(M) = M' \quad \bo_3(N) = N'} %
    {\bo_3(MN) = M' N'} %
  \end{mathpar}

  \caption{Natural semantics of readback evaluators discussed in
    Section~\ref{sec:monadic-eval-readback}. The \textsc{var} rules are
    omitted because readbacks are identities on variables. There is no
    \textsc{con} rule because readbacks merely distribute evaluation and do
    not contract redexes. The identity strategy $\id$ is used for
    non-evaluation instead of an omitted premise. The compositions are
    meta-notation.}
  \label{fig:nat-sem-readback}
\end{figure}

\subsection{Characterisation of the style}
\label{sec:eval-readback:provisos}
The examples so far illustrate some of the variations within the eval-readback
style. But we must rule out spurious eval-readback evaluators defined by
concocting vacuous readbacks. For example, the uniform eval-apply $\bv$ used
as an eval stage can be defined spuriously in eval-readback style as the
composition $\rb\circ\bv$ where $\rb$ is the following vacuous readback:
{\small\begin{mathpar} %
  \inferrule*[left=abs] %
  {\id(B) = B} %
  {\rb(\lambda x.B) = \lambda x.B} %
  \and %
  \inferrule*[left=neu] %
  {\id(M) = M \quad \id(N) = N} %
  {\rb(MN) = M N}
\end{mathpar}}%
Substituting $\rb$ for the identity strategy $\id$ in any of the premises
obtains a merely recursive readback that is also vacuous.

Readback must be called after eval and it must evaluate more than eval by
calling eval either directly or in an eval-readback composition over some
terms that were unevaluated by eval. For that, readback must also call itself
recursively or in a composition. The following provisos characterise authentic
readbacks in natural semantics:
\begin{enumerate}[label=(\textsc{er}$_\arabic*$)]
\item \label{prov:ER-optwo} Readback must call itself recursively on the
  operator $M$ to distribute itself down nested applications. (Notice that
  eval also calls itself recursively on operators.)
\item \label{prov:ER-la-artwo} In the other two premises, namely, evaluation
  of the body $B$ (non-weakness) and evaluation of the operand $N$
  (non-headness):
  \begin{enumerate}
  \item[(rb-ev)] Readback must call \emph{eval} on at least one of the two
    premises where eval was the identity, calling eval directly or in an
    eval-readback composition.
  \item [(rb-rb)] Readback must call \emph{readback} recursively on at least one
    of the two premises. In the case where eval was the identity on the
    premise, readback must call the eval-readback composition.
  \end{enumerate}
\end{enumerate}
The proviso \ref{prov:ER-la-artwo} holds when readback calls the eval-readback
composition on one of the two premises and the identity on the other premise.
It also holds when readback calls eval on a premise unevaluated by eval and
calls readback on the other premise evaluated by eval. Other combinations are
possible, all of which call eval at least once and readback at least once,
directly on in a composition. It is trivial to check that the provisos
guarantee $\rb \circ \ev \not= \ev$ and $\rb \circ \ev \not= \rb$.

\subsection{Equivalences within the style}
\label{sec:eval-readback:equiv}
The splitting into eval and readback stages provides a modular approach but
the stages are tied-in by eval's evaluation of redexes which restricts the
amount of evaluation that can be moved from eval to readback and
vice-versa. For illustration, consider the weak, non-strict, and non-head
strategy implemented by the following eval-\emph{apply} evaluator $\ea$: 
{\small\begin{mathpar} %
    \inferrule %
    {\id(B) = B } %
    {\ea(\lambda x.B) = \lambda x.B} %
    \and %
    \inferrule %
    {\ea(M) = \lambda x.B \quad \ea(\cas{N}{x}{B}) = B'} %
    {\ea(MN) = B'} %
    \and %
    \inferrule %
    {\ea(M) = M' \quad M'\not\equiv \lambda x.B \quad \ea(N) = N'} %
    {\ea(MN) = M'N'}%
  \end{mathpar}}%
This novel strategy has an interesting property and will be discussed and
named $\IIS$ in Section~\ref{sec:sub:cube}. The point now is that a one-step
equivalent eval-\emph{readback} evaluator is easily defined by moving the
evaluation of neutrals to a readback $\rb$ so that $\ea = \rb \circ \bn$ where
$\bn$ (call-by-name) is the head version of $\ea$. We show the definition of
$\rb$. (The definition of $\bn$ is in Figure~\ref{fig:strategies1}.)
{\small\begin{mathpar} %
    \inferrule %
    { \id(B) = B } %
    {\rb(\lambda x.B) = \lambda x.B} %
    \and %
    \inferrule %
    {\rb(M) = M'\and (\rb\circ\bn)(N) = N'} %
    {\rb(MN) = M'N'}
  \end{mathpar}}%
How much evaluation can be moved between eval and readback to obtain
equivalent eval-readback evaluators depends on the equivalence of
eval-readback and eval-apply evaluators, because eval is naturally defined in
eval-apply style as a uniform evaluator and the question relates to how much
evaluation can be moved between the uniform subsidiary and the hybrid.

\section{Structuring the pure lambda calculus's strategy space}
\label{sec:regimentation}

\subsection{Evaluation templates and generic evaluators}
\label{sec:generic-reducer}
The four criteria of weakness, strictness, headness, and uniform/hybrid are
manifested in the natural semantics by the presence or absence of a formula
for evaluation. For instance, strictness is manifested in the \textsc{con}
rule by the presence of a formula for the evaluation of the operand
$N$. Non-strictness is manifested by the absence of the formula. Both formulas
can be generalised into a single formula $\arone(N) = N'$ where $\arone$ is a
parameter. For strict evaluation, $\arone$ is instantiated to a recursive call
or to a subsidiary call. For non-strict evaluation, $\arone$ is instantiated
to $\id$.

\begin{figure}[p]
  \small
  \begin{flushleft}
    \textbf{Eval-apply evaluation template:}
  \end{flushleft}
  \vspace{4pt}
  \begin{mathpar}
    \inferrule*[left=abs] %
    {\la(B) = B'}%
    {\ea(\lambda x.B) = \lambda x.B'} %
    \and %
    \inferrule*[left=con] %
    {\opone(M) = \lambda x.B \quad \arone(N) = N' \quad \ea(\cas{N'}{x}{B}) = B'} %
    {\ea(MN) = B'} %
    \and %
    \inferrule*[left=neu] %
    {\opone(M) = M' \quad M' \not\equiv \lambda x.B \quad \optwo(M') = M''
      \quad \artwo(N) = N'} %
    {\ea(MN) = M''N'} %
  \end{mathpar}
  \begin{flushleft}
      \textbf{Eval-apply instantiation table:}
  \end{flushleft}
  \vspace{4pt}
  \begin{displaymath}
  \begin{array}{lc|ccccc}
    \hline
    \text{Name} & \text{Eval-apply}~(\ea) & \la & \opone & \arone & \optwo &
      \artwo \\
    \hline
    \text{Call-by-value}     & \bv     & \id & \bv    & \bv    & \id    & \bv \\
    \text{Call-by-name}      & \bn     & \id & \bn    & \id    & \id    & \id \\
    \text{Applicative order} & \ao     & \ao & \ao    & \ao    & \id    & \ao \\
    \text{Normal order}      & \no     & \no & \bn    & \id    & \no    & \no \\
    \text{Head reduction}    & \hr     & \hr & \bn    & \id    & \hr    & \id \\
    \text{Head spine}        & \he     & \he & \he    & \id    & \id    & \id  \\
    \text{Strict normalisation} & \sn     & \sn & \bv    & \bv    & \sn    & \sn \\
    \text{Hybrid normal order}  & \hn     & \hn & \he    & \id    & \hn    & \hn \\
    \text{Hybrid applicative order} & \ha & \ha & \bv    & \ha    & \ha    & \ha \\
    \text{Ahead machine}     & \am     & \am & \bv    & \bv    & \am    & \bv \\
    \text{Head applicative order} & \ho & \ho & \ho    & \ho    & \id    & \id \\
    \text{Spine applicative order} & \so & \so & \ho    & \so    & \so    & \so \\
    \text{Balanced spine applicative order} & \bs & \bs & \ho & \ho & \bs & \bs \\
    \hline
  \end{array}\label{pag:eval-apply-instantiation}
  \end{displaymath}
  \vspace{0.3pt}
  \begin{flushleft}
    \textbf{Eval-apply generic evaluator:}
  \end{flushleft}
\begin{code}
        type Red = forall m. Monad m => Term -> m Term
        gen_eval_apply :: Red -> Red -> Red -> Red -> Red -> Red
        gen_eval_apply \R{la} op1 \E{ar1} op2 \B{ar2} t = case t of
          v@(Var _) -> return v
          (Lam x b) -> do b' <- \R{la} b
                          return (Lam x b')
          (App m n) -> do m' <- op1 m
                          case m' of (Lam x b) -> do n' <- \E{ar1} n
                                                     this (subst n' x b)
                                     _         -> do m'' <- op2 m'
                                                     n'  <- \B{ar2} n
                                                     return (App m'' n')
          where this = gen_eval_apply \R{la} op1 \E{ar1} op2 \B{ar2}

        id :: Red
        id =  return

        bal_hybrid_sn :: Red
        bal_hybrid_sn =  sn
          where bv :: Red
                bv =  gen_eval_apply \R{id} bv \E{bv} id \B{bv}
                sn :: Red
                sn =  gen_eval_apply \R{sn} bv \E{bv} sn \B{sn}
\end{code}
\caption{Eval-apply evaluation template, instantiation table, and generic
  evaluator.}
  \label{fig:generic-eval-apply-3}
\end{figure}

\subsubsection{Eval-apply evaluation template and generic evaluator}
\label{sec:eval-apply-template}
The eval-apply evaluation template shown at the top of
Figure~\ref{fig:generic-eval-apply-3}
(page~\pageref{fig:generic-eval-apply-3}) generalises every formula with a
parameter. The template defines an eval-apply evaluator $\ea$ with evaluator
parameters $\la$, $\opone$, $\arone$, $\optwo$ and $\artwo$. The coloured
parameters respectively determine the evaluator's \R{weakness},
\E{strictness}, and \B{headness}. (We explain the need for colours in a
moment.)  The use of parameters also permits the instantiation of subsidiaries
and the control of staged evaluation in rule \textsc{neu} ($\opone$ followed
by $\optwo$). Concretely, in the \textsc{abs} rule, the $\la$ parameter
determines the evaluation of ($\la$mbda) abstractions. In the leftmost formula
of \textsc{con} and \textsc{neu}, the \emph{shared} $\opone$ parameter
determines the evaluation of operators in applications. In the \textsc{con}
rule, the $\arone$ parameter determines the evaluation of operands
(\E{ar}guments) of redexes. In the \textsc{neu} rule, the $\optwo$ parameter
determines the further evaluation of operators of neutrals. The $\artwo$
parameter determines the evaluation of operands of neutrals. The values
provided for all these parameters, and their interaction, determine the type
of evaluation and form of final result ($\WHNF$, $\WNF$, etc.). The meaning of
the numbering of parameters differ. The $\opone$ and $\optwo$ parameters are
stages (\textsc{neu}) whereas $\arone$ and $\artwo$ are independent parameters
for respectively evaluating operands in a redex (\textsc{con}) and in a neutral
(\textsc{neu}). Thus, we use colours for the important parameters $\la$,
$\arone$, and $\artwo$. The colours are also a visual aid and will be most
useful from Section~\ref{sec:sub:cube} onwards.

Recall from Section~\ref{sec:summary-survey} that the staged evaluation of the
operator in rule \textsc{neu} with $\opone$ and $\optwo$ is not a variability
point because it is subsumed by the uniform/hybrid configuration. Only hybrid
strategies use $\optwo$ in \textsc{neu} because they use the subsidiary on
$\opone$ in \textsc{neu} and \textsc{con}. Consequently, there is no $\optwo$
parameter in rule \textsc{con}. It would be nonsensical because the operator
abstraction $\lambda x.B$ would be further evaluated using $\optwo$ but the
redex $(\lambda x.B)N'$ would be contracted with the abstraction.

The instantiation table in the middle of Figure~\ref{fig:generic-eval-apply-3}
shows that the evaluation template expresses all the eval-apply evaluators in
Figures~\ref{fig:strategies1} and~\ref{fig:strategies2}. More evaluators can
be obtained by playing with the parameters. This is how we originally obtained
head applicative order ($\ho$), spine applicative order ($\so$), and balanced
spine applicative order ($\bs$). We will obtain more novel evaluators in this
fashion throughout this Section~\ref{sec:regimentation}.

The code at the bottom of Figure~\ref{fig:generic-eval-apply-3} shows a
generic (higher-order, parametric) eval-apply evaluator
\texttt{gen\_eval\_apply} that implements the eval-apply evaluation
template. It takes five plain evaluator parameters and delivers a plain
evaluator. The type \texttt{Red} (from \texttt{Red}ucer) is the type of a
plain evaluator (Section~\ref{sec:eval-apply}). The local identifier
\texttt{this} is an abbreviation to get a fixed-point definition. Plain
evaluators are obtained as fixed-points of the generic evaluator. The figure
shows the example for strict normalisation, which we call
\texttt{bal\_hybrid\_sn} in order to define \texttt{bv} and \texttt{sn} as
fixed points of \texttt{gen\_eval\_apply}.

\subsubsection{Eval-readback evaluation template and generic readback}
\label{sec:eval-readback-template}
An evaluation template for eval-readback evaluators can be similarly provided.
Actually, only a template for readback is needed because eval is naturally
expressed in eval-apply style. Rarely will it be expressed in eval-readback
style, and rarely still with the latter's eval again in eval-readback style,
etc.  Eventually, the last eval is written in eval-apply style.

Considering the readback rules and provisos in
Sections~\ref{sec:eval-readback} and~\ref{sec:normal-order}, and assuming the
abbreviation of compositions introduced in Section~\ref{sec:eval-readback},
there are two possible parameters for a readback evaluator $\rb$ in its
\textsc{abs} and \textsc{neu} rules, namely $\la$ and $\artwo$. Readback has
no \textsc{con} rule (no $\opone$ and $\arone$ parameters) because the
intermediate results delivered by eval have no redexes in outermost
position. Also, according to proviso \ref{prov:ER-optwo}, the $\optwo$
parameter must be a recursive call to readback to distribute it down nested
applications.

\begin{figure}
  \small
  \begin{flushleft}
    \textbf{Readback evaluation template:}
  \end{flushleft}
  \begin{mathpar}
    \\
    \inferrule*[left=abs] %
    {\la(B) = B'} %
    {\rb(\lambda x.B) = \lambda x.B'} %
    \and %
    \inferrule*[left=neu] %
    {\rb(M) = M' \quad \artwo(N) = N'} %
    {\rb(MN) = M'N'}%
    \\
  \end{mathpar}
  \begin{flushleft}
      \textbf{Eval-readback instantiation table:}
  \end{flushleft}
  \begin{displaymath}
  \begin{array}{lcc|ccc}
    \\
    \hline
    \text{Name} & \text{Eval} & \text{Readback}~(\rb) & \la & \artwo \\
    \hline
    \mathtt{byValue}     & \bv & \bo   & \bo\circ\bv  & \bo \\
    \mathtt{byName}      & \he & \ar   & \ar          & \ar\circ\he \\
    \text{Normal order}  & \bn & \rn   & \rn\circ\bn  & \rn\circ\bn \\
    \text{Ahead machine} & \bv & \bo_2 & \bo_2\circ\bv & \id \\
    \mathit{unnamed}     & \bv & \bo_3 & \bv          & \bo_3 \\
    \hline
  \end{array}
\end{displaymath}
  \vspace{0.3pt}
  \begin{flushleft}
    \textbf{Readback generic evaluator:}
  \end{flushleft}
\begin{code}
        type Red = forall m. Monad m => Term -> m Term
        gen_readback :: Red -> Red -> Red
        gen_readback \R{la} \B{ar2} t = case t of
          v@(Var _) -> return v
          (Lam s b) -> do b' <- \R{la} b
                          return (Lam s b')
          (App m n) -> do m' <- this m
                          n' <- \B{ar2} n
                          return (App m' n')
          where this = gen_readback \R{la} \B{ar2}

        eval_readback_byValue :: Red
        eval_readback_byValue =  bodies <=< bv
          where bv     :: Red
                bv     =  gen_eval_apply \R{id} bv \E{bv} id \B{bv}
                bodies :: Red
                bodies =  gen_readback \R{(bodies <=< bv)} \B{bodies}
\end{code}
\caption{Readback evaluation template, instantiation table, and generic
  readback evaluator.}
  \label{fig:generic-readback-3}
\end{figure}

Figure~\ref{fig:generic-readback-3} (page~\pageref{fig:generic-readback-3})
shows the readback evaluation template, an instantiation table that expresses
all the readback evaluators discussed in Section~\ref{sec:eval-readback}, and
a generic (higher-order, parametric) readback evaluator that implements the
readback evaluation template. Plain eval-readback evaluators are obtained by
monadic composition of an eval defined using \texttt{gen\_eval\_apply} and a
readback defined using \texttt{gen\_readback}. The figure shows an example for
\texttt{byValue}, which we call \texttt{eval\_readback\_byValue} in order to
define \texttt{bv} and \texttt{bodies} as fixed points. We will define more
eval-readback evaluators in Section~\ref{sec:regimentation-eval-readback} by
playing with the parameters.

\subsection{Uniform/Hybrid as structuring criteria}
\label{sec:uniform-vs-hybrid}
The interaction of weakness, strictness and headness affects the amount of
evaluation. Hence, these criteria do not provide an orthogonal classification
of the strategy space. An orthogonal classification is provided by the
uniform/hybrid criterion that splits the strategy space in two and is
independent of the other three criteria.

As discussed in contribution \ref{C:hybrid}, we must distinguish
uniform/hybrid \emph{strategy} from uniform/hybrid \emph{style} of
definitional device. It may be possible to write a uniform-style device
(without subsidiary devices) for what is actually a hybrid strategy. We
provide examples in Section~\ref{sec:intuition-uniform-hybrid}. It may be also
possible to write a hybrid-style device (with spurious subsidiary devices) for
a uniform strategy.  We will show an example in
Section~\ref{sec:regimentation-hybrid} (page~\pageref{pag:IIS-hybrid}).

Hybrid-style definitional devices have appeared and continue to appear in the
literature, either under no denomination (\eg\
\cite{Rea89,Mac90,Pau96,Pie02}), or under names such as `levelled' or
`of-level-$n$' (\eg\ \cite{PR99,RP04}), `stratified' (\eg\ \cite{AP12,CG14}),
or `layered' (see \cite{GPN14} for references). There are many hybrid
evaluators and strategies (see Figure~\ref{fig:summary-end-figure-eval-apply}
on page~\pageref{fig:summary-end-figure-eval-apply}). The notion has
theoretical interest, to obtain properties such as completeness, spineness,
etc., and has proven useful in gradual typing \cite{GNS14}, solvability
\cite{RP04,CG14,GPN16,AG22}, machine-checked program derivations of
operational semantics devices \cite{BCZ17,BC19,BBCD20}, and reduction theories
for effect handlers \cite{SPB23}. The concept of hybrid style is also
intuitively connected to that of a staged eval-readback evaluator
\cite{Pau96,GL02}.  Making the connection precise is partly the motivation for
this paper.

In the introduction we gave an intuitive definition of uniform/hybrid strategy
in terms of in/dependence on less-reducing subsidiary strategies. In
\cite{GPN14} we formalised the notion of uniform/hybrid strategy (which we
redundantly called in that paper `uniform/hybrid strategy in its nature')
using small-step context-based reduction semantics \cite{Fel87,FH92,FFF09}.
This operational semantics style works with the notion of a one-step strategy
(Appendix~\ref{app:prelim:lambda}), namely, a total function (a subset of the
reduction relation) that takes an input term and locates and contracts one
redex. An evaluation sequence is obtained by iteration. A strategy can be
defined extensionally by its set of so-called `contexts' which are terms with
one hole where the latter marks the redex position. The extensional definition
is captured intensionally by a context-based reduction semantics which
basically defines the set of contexts and permissible redexes using
context-free grammars. Since we wish to improve on \cite{GPN14}, in this
section we temporarily leave the big-step style and switch to that small-step
definitional device.  We provide a quick overview of it in
Section~\ref{sec:context-based-red-semantics}. We provide the necessary
intuitions for the formal definition of uniform/hybrid strategy in
Section~\ref{sec:intuition-uniform-hybrid} using illustrative examples. We
arrive at the formal definition in Section~\ref{sec:def-uniform-hybrid}.

\subsubsection{Small-step context-based reduction semantics}
\label{sec:context-based-red-semantics}
A \emph{context} is a term with one hole, the latter written $\hole$. For
example, $\lambda x.\hole x$ is a context. The set of contexts $\NT{C}$ is
defined by adding a hole to terms: $\NT{C} ::= \hole\ |\ \lambda \Var.\NT{C}\
|\ \NT{C}\,\Lambda\ |\ \Lambda\,\NT{C}$.\footnote{Recall that throughout the
  paper we overload non-terminal symbols to denote the sets generated by them,
  \eg~$\NT{C}$ also denotes the set generated by that non-terminal.}  We use
$\ctx_1$, $\ctx_2$, etc., for contexts in $\NT{C}$. Given a context $\ctx_1$
we use the function application notation $\ctx_1(R)$ and $\ctx_1(\ctx_2)$
respectively for the operation of replacing the hole by a redex $R$ and by
another context $\ctx_2$, irrespectively of variable captures. The first
operation delivers a term and the second delivers a context. Because the
function application notation is uncurried, we also use composition
$\ctx_1\circ\ctx_2$. For example, let $\ctx_1 = \lambda x.\hole x$, $\ctx_2 =
y\hole$, and $R$ a redex, then $\ctx_1(R) = \lambda x.R x$, $\ctx_2(R) = yR$,
$\ctx_1(\ctx_2) = \lambda x.y \hole x$, $\ctx_2(\ctx_1) = y(\lambda x. \hole
x)$, $(\ctx_1(\ctx_2))(R) = (\ctx_1\circ \ctx_2)(R) = \lambda x.yRx$, and
$(\ctx_2(\ctx_1))(R) = (\ctx_2\circ \ctx_1)(R) = y(\lambda x. R x)$.

A strategy $\st$ is defined by a set of contexts $\NT{ST}\in\NT{C}$ and a set
of permissive redexes within holes $\NT{R}$, both defined by context-free (in
our case, EBNF) grammars with non-terminal axioms $\NT{ST}$ and $\NT{R}$, such
that for every term $M$ there is at most one context $\ctx_1 \in \NT{ST}$
(derivable from non-terminal $\NT{ST}$) such that $M = \ctx_1(R)$ and
$R\in\NT{R}$ is the next redex to contract. The redex is contracted and the
process is repeated to produce an evaluation sequence that stops when there is
no derivable context.  Spurious strategies like $\id$ and strategies that have
no standalone hole context ($\ctx_1 = \hole$) for the outermost redex are
discarded.

To illustrate, take the following sets of contexts and of redexes for the
applicative order strategy:
\begin{displaymath}
  \begin{array}{lll}
    \NT{AO} & ::= & \hole\; |\; \lambda \Var.\NT{AO}\; |\; \NT{AO}{} \;
      \Lambda\; |\; \NF \; \NT{AO} \\
    \NT{R} & ::= & (\lambda \Var.\NF)\NF
  \end{array}
\end{displaymath}
The set of contexts is generated by $\NT{AO}$, and the set of permissible
redexes is generated by $\NT{R}$. Reading $\NT{AO}$'s grammar from left to
right, a redex can occur in outermost position (standalone hole), under
lambda, in the operator, or in the operand if the operator is in
$\NF$. Contraction takes place iff the redex has the body and the operand in
$\NF$. Given a term $M$, there is at most one context $\ctx_1$ derivable from
non-terminal $\NT{AO}$ such that $M = \ctx_1(R)$ and $R\in\NT{R}$. For
example, given $(\lambda x.(\lambda y.z)x)N$ with arbitrary $N$, the
uniquely-derivable context is $(\lambda x.\hole)N$ with $(\lambda y.z)x \in
\NT{R}$. Contraction takes place delivering the contractum $(\lambda x.z)N$.
The process is repeated to produce an evaluation sequence that stops when
there is no derivable context.

\subsubsection{Intuition of uniform/hybrid strategy}
\label{sec:intuition-uniform-hybrid}
To provide the intuition we use the normal order strategy because it is the
most well-known strategy of the pure lambda calculus. We reuse the context
grammars presented in \cite{GPN14} without redundant productions (\eg\
redundant holes as discussed in \cite{GPN14}). The details of the grammars are
unimportant. The point is to show \emph{alternative} grammars (styles) for the
\emph{same} hybrid strategy and illustrate from those examples the notions of
dependency and inclusion that makes `hybrid' a style-independent property. We
start with a hybrid-style grammar for normal order which explicitly relies on
a subsidiary call-by-name grammar.  We then show two alternative hybrid-style
grammars which `hide' the call-by-name grammar and where the second grammar
has a mutual dependency. We finally show a uniform-style grammar for normal
order.

First, we start with a hybrid-style grammar of normal order contexts $\NT{NO}$
that uses subsidiary call-by-name contexts $\NT{BN}$ explicitly. The $\NT{NE}$
contexts are for reducing neutrals to $\NF$. The set of redexes is $\NT{R}$.
We omit it in subsequent grammars because it remains the same throughout the
discussion. In \cite{GPN14} we showed that an implementation of the
context-based reduction semantics with this grammar derives by program
transformation to the normal order ($\no$) evaluator of
Figure~\ref{fig:strategies1}.
\begin{displaymath}
  \begin{array}{lr}
    \begin{array}{lll}
      \NT{BN}  & ::= & \hole\ |\ \NT{BN} \ \Lambda \\
      \NT{NO}  & ::= & \NT{BN}\ |\ \lambda \Var.\NT{NO}\ |\ \NT{NE} \\
      \NT{NE}  & ::= & \NT{NE} \ \Lambda\ |\ \Var \ \{\NF\}^* \ \NT{NO} \\

      \NT{R}   & ::= & (\lambda \Var.\Lambda)\Lambda
    \end{array}
  \end{array}
\end{displaymath}
Next, we show a hybrid-style grammar that hides $\NT{BN}$. The subsidiary is
now $\NT{NW}$ which contains the hole, the call-by-name contexts now
represented by $\NT{NW}\,\Lambda$, and the contexts that evaluate neutrals to
$\WNF$.
\begin{displaymath}
  \begin{array}{lll}
    \NT{NW}  & ::= & \hole\ |\ \NT{NW}\ \Lambda\
                     |\ \Var\ \{\NF\}^*\ \NT{NW}\ \{\Lambda\}^* \\
    \NT{NO}  & ::= & \NT{NW}\ |\ \lambda \Var.\NT{NO}\
                     |\ \Var\ \{\NF\}^*\ \NT{NO}\ \{\Lambda\}^*
  \end{array}
\end{displaymath}
Next, we show a hybrid-style grammar which has a \emph{mutual dependency} on a
different non-terminal $\NT{NN}$ that has the hole, the call-by-name contexts,
and the contexts that evaluate neutrals to $\NF$. This is the `preponed'
grammar in \cite[p.\,195]{GPN14}.
\begin{displaymath}
\begin{array}{lll}
    \NT{NN}  & ::= & \hole\ |\ \NT{NN} \ \Lambda\
                   |\ \Var\ \{\NF\}^*\ \NT{NO}\ \{\Lambda\}^* \\
    \NT{NO}  & ::= & \NT{NN}\ |\ \lambda \Var.\NT{NO}
  \end{array}
\end{displaymath}
Finally, we show a uniform-style grammar, it has only one
non-terminal.\footnote{The one non-terminal grammar cannot be written in
  context-free grammars like BNF that lack the sequence notation
  $\{\NT{X}\}^*$. However, the extra productions would only be for sequences,
  \ie\ $\NT{A}\, ::=\, \epsilon\, |\, \NT{X}\, \NT{A}$. A small-step device
  related to the one non-terminal grammar is the uniform-style structural
  operational semantics of normal order given in \cite[p.\,502]{Pie02}. }
\begin{displaymath}
  \begin{array}{lll}
  \NT{NO}  & ::= & \hole\ \{\Lambda\}^*\ |\ \lambda\Var.\NT{NO}\ |\
    \Var\ \{\NF\}^*\ \NT{NO}\ \{\Lambda\}^*
  \end{array}
\end{displaymath}
All the grammars generate the same set of normal order contexts. All
$\NT{BN}$, $\NT{NW}$, and $\NT{NN}$ contexts are in $\NT{NO}$. It is the
inclusion (\ie\ composition) of contexts what determines a dependency, not the
expression of the inclusion using explicit non-terminals, which is a matter of
style. The dependency is inextricable when contexts are not closed under
inclusion.  Let us rephrase the argument in \cite[Sec.\,4]{GPN14}.  Take the
call-by-name context $\hole N$. If we place in $\hole N$ another call-by-name
context, say $\hole M$, then we get a call-by-name context $\hole M N$. This
is the case for any call-by-name context. Call-by-name is closed under context
inclusion.  If we place in $\hole N$ a normal order context $\lambda x.\hole$
then we get $(\lambda x.\hole)N$ which is not a call-by-name context.
Call-by-name does not include (depend on) normal order.

In contrast, take the context $\lambda x.\hole N$ which is a normal order
context. If we replace the hole with the normal order context $\lambda
y.\hole$ then we get $\lambda x.(\lambda y.\hole)N$ which is not a normal
order context. Normal order is not closed under context inclusion. If we replace
the hole with any call-by-name context, say $\hole M$, then we get a normal
order context $\lambda x.\hole M N$. Normal order depends on call-by-name.

$\NT{NO}$ defines a hybrid strategy that depends on what can be identified (or
as we say in \cite{GPN14}, `unearthed') as call-by-name contexts. It is made
to explicitly depend on $\NT{BN}$ in the first definiens, on $\NT{NW}$ in the
second, on $\NT{NN}$ (mutually) in the third, and with no explicit dependency
(uniform-style) in the last.

\subsubsection{Definition of uniform/hybrid strategy}
\label{sec:def-uniform-hybrid}
The following definitions formalise the concepts of uniform and hybrid
strategy using closeness under context inclusion (or context composition) and
dependency on subsidiary. There may be many possible choices of subsidiaries
each of which determines a particular definiens.
\begin{defi}
  \label{def:nature}
  \begin{enumerate}[(i)]
  \item A strategy $\NT{ST}$ is uniform iff for any $\ctx_1,\ctx_2 \in \NT{C}$
    it is the case that $\ctx_1(\ctx_2) \in \NT{ST}$ iff $\ctx_1\in\NT{ST}$
    and $\ctx_2\in\NT{ST}$.
  \item A strategy is hybrid iff it is not uniform.
  \end{enumerate}
\end{defi}
\begin{defi}
  \label{def:dependency}
  A strategy $\NT{HY}$ \emph{depends on} a different subsidiary strategy
  $\NT{SU}$ iff there exists a context $\ctx_1\in\NT{HY}$ such that for any
  context $\ctx_2\in\NT{C}$, it is the case that $\ctx_1(\ctx_2) \in \NT{HY}$
  iff $\ctx_2\in\NT{SU}$.
\end{defi}
These two definitions supersede the definitions in \cite[Sec.\,4]{GPN14}.
Notice that Definition~\ref{def:nature}(ii) implies
Definition~\ref{def:dependency}.
\begin{prop}
  Definition~\ref{def:nature}(i) and Definition~\ref{def:dependency} cannot
  be the case simultaneously.
\end{prop}
\begin{proof}
  Definition~\ref{def:nature}(i) requires that $\NT{ST}$ contexts be closed
  under context inclusion.  Definition~\ref{def:dependency} requires that
  there exist $\NT{SU}$ contexts included in $\NT{ST}$. Because all $\NT{ST}$
  contexts are closed under inclusion then all $\NT{SU}$ contexts are also
  closed under inclusion and $\NT{SU} = \NT{ST}$.
\end{proof}
As indicated in contribution \ref{C:hybrid}, a hybrid strategy uses the
subsidiary to evaluate the operator $M$ in applications $MN$ to obtain a redex
$(\lambda x.B)N$. A balanced hybrid is either a non-strict hybrid (does not
evaluate $N$) or is a strict hybrid that uses the subsidiary to evaluate $N$.
All the balanced hybrid evaluators discussed in the survey do exactly this. In
contrast, the unbalanced hybrid evaluators use themselves recursively on $N$.
\begin{defi}
  \label{def:balanced}
  A hybrid strategy is \emph{balanced} when it has the same set of permissible
  redexes than its subsidiary.
\end{defi}

\subsection{Structuring the uniform space}
\label{sec:sub:cube}
We structure the strategy space starting with uniform strategies because they
are simpler. The following provisos gleaned from the eval-apply instantiation
table in Figure~\ref{fig:generic-eval-apply-3} characterise uniform-style
eval-apply evaluators (`uniform evaluators', for short):
\begin{enumerate}[label=(\textsc{u}$_{\arabic*}$)]
\item \label{prov:U-opone-optwo} Uniform evaluators call themselves
  recursively on the shared $\opone$ parameter of \textsc{con} and
  \textsc{neu}, and thus call the identity in \textsc{neu}'s $\optwo$.

\item \label{prov:U-la-arone-artwo} Uniform evaluators differ only on the
  \R{weakness} ($\la$), \E{strictness} ($\arone$), and \B{headness} ($\artwo$)
  parameters, where they either call themselves recursively or call the
  identity.
\end{enumerate}
\begin{prop}
  \label{prop:uniform-evaluator-strategy}
  An evaluator that satisfies the uniform provisos defines a uniform strategy.
\end{prop}
\begin{proof}
  In all the natural semantics rules, each premise at most calls the evaluator
  recursively so that every term is evaluated by the same evaluator. There can
  be no inclusion of two contexts other than those generated by the evaluator.
\end{proof}
Recall from Section~\ref{sec:eval-apply-natural-sem} and
Figure~\ref{fig:example-nat-sem} that the set of contexts generated by the
evaluator can be obtained from the \textsc{con} nodes by replacing the
contracta with holes in the evaluation sequences obtained from derivation
trees.

A uniform evaluator inexorably defines a uniform strategy. In contrast, as
shown in Section~\ref{sec:intuition-uniform-hybrid}, a uniform-style
small-step context-based reduction semantics can define a hybrid strategy.

To continue the analysis of uniform strategies without explicit reference to
the template and instantiation table of Figure~\ref{fig:generic-eval-apply-3}
(page~\pageref{fig:generic-eval-apply-3}), we introduce a convenient
systematic notation for encoding uniform evaluators/strategies as a triple of
letters where each letter can either be $\stgy{S}$ for `self' (recursive) or
$\stgy{I}$ for `identity'. The following table shows the encoding for the
uniform strategies in the survey.
\begin{displaymath}
  \begin{array}{lcccc}
    \hline
    \text{Name} & \la & \arone & \artwo & \text{Systematic notation} \\
    \hline
    \text{Call-by-value}~(\bv) & \id & \bv & \bv & \stgy{ISS} \\
    \text{Call-by-name}~(\bn)  & \id & \id & \id & \stgy{III} \\
    \text{Applicative order}~(\ao) & \ao & \ao & \ao & \stgy{SSS} \\
    \text{Head spine}~(\he) & \he & \id & \id & \stgy{SII}  \\
    \text{Head applicative order}~(\ho) & \ho & \ho & \id & \stgy{SSI} \\
    \IIS & \id & \id & \IIS & \IIS \\
    \SIS & \SIS & \id & \SIS & \SIS \\
    \ISI & \id & \ISI & \id & \ISI \\
    \hline
  \end{array}
  \label{sec:sub:cube:table}
\end{displaymath}
The last four are novel, and the last three are unnamed. We name them using
their encoding. The first one is $\IIS$ (weak, non-strict, non-head) which
evaluates operands of neutrals that are not abstractions. The \texttt{ea}
evaluator of Section~\ref{sec:eval-readback:equiv} is exactly $\IIS$. This
strategy is interesting because it is complete for $\WNF$ in the pure lambda
calculus and can be used to evaluate general recursive functions as follows:
\begin{center}
  \begin{tabular}[t]{ccccc}
    \hline
    Strategy & Result & Combinator & Style & Thunking \\
    \hline
    $\stgy{IIS}$ & $\WNF$ & $\CH{Y}$ & direct & none \\
    \hline \\
  \end{tabular}
\end{center}
The remaining two strategies are $\SIS$ (non-weak, non-strict, non-head and
delivers $\NF$s) and $\ISI$ (weak, strict, head and delivers $\WHNF$s) which
are uninteresting by themselves but are useful for hybrid and eval-readback
evaluator definitions (Section~\ref{sec:conclusions}).

By considering $\stgy{I}=0$, $\stgy{S}=1$, and the obvious partial order on
boolean triples, from least reducing $\stgy{III}$ (weak, non-strict, head) to
most reducing $\stgy{SSS}$ (non-weak, strict, non-head), the uniform strategy
space is neatly structured by \emph{Gibbons's Beta Cube} lattice shown below
left.\footnote{As mentioned in the introduction, we name the cube after Jeremy
  Gibbons who suggested to us the use of booleans to construct a lattice.}
Non-strict strategies are on the front face and strict strategies are on the
rear face. The `more-reducing' relation is reflected on the inclusion relation
of final forms along the \R{weakness} and \B{headness} axes, shown on the
right. The arrows in the opposite direction indicate the inclusion.  For
instance, $\NF$s are valid $\HNF$s so the latter include the former.
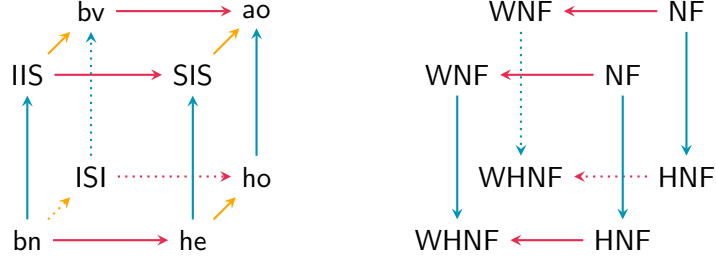
\begin{center}
  \begin{tikzpicture}
    [ arrows=-stealth 
    , thick 
    , y = {(-3.85mm,-3.85mm)}, z = {(0cm,1cm)} 
    , xscale = 2.2, yscale = 2.2 
    ]
    \node (000) at (0,0,0)  {$\III$} ; 
    \node (001) at (0,0,1)  {$\IIS$} ; 
    \node (010) at (0,-1,0) {$\ISI$} ; 
    \node (011) at (0,-1,1) {$\ISS$} ; 
    \node (100) at (1,0,0)  {$\SII$} ; 
    \node (101) at (1,0,1)  {$\SIS$} ; 
    \node (110) at (1,-1,0) {$\SSI$} ; 
    \node (111) at (1,-1,1) {$\SSS$} ; 
    \draw [\myred]          (000) -- (100) ;
    \draw [dotted,\myred]   (010) -- (110) ;
    \draw [\myred]          (001) -- (101) ;
    \draw [\myred]          (011) -- (111) ;
    \draw [dotted,\mygreen] (000) -- (010) ;
    \draw [\mygreen]        (100) -- (110) ;
    \draw [\mygreen]        (101) -- (111) ;
    \draw [\mygreen]        (001) -- (011) ;
    \draw [\myblue]         (000) -- (001) ;
    \draw [\myblue]         (100) -- (101) ;
    \draw [\myblue]         (110) -- (111) ;
    \draw [dotted,\myblue]  (010) -- (011) ;
  \end{tikzpicture}\hspace{.1\linewidth}
  \begin{tikzpicture}
    [ arrows=-stealth 
    , thick 
    , y = {(-3.85mm,-3.85mm)}, z = {(0cm,1cm)} 
    , xscale = 2.2, yscale = 2.2 
    ]
    \node (000) at (0,0,0)  {$\WHNF$};
    \node (001) at (0,0,1)  {$\WNF$} ;
    \node (010) at (0,-1,0) {$\WHNF$};
    \node (011) at (0,-1,1) {$\WNF$} ;
    \node (100) at (1,0,0)  {$\HNF$} ;
    \node (101) at (1,0,1)  {$\NF$}  ;
    \node (110) at (1,-1,0) {$\HNF$} ;
    \node (111) at (1,-1,1) {$\NF$}  ;
    \draw [\myred]          (100) -- (000) ;
    \draw [dotted,\myred]   (110) -- (010) ;
    \draw [\myred]          (101) -- (001) ;
    \draw [\myred]          (111) -- (011) ;
    \draw [\myblue]         (001) -- (000) ;
    \draw [\myblue]         (101) -- (100) ;
    \draw [\myblue]         (111) -- (110) ;
    \draw [dotted,\myblue]  (011) -- (010) ;
  \end{tikzpicture}
\end{center}
Notice that no uniform strategy is complete for normal forms. This requires a
hybrid strategy (contribution~\ref{C:struct},
Section~\ref{sec:regimentation-hybrid},
Appendix~\ref{app:completeness-leftmost-spine}, \cite{GPN14}).

The inclusion relation of final forms suggests the possibility of
absorption, which indeed holds for strategies on the left `L' of the front
(non-strict) face of the cube.
\begin{defi}[Absorption]
  A strategy $\st_2$ absorbs a strategy $\st_1$ iff $\st_2 \circ \st_1 =
  \st_2$.
\end{defi}
Absorption is a big-step property. It is not necessary for $\st_1$ to deliver
a prefix of $\st_2$'s evaluation sequence. Absorption does not hold when
$\st_2$ evaluates less than $\st_1$.
\begin{prop}
  Both $\IIS$ and head spine $(\he)$ each absorb call-by-name $(\bn)$.
\end{prop}
\begin{proof}
The three strategies are quasi-leftmost, which subsumes the spine and
leftmost-outermost strategies
(Appendix~\ref{app:completeness-leftmost-spine}), and are complete for their
respective final results, which satisfy the inclusions $\WNF\subset\WHNF$ and
$\HNF\subset\WHNF$. If a term has a $\WNF$/$\HNF$ then it also has a
$\WHNF$. Call-by-name ($\bn$) finds the $\WHNF$ and $\IIS$/$\he$ finds the
$\WNF$/$\HNF$.
\begin{center}
  \begin{tikzpicture}
    [ arrows=-stealth 
    , thick 
    , y = {(-3.85mm,-3.85mm)}, z = {(0cm,1cm)} 
    , xscale = 2.2, yscale = 2.2 
    ]
    \node (000) at (0,0,0)  {$\III$} ; 
    \node (001) at (0,0,1)  {$\IIS$} ; 
    \node (100) at (1,0,0)  {$\SII$} ; 
    \draw [\myred]          (000) -- (100) ;
    \draw [\myblue]         (000) -- (001) ;
  \end{tikzpicture}\hspace{.1\linewidth}
  \begin{tikzpicture}
    [ arrows=-stealth 
    , thick 
    , y = {(-3.85mm,-3.85mm)}, z = {(0cm,1cm)} 
    , xscale = 2.2, yscale = 2.2 
    ]
    \node (000) at (0,0,0)  {$\WHNF$};
    \node (001) at (0,0,1)  {$\WNF$} ;
    \node (100) at (1,0,0)  {$\HNF$} ;
    \draw [\myred]          (100) -- (000) ;
    \draw [\myblue]         (001) -- (000) ;
  \end{tikzpicture}
\end{center}
\vspace{-2em}\qedhere
\end{proof}
\begin{prop}
  Absorption does not hold with $\SIS$.
\end{prop}
\begin{proof}
  A counter-example term is $(\lambda k.k\OMEGA)(\lambda x.y)$. Indeed, $\SIS$
  is non-weak and non-head, and diverges on that term whereas the other
  strategies in the `L' are either weak or head and deliver $y$.
\end{proof}
\begin{prop}
  Absorption does not hold between strategies that differ in strictness.
\end{prop}
\begin{proof}
  The counter-example term is $(\lambda x.y)\OMEGA$ on which a strict strategy
  diverges but a non-strict strategy converges.
\end{proof}
\begin{prop}
  Absorption does not hold on the strict face.
\end{prop}
\begin{proof}
  The counter-examples are easily constructed using weakness and headness.
  \begin{itemize}
  \item The counter-example that $\ao$ does not absorb $\ISI$, nor $\bv$, nor
    $\ho$, is the term $(\lambda x.y)(\lambda k.k\OMEGA)$. This is the
    counter-example for $\SIS$ above but with operator and operand flipped
    around. Clearly, $\ao$ diverges on the operand with $\OMEGA$. But $\ISI$,
    $\bv$, and $\ho$ are weak or head and do not evaluate it.
  \item The counter-example that $\bv$ does not absorb $\ISI$ is $(\lambda
    x.y)(x\Omega)$.
  \item The counter-example that $\ho$ does not absorb $\ISI$ is
    $(\lambda x.y)(\lambda x.\Omega)$. \qedhere
  \end{itemize}
\end{proof}

\subsection{Structuring the hybrid space}
\label{sec:regimentation-hybrid}
The following table, extracted from the instantiation table in
Figure~\ref{fig:generic-eval-apply-3}
(page~\pageref{fig:generic-eval-apply-3}), collects the hybrid evaluators
discussed in the survey. The calls to subsidiaries are highlighted with the
remaining parameters greyed out.
\begin{displaymath}
\centering
  \begin{array}[t]{lccccc}
    \hline
    \text{Name} & \la & \opone & \arone & \optwo & \artwo \\
    \hline
    \text{Normal order}~(\no) &
                \D{\no} & \bn & \D{\id} & \D{\no} & \D{\no} \\
    \text{Head reduction}~(\hr) &
                \D{\hr} & \bn & \D{\id} & \D{\hr} & \D{\id} \\
    \text{Strict normalisation}~(\sn)&
                \D{\sn} & \bv & \bv     & \D{\sn} & \D{\sn} \\
    \text{Hybrid normal order}~(\hn) &
                \D{\hn} & \he & \D{\id} & \D{\hn} & \D{\hn} \\
    \text{Hybrid applicative order}~(\ha) &
                \D{\ha} & \bv & \underline{\D{\ha}} & \D{\ha} & \D{\ha} \\
    \text{Ahead machine}~(\am) &
                \D{\am} & \bv & \bv     & \D{\am} & \bv \\
    \text{Spine applicative order}~(\so) &
                \D{\so} & \ho & \underline{\D{\so}} & \D{\so} & \D{\so} \\
    \text{Balanced spine applicative order}~(\bs) &
                \D{\bs} & \ho & \ho & \D{\bs} & \D{\bs} \\
    \hline \\
  \end{array}
\end{displaymath}
The only unbalanced hybrid evaluators in the survey are hybrid applicative
order ($\ha$) and spine applicative order ($\so$), which call themselves
recursively in $\arone$ (underlined in the table). Any strategy could be used
on any of the five parameters but in practice the same uniform subsidiary is
used on some parameters to satisfy the desired properties (completeness,
spineness, etc.). Since Definition~\ref{def:nature}(ii) implies
Definition~\ref{def:dependency}, it is always possible to fix a subsidiary
strategy on which a hybrid strategy depends. In particular, the subsidiary
evaluator is called to evaluate operators in $\opone$, and the instantiations
for the other parameters must be such that the hybrid evaluator evaluates more
redexes than the subsidiary. The following provisos gleaned from the table
characterise hybrid-style eval-apply evaluators (`hybrid evaluators', for
short). The provisos assume that the subsidiary uniform evaluator upholds
provisos \ref{prov:U-opone-optwo} and \ref{prov:U-la-arone-artwo}.
\begin{enumerate}[label=(\textsc{h}$_{\arabic*}$)]
\item \label{prov:H-opone-optwo} On the shared $\opone$ parameter of
  \textsc{con} and \textsc{neu}, the hybrid evaluator calls the subsidiary
  evaluator. On \textsc{neu}'s $\optwo$ parameter, the hybrid evaluator calls
  itself recursively.

\item \label{prov:H-la-artwo} On the $\la$ and $\artwo$ parameters, the hybrid
  evaluator must evaluate strictly more than the subsidiary evaluator, namely,
  at least one of the parameters must be a call to the hybrid, and there must
  be at least one call to the subsidiary or to the hybrid where the homonymous
  parameter in the subsidiary is the identity.

\item \label{prov:H-arone} On the $\arone$ parameter, the hybrid evaluator
  calls the identity when non-strict, and evaluates at least as much as the
  subsidiary when strict. For balanced hybrid evaluators, the $\arone$
  parameter has the same value as the homonymous parameter of the subsidiary.
\end{enumerate}
Evaluating `more' and `at least as much' is succinctly expressed by $\{ \id \}
\leq \{ \id, \su, \hy \}$ and $\{ \su \} \leq \{ \su, \hy \}$ where the left
sets have the permissible values of an homonymous subsidiary parameter and the
right sets have the permissible values of the hybrid parameter given that
homonymous subsidiary parameter.

For illustration, proviso \ref{prov:H-la-artwo} holds when the hybrid
evaluator calls itself on one of $\la$, $\artwo$ where the subsidiary calls
$\id$, and on the other parameter it has the same value as the homonymous
parameter of the subsidiary. For example, ahead machine ($\am$) calls itself
on $\la$ whereas its subsidiary call-by-value ($\bv$) calls the identity on
$\la$, and ahead machine calls $\bv$ on $\artwo$ exactly as $\bv$ does on its
$\artwo$. Proviso \ref{prov:H-la-artwo} also holds when the hybrid calls the
subsidiary on one of $\la$, $\artwo$ where the subsidiary calls $\id$, and on
the other parameter it calls itself recursively. For example, normal order
($\no$) calls itself on $\la$ whereas its subsidiary call-by-name ($\bn$)
calls the identity, and normal order calls itself recursively on $\artwo$.

The provisos are for evaluators, that is, for hybrid \emph{style}. As we
anticipated in Section~\ref{sec:uniform-vs-hybrid} and will show in
Section~\ref{sec:hy-systematic}, a hybrid evaluator that upholds all the
provisos can still define a uniform strategy. The aim of the provisos is not
to match evaluators and strategies. The provisos express obvious syntactic and
reasonable operational restrictions under which a hybrid-style evaluator can
define a hybrid strategy. The provisos are gleaned from the evaluators in the
survey, with an eye on the connection with the eval-readback provisos to prove
one-step equivalence, including the cases when the eval-readback and the
hybrid eval-apply evaluator define the same uniform strategy. (We compare the
provisos for both styles after the following proposition.)
\begin{prop}
  Let $\hy$ be a hybrid eval-apply evaluator that uses a uniform eval-apply
  evaluator $\su$ that upholds provisos \textnormal{\ref{prov:U-opone-optwo}}
  \textnormal{\ref{prov:U-la-arone-artwo}}. The $\hy$ evaluator cannot define
  a hybrid strategy according to Definitions~\ref{def:nature}(ii)
  and~\ref{def:dependency} if it does not uphold the hybrid provisos
  \textnormal{\ref{prov:H-opone-optwo}}, \textnormal{\ref{prov:H-la-artwo}}
  \textnormal{\ref{prov:H-arone}}. If the strategy is balanced then $\hy$
  upholds the `balanced' part of proviso \textnormal{\ref{prov:H-arone}}.
\end{prop}
\begin{proof}
  In contrapositive form, if the intended strategy is hybrid then the
  evaluator upholds the provisos. The $\su$ evaluator defines a uniform
  strategy by Proposition~\ref{prop:uniform-evaluator-strategy}.  By proviso
  \ref{prov:U-opone-optwo}, the hybrid strategy depends on $\su$ because it is
  used (at least) on the $\opone$ parameter, and there is a derivation tree
  stacked on that premise with only subsidiary contexts.  By the context
  inclusion property of the hybrid strategy, it evaluates more redexes than
  the subsidiary. This is satisfied by calling the hybrid evaluator in
  $\optwo$ and by proviso \ref{prov:H-la-artwo} and the first part of
  \ref{prov:H-arone} which cover all possible stackings of derivation trees
  and contexts where $\hy$ evaluates more than $\su$. If the hybrid strategy
  is balanced, the set of permissible redexes must be the same as $\su$'s.
  This is satisfied by $\hy$ in its use of $\su$ in $\opone$ to uphold proviso
  \ref{prov:H-opone-optwo} and by the balanced part of
  \ref{prov:H-arone}.
\end{proof}
A comparison between the hybrid provisos and the eval-readback provisos of
Section~\ref{sec:eval-readback:provisos} provides intuitions for the role of
the provisos in determining the conditions for one-step equivalence.
\begin{itemize}
\item The readback evaluator comes after the eval evaluator, and readback
  evaluates more than eval. The hybrid evaluator comes after the subsidiary
  evaluator (on operators), and the hybrid evaluates more than the subsidiary.
\item Eval calls itself on $\opone$. The subsidiary calls itself on
  $\opone$. Readback calls itself on $\optwo$. The hybrid calls itself on
  $\optwo$.
\item The provisos \ref{prov:ER-la-artwo} and
  \ref{prov:H-la-artwo}\,\ref{prov:H-arone} are related: a balanced hybrid
  evaluator uses the same parameter as its subsidiary for $\arone$, and in the
  other two parameters ($\la$ and $\artwo$) one of the following must be the
  case:
  \begin{itemize}
  \item The hybrid calls the identity, or the subsidiary, or itself, if the
    subsidiary calls the identity in the homonymous parameter.
  \item The hybrid calls the subsidiary or itself, if the subsidiary calls
    itself in the homonymous parameter.
  \end{itemize}
\end{itemize}

\subsubsection{A systematic notation for hybrid evaluators}
\label{sec:hy-systematic}
To continue the analysis of hybrid evaluators without explicit reference to
the template and instantiation table of Figure~\ref{fig:generic-eval-apply-3}
(page~\pageref{fig:generic-eval-apply-3}), we introduce a convenient
systematic notation $\stgy{\R{X}\E{X}\B{X}} \hyb \stgy{\R{Y}\E{Y}\B{Y}}$ for
encoding hybrid evaluators by means of two triples. The right triple
$\stgy{\R{Y}\E{Y}\B{Y}}$ is the Beta Cube triple of the uniform evaluator used
as a subsidiary (table in Section~\ref{sec:sub:cube},
page~\pageref{sec:sub:cube:table}). We may on occasion abbreviate it using the
uniform's name when it has one (\eg~$\bv$). The left triple
$\stgy{\R{X}\E{X}\B{X}}$ encodes the behaviour of the hybrid evaluator on its
$\la$, $\arone$, and $\artwo$ parameters with the following values: $\stgy{I}$
for call to the identity, $\stgy{S}$ for call to the subsidiary, and
$\stgy{H}$ for recursive call to the hybrid. Conveniently, $\stgy{S}$ can be
read as `subsidiary' across the encoding because in the right triple it stands
for `self' which is a recursive call of the subsidiary. The diagram in
Figure~\ref{fig:hy-systematic} summarises. As an aid to the reader,
Figure~\ref{fig:encoding-to-natural-semantics-example} shows how to obtain the
natural semantics of a hybrid evaluator from its systematic notation by
instantiating the eval-apply evaluation template of
Figure~\ref{fig:generic-eval-apply-3}. The example is for the hybrid evaluator
${\stgy{\stgy{\R{H}\E{I}\B{S}}}}\hyb \stgy{\stgy{\R{S}\E{I}\B{I}}}$ which is
not discussed in the survey. We invite readers to take a couple of notations
from the above table, obtain the natural semantics in the same way as
illustrated in Figure~\ref{fig:encoding-to-natural-semantics-example}, and
compare the result (modulo evaluator names) to the natural semantics in the
survey.

\begin{figure}[htb]
  \centering
  \begin{displaymath}
    \centering
    \begin{array}{ccccc}
      \multicolumn{5}{c}{\overbrace{\rule{5cm}{0pt}}^{\mathrm{hybrid~parameters}}} \\
      \la & \arone & \artwo & & \opone \\
      \hline
      \stgy{\R{X}} & \stgy{\E{X}} & \stgy{\B{X}} & \hyb &
      \underbrace{\stgy{\R{Y}}~~~~~\stgy{\E{Y}}~~~~~\stgy{\B{Y}}}_{\mathrm{uniform~parameters}} \\
    \end{array}
    \hspace{1cm}
    \begin{array}{ccl}
      \stgy{X} & \in & \{ \stgy{I}, \stgy{S}, \stgy{H} \} \\
      \stgy{Y} & \in & \{ \stgy{I}, \stgy{S} \}
    \end{array}
  \end{displaymath}
  \vspace{0.2cm}
  \begin{displaymath}
    \centering
    \text{Fixed parameters}\left\{
      \begin{array}{lll}
        \mathrm{Hybrid:}  & \optwo  = \stgy{H} &
        \text{by proviso \ref{prov:H-opone-optwo}} \\
        \\
        \mathrm{Uniform:} & \opone = \stgy{S} &
        \text{by proviso \ref{prov:U-opone-optwo}} \\
        & \optwo =  \stgy{I} &
        \text{by proviso \ref{prov:U-opone-optwo}}
      \end{array}\right.
  \end{displaymath}
  \caption{Systematic notation for encoding hybrid evaluators.}
  \label{fig:hy-systematic}
\end{figure}

\begin{figure}
  \centering
  \def\HY{\stgy{HIS}} %
  \def\SU{\stgy{SII}} %
  \def\HYSU{(\HY\!\hyb\!\SU)} %
{\small\begin{mathpar} %
  \inferrule*[left=abs] %
  {\R{\SU}(B) = B'}%
  {\SU(\lambda x.B) = \lambda x.B'} %
  \and %
  \inferrule*[left=con] %
  {\SU(M) = \lambda x.B \quad \stgy{\E{id}}(N) = N \quad \SU(\cas{N}{x}{B})=B'} %
  {\SU(MN) = B'} %
  \and %
  \inferrule*[left=neu] %
  {\SU(M) = M' \quad M' \not\equiv \lambda x.B \quad \id(M') = M' \quad
    \stgy{\B{id}}(N) = N} %
  {\SU(MN) = M'N} %
  \\
  \\
  \inferrule*[left=abs] %
  {\R{\HYSU} (B) = B'}%
  {\HYSU (\lambda x.B) = \lambda x.B'} %
  \and %
  \inferrule*[left=con] %
  {\SU (M) = \lambda x.B \quad \stgy{\E{id}}(N) = N \quad \HYSU
    (\cas{N}{x}{B})=B'} %
  {\HYSU (MN) = B'} %
  \and %
  \inferrule*[left=neu] %
  {\SU(M) = M' \quad M' \not\equiv \lambda x.B \quad \HYSU (M') = M''
    \quad \B{\SU} (N) = N'} %
  {\HYSU(MN) = M''N'} %
\end{mathpar}}
\caption{Natural semantics obtained by instantiating $\stgy{\R{H}\E{I}\B{S}}
  \hyb \stgy{\stgy{\R{S}\E{I}\B{I}}}$ on the eval-apply evaluation template of
  Figure~\ref{fig:generic-eval-apply-3}. The instantiations of parameters are
  coloured according to the colour of the parameter.}
\label{fig:encoding-to-natural-semantics-example}
\end{figure}%

The following table lists the encodings for the hybrid evaluators in the
survey. We omit the colours.
\begin{displaymath}
\centering
  \begin{array}[t]{lcr}
    \hline
    \text{Name} & \text{Systematic} & \text{Abbreviated} \\
    \hline
    \text{Normal order}~(\no)   & \HIH \hyb \stgy{III} & \HIH \hyb \bn \\
    \text{Head reduction}~(\hr) & \HIH \hyb \stgy{III} & \HII \hyb \bn \\
    \text{Strict normalisation}~(\sn) & \HSH \hyb \stgy{ISS} & \HSH \hyb \bv \\
    \text{Hybrid normal order}~(\hn) & \HIH \hyb \stgy{SII} & \HIH \hyb \he \\
    \text{Hybrid applicative order}~(\ha) & \HHH \hyb \stgy{ISS} & \HHH \hyb \bv \\
    \text{Ahead machine}~(\am) & \HSS \hyb \stgy{ISS} & \HSS \hyb \bv \\
    \text{Spine applicative order}~(\so) & \HHH \hyb \stgy{SSI} & \HHH \hyb \ho \\
    \text{Balanced spine applicative order}~(\bs) & \HSH \hyb \stgy{SSI} &
    \HSH \hyb \ho \\
    \hline \\
  \end{array}
\end{displaymath}
We can use the notation to name hybrid evaluators for strategies not shown in
the survey. For example, the following evaluators use the subsidiary on the
$\la$ parameter unlike the evaluators in the survey, which call themselves
recursively on that parameter.
\begin{displaymath}
  \begin{array}[t]{cccccc}
    \hline
    \text{Systematic} & \la & \opone & \arone & \optwo & \artwo \\
    \hline
    \SIH\hyb\he & \he & \he & \D{\id} &
                \D{\SIH\hyb\SII} & \D{\SIH\hyb\SII} \\
    \SSH\hyb\ho & \ho & \ho & \ho &
                \D{\SSH\hyb\SSI} & \D{\SSH\hyb\SSI} \\
    \hline
  \end{array}\\
\end{displaymath}
The notation is semantically restricted by the hybrid provisos. When the left
triple is identical to the right triple (\eg~$\SIS\hyb\SIS$) then proviso
\ref{prov:H-la-artwo} fails and the evaluator defined is actually uniform
(\eg~$\SIS$) because the evaluation trees differ only in the evaluator name at
the roots. Also, evaluators such as $\HIH\hyb\SIS$, $\HSI\hyb\stgy{SSI}$,
$\IHH\hyb\stgy{ISS}$, $\stgy{HHI}\hyb\stgy{SSI}$, etc., are spurious hybrid
evaluators because they do not evaluate more than the subsidiary: the hybrid
simply calls itself on the same parameters evaluated by the subsidiary and
does not evaluate more on the other parameters.

Proviso \ref{prov:H-arone} dictates that balanced hybrid evaluators have the
form $\stgy{\R{X}} \stgy{I} \stgy{\B{X}} \hyb \stgy{\R{Y}} \stgy{I}
\stgy{\B{Y}}$ (non-strict) and $\stgy{\R{X}} \stgy{S} \stgy{\B{X}} \hyb
\stgy{\R{Y}} \stgy{S} \stgy{\B{Y}}$ (strict), with the other two parameters
$\stgy{\R{X}}$ and $\stgy{\B{Y}}$ restricted by proviso~\ref{prov:H-la-artwo}.

We now show \label{pag:IIS-hybrid} the example anticipated in
Section~\ref{sec:uniform-vs-hybrid} of a hybrid evaluator that actually
defines a uniform strategy. The uniform strategy is $\IIS$, defined by the
uniform evaluator encoded by the notation. This strategy can be defined by the
one-step equivalent hybrid evaluator $\IIH \hyb \bn$.  The equivalence is easy
to see with the notation: the hybrid evaluator is constructed by moving
non-headness from the uniform triple $\stgy{II\B{S}}$ to a hybrid triple
$\stgy{II\B{H}}$, leaving a head uniform subsidiary ($\bn = \stgy{II\B{I}}$).
This is the same move shown in Section~\ref{sec:eval-readback:equiv} with the
eval-apply $\ea$ evaluator and the eval-readback $\rb \circ \bn$ evaluator.
The $\ea$ evaluator is exactly $\IIS$, and the eval-readback evaluator is
obtained by moving non-headness to the readback stage. Indeed, the hybrid
evaluator $\IIH \hyb \bn$ and the eval-readback evaluator $\rb \circ \bn$ of
Section~\ref{sec:eval-readback:equiv} are one-step equivalent
(Figure~\ref{fig:summary-end-figure-eval-readback}). Both define the uniform
strategy $\IIS$ which is better defined in uniform style.

Notice that there are shallow hybrid evaluators such as $\stgy{SIH} \hyb
\stgy{III}$ which uphold the provisos. This evaluator calls its subsidiary
$\bn$ on $\la$ for a shallow non-weak evaluation, and calls itself on $\artwo$
to perform that very shallow evaluation on operands of neutrals. Such shallow
hybrid evaluators are suitable non-uniform eval candidates for eval-readback
evaluators.

\subsubsection{(Non-)Equivalences within the style}
\label{sec:hyb:equiv}
With the systematic notation we can conveniently show (non-)equivalences among
hybrid evaluators. A simple example is the non-equivalence of $\HIS\hyb\III$
and $\HIS\hyb\IIS$. As their common left triple indicates, both hybrids use
the subsidiary to evaluate neutrals ($\stgy{HI\B{S}}$). But $\bn\ (=
\stgy{III})$ is head and $\IIS$ is non-head. A counter-example neutral that
disproves the equivalence is trivial: $x(xR)$ where $R$ is a redex. This redex
is evaluated by $\HIS\hyb\IIS$ but not by $\HIS\hyb\III$.


An equivalence can be obtained when the hybrid calls itself recursively on the
neutral instead of calling the subsidiary, that is, the left triple is not
$\stgy{HI\B{S}}$ but $\stgy{HI\B{H}}$. The hybrids are $\HIH\hyb\III$ (normal
order) and $\HIH\hyb\IIS$.  However, the equivalence is one-step modulo
commuting redexes because $\IIS$ is weak and non-head, so it evaluates
operands of neutrals that are not abstractions. A counter-example neutral is
$x(\lambda x.R_1)R_2R_3$ where $R_i$ are redexes. Normal order evaluates the
redexes from left to right whereas $\HIH\hyb\IIS$ calls $\IIS$ on the operator
$x(\lambda x.R_1)R_2$ and thus $R_2$ is the first evaluated redex.

As another example, we show the non-equivalence of normal order
($\HIH\hyb\III$) respectively with hybrid normal order ($\HIH\hyb\SII$) and
head reduction ($\HII\hyb\III$). In the first case, the encodings differ only
on the $\la$ value of the uniform triple ($\stgy{\R{I}II}$ vs
$\stgy{\R{S}II}$), which marks the difference between delivering a $\WHNF$ or
a $\HNF$ by the subsidiary. The hybrid triple $\HIH$ indicates that evaluation
continues (and in this case, completes) with non-weak and non-head evaluation
to $\NF$. In the second case of normal order vs head reduction, the common
subsidiary delivers a $\WHNF$ and the encodings differ only on the $\artwo$
parameter of the hybrid triple ($\stgy{HI\B{H}}$ vs $\stgy{HI\B{I}}$), which
marks the difference between delivering a $\NF$ or a $\HNF$ as a final result.

There are no equivalences on the strict space because a strict subsidiary is
used by the hybrid to evaluate operators and, being strict, the subsidiary
evaluates the operands of redexes within the operators. The weakness and
headness is tied-in with the evaluation of those operands. For example, all
the hybrids with respective subsidiaries $\ISI$ (head) and $\bv$
$(= \stgy{ISS})$ (non-head) are not equivalent. For instance $\HSS\hyb\ISI$
and $\HSS\hyb\stgy{ISS}$, or $\HSH\hyb\ISI$ and $\HSH\hyb\stgy{ISS}$ (strict
normalisation). A counter-example is easy to find by placing a redex $R$
within a neutral that is the operand of a redex that is the operator first
evaluated by the subsidiaries: $((\lambda x.B)(xR))N$.

\subsection{Structuring the eval-readback space}
\label{sec:regimentation-eval-readback}
The provisos for eval-readback evaluators have been introduced in
Section~\ref{sec:eval-readback:provisos}. We introduce a systematic notation
$\stgy{\R{X}\B{X}} \circ \stgy{\R{Y}\E{Y}\B{Y}}$ for encoding eval-readback
evaluators by means of a pair and a triple. The right triple
$\stgy{\R{Y}\E{Y}\B{Y}}$ is the Beta Cube triple of the uniform evaluator used
as eval. We may on occasion abbreviate it using the uniform's name when it has
one. The left pair $\stgy{\R{X}\B{X}}$ encodes the behaviour of the readback
on its $\la$ and $\artwo$ parameters with the following values: $\stgy{I}$ for
call to the `identity', $\Ev$ for call to `eval', $\Rb$ for call to
`readback', and $(\Rb\Ev)$ for the composition of readback after eval. Mind
the parenthesis: $(\Rb\Ev)$ is not a pair. Valid pairs are $\Ev\Rb$, $\Rb\Ev$,
and $(\Rb\Ev)\stgy{X}$ and $\stgy{X}(\Rb\Ev)$ for any valid value of
$\stgy{X}$.  The following table shows the encodings for the eval-readback
evaluators in Figure~\ref{fig:generic-readback-3} (colours omitted).
\begin{displaymath}
  \begin{array}[t]{lrll}
    \hline
    \text{Name} & \multicolumn{3}{c}{\text{Systematic}} \\
    \hline
    \mathtt{byValue}        & (\Rb\Ev)\Rb      & \circ & \bv \\
    \mathtt{byName}         & \Rb(\Rb\Ev)      & \circ & \he \\
    \text{Normal order}     & (\Rb\Ev)(\Rb\Ev) & \circ & \bn \\
    \text{Ahead machine}    & (\Rb\Ev)\I       & \circ & \bv \\
    \mathit{unnamed}        & \Ev\Rb           & \circ & \bv \\
    \hline
    \hline \\
  \end{array}
\end{displaymath}
The proviso \ref{prov:ER-la-artwo} dictates that one element $\stgy{X}$ of the
readback pair must be $\Ev$ or $(\Rb\Ev)$. By unfolding the proviso we obtain
some new eval-readback evaluators:
\begin{itemize}
\item One element of the pair is $(\Rb\Ev)$ and the other is $\I$. For
  instance, $\I(\Rb\Ev)\circ\bn$, $(\Rb\Ev)\I\circ\bn$, and
  $(\Rb\Ev)\I\circ\bv$.
\item One element of the pair is $\Ev$ and the other is $\Rb$ when the
  corresponding premises in eval are $\I$. For instance, $\Ev\Rb\circ\IIS$,
  $\Rb\Ev\circ\he$, and $\Ev\Rb\circ\bv$.
\end{itemize}
The $\I(\Rb\Ev)\circ\bn$ and $(\Rb\Ev)\I\circ\bn$ evaluators are respectively
one-step equivalent to the uniform evaluators $\IIS$ and head reduction
$(\hr)$. The $(\Rb\Ev)\I\circ\bv$ evaluator is one-step equivalent to the
hybrid evaluator for the ahead machine $(\am)$. These and more equivalences
are collected in Figure~\ref{fig:summary-end-figure-eval-readback}
(page~\pageref{fig:summary-end-figure-eval-readback}) and follow from the LWF
proof given in the following section.

\section{Equivalence proof by LWF transformation}
\label{sec:hybrid-evalreadback}
Our equivalence proof uses LWF to fuse the eval and readback stages into the
single-function balanced hybrid eval-apply evaluator. The proof does
\emph{not} apply LWF on a generic eval-readback evaluator. It applies LWF on a
\emph{plain but arbitrary} eval-readback evaluator where its uniform eval
evaluator is a fixed-point of generic \texttt{gen\_eval\_apply}
(Figure~\ref{fig:generic-eval-apply-3},
page~\pageref{fig:generic-eval-apply-3}) and its readback evaluator is a
fixed-point of generic \texttt{gen\_readback}
(Figure~\ref{fig:generic-readback-3}, page~\pageref{fig:generic-readback-3}).
The arbitrariness is achieved by considering some parameters as `free'
(existentially quantified in the sense of logic programming) by giving them a
well-typed \texttt{undefined} value. The undefined identifiers provide
genericity in the form of arbitrariness without the need for the higher-order
parametrisation of generic evaluators. The provisos and the LWF steps impose
constraints on the possible concrete values that can be passed to the
parameters. The result of the transformation is the one-step equivalent
\emph{plain} balanced hybrid eval-apply evaluator and a set of constraints
that capture the possible substitutions. The subsidiary and the hybrid
evaluators are fixed-points of \texttt{gen\_eval\_apply}. The following
sections elaborate.

\subsection{Plain but arbitrary eval-readback evaluators}
\label{sec:plain-arbitrary-eval-readback}
Take the plain \texttt{byValue} evaluator of
Figure~\ref{fig:generic-readback-3} (page~\pageref{fig:generic-readback-3}):
\begin{code}
  eval_readback_byValue :: Red
  eval_readback_byValue =  bodies <=< bv
    where bv     :: Red
          bv     =  gen_eval_apply \R{id} bv \E{bv} id \B{bv}
          bodies :: Red
          bodies =  gen_readback \R{(bodies <=< bv)} \B{bodies}
\end{code}
Both \texttt{bv} and \texttt{bodies} are locally-defined Haskell identifiers.
The first is a fixed-point of \texttt{gen\_eval\_apply} that satisfies the
provisos~\ref{prov:U-opone-optwo} and~\ref{prov:U-la-arone-artwo} for uniform
evaluators. The second is a fixed-point of \texttt{gen\_readback} that
satisfies the provisos~\ref{prov:ER-optwo} and~\ref{prov:ER-la-artwo} for
readback evaluators.

We generalise to a plain but arbitrary eval-readback evaluator that implements
a strategy $\stgy{xx}$ by using type-checked but undefined identifiers for the
\emph{actual} parameters of \texttt{gen\_eval\_apply} and
\texttt{gen\_readback} that are not fixed by the provisos. The undefined
identifiers are named by combining the parameter's name with a letter
\texttt{e} and \texttt{r} to distinguish between the eval and the readback
parameter. They stand for \emph{actual, not formal}, parameters.
\begin{code}
  eval_readback_xx :: Red
  eval_readback_xx =  reb_xx <=< evl_xx
    where evl_xx :: Red
          evl_xx =  gen_eval_apply \R{lae_xx} evl_xx \E{ar1e_xx} id \B{ar2e_xx}
          reb_xx :: Red
          reb_xx =  gen_readback \R{lar_xx} \B{ar2r_xx}

  \R{lae_xx}   :: Red
  \R{lae_xx}   =  undefined
  ...
  \B{ar2r_xx} :: Red
  \B{ar2r_xx} =  undefined
\end{code}
The original \texttt{eval\_readback\_byValue} evaluator is obtained by editing
the code and rewriting \texttt{\R{lae\_xx}} to \texttt{\R{id}},
\texttt{\R{lar\_xx}} to \texttt{\R{(reb\_xx <=< evl\_xx)}}, etc.

The provisos for the uniform evaluator and the readback evaluator constrain
the possible values of the actual parameters. Let us transcribe the provisos
in code:
\label{prov:U-ER-code}
\begin{enumerate}[label=(\textsc{u}$_{\arabic*}$)]
\item The second and fourth actual parameters of \texttt{gen\_eval\_apply} are
  respectively \texttt{evl\_xx} and \texttt{id}.

\item The identifiers \texttt{\R{lae\_xx}}, \texttt{\E{ar1e\_xx}}, and
  \texttt{\B{ar2e\_xx}} stand for either \texttt{eval\_xx} or \texttt{id}.
\end{enumerate}
\begin{enumerate}[label=(\textsc{er}$_\arabic*$)]
\item Function \texttt{gen\_readback} only takes two actual parameters:
  \texttt{\R{lar\_xx}} and \texttt{\B{ar2e\_xx}}.

\item In the other two parameters \texttt{\R{lar\_xx}} and \texttt{\B{ar2r\_xx}}:
  \begin{itemize}
  \item[(rb-ev)] At least one stands for \texttt{evl\_xx} where the homonymous
    parameter (\texttt{\R{lae\_xx}} or \texttt{\B{ar2e\_xx}}) stands for
    \texttt{id}.
  \item[(rb-rb)] At least one stands for \texttt{reb\_xx}.
  \end{itemize}
\end{enumerate}
Since \texttt{reb\_xx} can only be called on a parameter after
\texttt{evl\_xx} has been already called on the term affected by that
parameter, readback may call the identity, itself, or the eval-readback
composition as specified by the following table:
\begin{flushleft}
  \begin{tabular}{|l|l|}
    \hline
    Eval value & Readback value \\
    \hline
    \texttt{id} & \texttt{id} or \texttt{eval\_xx} or
    \texttt{reb\_xx}~\texttt{<=<}~\texttt{evl\_xx} \\
    \texttt{evl\_xx} & \texttt{id} or \texttt{reb\_xx} \\
    \hline
  \end{tabular}
\end{flushleft}
If one of the parameters has no call to \texttt{evl\_xx} (resp.\
\texttt{reb\_xx}) then the other parameter must contain a call to
\texttt{evl\_xx} (resp.\ \texttt{reb\_xx}). These constraints will be
elaborated during the LWF transformation.

\subsection{Plain but arbitrary balanced hybrid eval-apply evaluators}
\label{sec:plain-arbitrary-balanced-hybrid}
Take the plain balanced hybrid `strict normalisation' evaluator of
Figure~\ref{fig:generic-eval-apply-3}
(page~\pageref{fig:generic-eval-apply-3}):
\begin{code}
  bal_hybrid_sn :: Red
  bal_hybrid_sn =  sn
    where bv :: Red
          bv =  gen_eval_apply \R{id} bv \E{bv} id \B{bv}
          sn :: Red
          sn =  gen_eval_apply \R{sn} bv \E{bv} sn \B{sn}
\end{code}
Both \texttt{bv} and \texttt{sn} are locally-defined Haskell identifiers. The
first has the same code as in Section~\ref{sec:plain-arbitrary-eval-readback}.
The second is a fixed-point of \texttt{gen\_eval\_apply} that satisfies the
provisos~\ref{prov:H-opone-optwo}, \ref{prov:H-la-artwo}, and
\ref{prov:H-arone} for balanced hybrid evaluators.

We generalise to a plain but arbitrary balanced hybrid evaluator that
implements a strategy $\stgy{xx}$ by using type-checked undefined identifiers
as before:
\begin{code}
  bal_hybrid_xx  :: Red
  bal_hybrid_xx  =  hyb_xx
    where sub_xx :: Red
          sub_xx =  gen_eval_apply \R{las_xx} sub_xx \E{ar1s_xx} id     \B{ar2s_xx}
          hyb_xx :: Red
          hyb_xx =  gen_eval_apply \R{lah_xx} sub_xx \E{ar1s_xx} hyb_xx \B{ar2h_xx}

  \R{las_xx}  :: Red
  \R{las_xx}  =  undefined
  ...
  \B{ar2h_xx} :: Red
  \B{ar2h_xx} =  undefined
\end{code}
The undefined identifiers are named by combining the parameter's name with a
letter \texttt{s} and \texttt{h} to distinguish between the subsidiary and the
hybrid parameter.

Let us transcribe in code \label{prov:U-H-code} the provisos for the uniform
evaluator and the balanced hybrid evaluator:
\begin{enumerate}[label=(\textsc{u}$_{\arabic*}$)]
\item The second and fourth actual parameters of the first
  \texttt{gen\_eval\_apply} are respectively \texttt{sub\_xx} and \texttt{id}.

\item The identifiers \texttt{\R{las\_xx}}, \texttt{\E{ar1s\_xx}}, and
  \texttt{\B{ar2s\_xx}} stand for either \texttt{sub\_xx} or \texttt{id}.
\end{enumerate}
\begin{enumerate}[label=(\textsc{h}$_{\arabic*}$)]
\item The second and fourth actual parameters of the second
  \texttt{gen\_eval\_apply} are respectively \texttt{sub\_xx} and
  \texttt{hyb\_xx}.

\item \label{prov:H2-code} At least one of \texttt{\R{lah\_xx}} and
  \texttt{\B{ar2h\_xx}} stand for \texttt{hyb\_xx}, and one must stand for
  \texttt{sub\_xx} or \texttt{hyb\_xx} where the homonymous parameter
  (\texttt{\R{las\_xx}} or \texttt{\B{ar2s\_xx}}) stands for \texttt{id}.

\item The third actual parameter of both \texttt{gen\_eval\_apply} functions
  is \texttt{\E{ar1s\_xx}} which stands for either \texttt{id} or
  \texttt{sub\_xx}.
\end{enumerate}
Let us unpack \ref{prov:H2-code} focusing on \texttt{\R{lah\_xx}}. The same
unpacking applies to \texttt{\B{ar2h\_xx}}. Let \texttt{\R{lah\_xx}} stand for
\texttt{hyb\_xx}. If \texttt{\R{las\_xx}} stands for \texttt{sub\_xx} then
\texttt{\B{ar2h\_xx}} may either stand for \texttt{sub\_xx} or for
\texttt{hyb\_xx}. If \texttt{\R{las\_xx}} stands for \texttt{id} then
\texttt{\B{ar2h\_xx}} may stand for the same as \texttt{\B{ar2s\_xx}}, or for
\texttt{sub\_xx}, or for \texttt{hyb\_xx}. In sum, for \texttt{\R{lah\_xx}}
and \texttt{\B{ar2h\_xx}} the hybrid may call the identity, the subsidiary, or
itself as specified by the following table:
\begin{flushleft}
  \begin{tabular}{|l|l|}
    \hline
    Subsidiary value & Hybrid value \\
    \hline
    \texttt{id} & \texttt{id} or \texttt{sub\_xx} or
    \texttt{hyb\_xx} \\
    \texttt{sub\_xx} & \texttt{sub\_xx} or \texttt{hyb\_xx} \\
    \hline
  \end{tabular}
\end{flushleft}

\subsection{Statement of one-step equivalence}
\label{sec:one-step-equiv-thm}
In Section~\ref{sec:lwf-steps} we show the LWF transformation of
\verb!eval_readback_xx! into \verb!bal_hybrid_xx! when they uphold their
respective provisos in code stated in
Sections~\ref{sec:plain-arbitrary-eval-readback}
(page~\pageref{prov:U-ER-code}) and \ref{sec:plain-arbitrary-balanced-hybrid}
(page~\pageref{prov:U-H-code}). At the end of the transformation we collect
the possible substitutions for the undefined parameters in a list of
equations.  The equations are also a means to produce one evaluator style
given actual code for the other.
\begin{prop}
  \label{thm:one-step-equiv-thm}
  Let \verb!rb! \verb!<=<!\ \verb!ev! be a concrete eval-readback evaluator
  defined by concretising the arbitrary parameters of \verb!eval_readback_xx!
  such that \verb!ev! and \verb!rb!  uphold their respective provisos in code
  (page~\pageref{prov:U-ER-code}).  Let \verb!hy! be a concrete balanced
  hybrid eval-apply evaluator defined by concretising the arbitrary parameters
  of \verb!bal_hybrid_xx! such that \verb!su! and \verb!hy! uphold their
  respective provisos in code (page~\pageref{prov:U-H-code}).

  The evaluators \verb!rb! \verb!<=<!  \verb!ev! and \verb!hy! are one-step
  equivalent (modulo commuting redexes when $\mathtt{\B{ar2e\_xx}}$ is not the
  identity) when the following equations hold for the concrete values provided
  for the undefined parameters in the equations:
  \begin{displaymath}
    \begin{array}{rcl}
      \R{\mathtt{lae\_xx}} &=& \R{\mathtt{las\_xx}}  \\
      \E{\mathtt{ar1e\_xx}} &=& \E{\mathtt{ar1s\_xx}} \\
      \B{\mathtt{ar2e\_xx}} &=& \B{\mathtt{ar2s\_xx}} \\
      \R{\mathtt{lar\_xx}\ \textnormal{\texttt{<=<}}\ \mathtt{lae\_xx}}
         &=& \R{\mathtt{lah\_xx}} \\
      \B{\mathtt{ar2r\_xx}\ \textnormal{\texttt{<=<}}\ \mathtt{ar2e\_xx}}
         &=& \B{\mathtt{ar2h\_xx}}
    \end{array}
  \end{displaymath}
\end{prop}
\begin{figure}[htb]
  \centering
  \small
  \begin{displaymath}
    \begin{array}{lrcl}
      \textup{\bf Normal order:} \quad & (\RE)(\RE)\circ\stgy{III} & &
                                     \HIH\hyb\stgy{III} \\
      &\R{\texttt{id}}  &=& \R{\texttt{id}} \\
      &\E{\texttt{id}}  &=& \E{\texttt{id}} \\
      &\B{\texttt{id}}  &=& \B{\texttt{id}} \\
      &\R{\texttt{(reb <=< evl) <=< id}} &=& \R{\texttt{hyb}} \\
      &\B{\texttt{(reb <=< evl) <=< id}} &=& \B{\texttt{hyb}} \\

      \multicolumn{4}{l}{%
        (\RE)(\RE)\circ\stgy{III}\ \textup{is one-step equivalent to}\
         \HIH\hyb\stgy{III}\ \textup{as}\ (\RE)\ \textup{after}\ \I\
       \textup{is equal to}\ \stgy{H}.} \\

      \\
      \textup{\bf Hybrid normal order:} \quad & \Rb(\RE)\circ\stgy{SII} & &
                                            \HIH\hyb\stgy{SII} \\
      & \R{\texttt{evl}}                  &=& \R{\texttt{sub}} \\
      & \E{\texttt{id}}                   &=& \E{\texttt{id}}  \\
      & \B{\texttt{id}}                   &=& \B{\texttt{id}}  \\
      & \R{\texttt{reb <=< evl}}          &=& \R{\texttt{hyb}} \\
      & \B{\texttt{(reb <=< evl) <=< id}} &=& \B{\texttt{hyb}} \\

      \multicolumn{4}{l}{%
        \Rb(\RE)\circ\stgy{SII}\ \textup{is one-step equivalent to}\
         \HIH\hyb\stgy{SII}\  \textup{as both}\ \stgy{R}\ \textup{after}\
         \stgy{S}\ \textup{and}\ \RE\ \textup{after}\ \I} \\
      \multicolumn{4}{l}{\textup{are equal to}\ \stgy{H}.} \\

      \\

      \textup{\bf Head reduction:} \quad & (\RE)\stgy{I}\circ\stgy{III} & &
                                       \HII\hyb\stgy{III} \\
      &\R{\texttt{id}}  &=& \R{\texttt{id}}\\
      &\E{\texttt{id}}  &=& \E{\texttt{id}}\\
      &\B{\texttt{id}}  &=& \B{\texttt{id}}\\
      &\R{\texttt{reb <=< evl}} &=& \R{\texttt{hyb}} \\
      &\B{\texttt{id <=< id}}   &=& \B{\texttt{id}} \\

      \multicolumn{4}{l}{%
        (\RE)\I \circ
        \stgy{III}\ \textup{is one-step equivalent to}\ \HII\hyb\stgy{III}\
        \textup{as}\ \RE\ \textup{is equal to}\ \stgy{H}.} \\

      \\

      \textup{\bf Strict normalisation:} \quad & (\RE)\Rb\circ\stgy{ISS} & &
                                             \HSH\hyb\stgy{ISS}\\
      &\R{\texttt{id}}  &=& \R{\texttt{id}} \\
      &\E{\texttt{evl}}  &=& \E{\texttt{sub}} \\
      &\B{\texttt{evl}}  &=& \B{\texttt{sub}} \\
      &\R{\texttt{(reb <=< evl) <=< id}}  &=& \R{\texttt{hyb}} \\
      &\B{\texttt{reb <=< evl}}           &=& \B{\texttt{hyb}} \\

      \multicolumn{4}{l}{%
        (\RE)\stgy{R}\circ\stgy{ISS}\ \textup{is one-step equivalent to}\
        \HSH\hyb\stgy{ISS}\ \textup{as both}\ \RE\ \textup{after}\ \I\
        \textup{and}\ \Rb\ \textup{after}\ \stgy{S}} \\
      \multicolumn{4}{l}{\textup{are equal to}\ \stgy{H}.} \\

      \\

      \textup{\bf Ahead machine:} \quad & (\RE)\stgy{I}\circ\stgy{ISS} & &
                                      \HSS\hyb\stgy{ISS} \\
      &\R{\texttt{id}} &=& \R{\texttt{id}}  \\
      &\E{\texttt{evl}} &=& \E{\texttt{sub}} \\
      &\B{\texttt{evl}} &=& \B{\texttt{sub}} \\
      &\R{\texttt{(reb <=< evl) <=< id}} &=& \R{\texttt{hyb}} \\
      &\B{\texttt{id <=< evl}}           &=& \B{\texttt{sub}} \\

      \multicolumn{4}{l}{%
        (\RE)\I\circ\stgy{ISS}\ \textup{is one-step equivalent to}\
         \HSS\hyb\stgy{ISS}\ \textup{as}\ \RE\ \textup{after}\ \I\
         \textup{is equal to}\ \stgy{H},\ \textup{and}} \\
      \multicolumn{4}{l}{\I\ \textup{after}\ \stgy{S}\
        \textup{is equal to}\ \stgy{S}.}
    \end{array}
  \end{displaymath}
  \caption{Equation instantiations for some example strategies.}
  \label{fig:inst-equations}
\end{figure}
Figure~\ref{fig:inst-equations} shows the instantiations of the equations, and
thus the evaluator equivalences, for some example strategies.
\begin{cor}
  $\mathtt{rb}$ \verb!<=<! $\mathtt{ev}$ $\mathtt{=}$ $\mathtt{hy}$.
\end{cor}
\begin{cor}
  The implemented strategies uphold $\ev = \su$ and $\rb \circ \ev = \hy$.
\end{cor}
\begin{cor}
  $\hy \circ \ev = \hy$ \quad $(\hy$ absorbs $\ev)$
\end{cor}
\begin{proof}
    \begin{align*}
      &\hy \circ \ev = \hy \\
      &\mathord{\Leftrightarrow}\
      \{\ \rb \circ \ev = \hy\ \} \\
      &(\rb \circ \ev) \circ \ev = \rb \circ \ev \\
      &\mathord{\Leftrightarrow}\ \{\ \mathrm{associativity}\ \} \\
      &\rb \circ (\ev \circ \ev) = \rb \circ \ev \\
      &\mathord{\Leftrightarrow}\ \{\ \ev \circ \ev = \ev\ \} \\
      &\rb \circ \ev = \rb \circ \ev \\
      &\mathord{\Leftrightarrow}\ \{\ \mathrm{reflexivity}\ \}\\
      &\mathrm{true} \qedhere
    \end{align*}
\end{proof}

\subsection{Proof by LWF transformation}
\label{sec:lwf-steps}
Starting with \texttt{eval\_readback\_xx} we fuse \texttt{evl\_xx} and
\texttt{reb\_xx} into a single function to arrive at \texttt{bal\_hybrid\_xx}
according to the following LWF-steps \cite{OS07}.
\begin{enumerate}
\newpage
\item Start with \texttt{eval\_readback\_xx}:
\begin{code}
  eval_readback_xx :: Red
  eval_readback_xx =  reb_xx <=< evl_xx
    where evl_xx :: Red
          evl_xx =  gen_eval_apply \R{lae_xx} evl_xx \E{ar1e_xx} id \B{ar2e_xx}
          reb_xx :: Red
          reb_xx =  gen_readback \R{lar_xx} \B{ar2r_xx}
\end{code}

\item Inline \texttt{evl\_xx} in
  \texttt{reb\_xx}~\texttt{<=<}~\texttt{evl\_xx} and beta-convert the
  composition attending to the cases of an input term \texttt{t}. The result
  is \texttt{eval\_readback\_xx1} below where the calls to \texttt{reb\_xx}
  are all in tail position:
  \begin{code}
  eval_readback_xx1 :: Red
  eval_readback_xx1 t = case t of
    v@(Var _) -> return v
    (Lam s b) -> do b' <- \R{lae_xx} b
                    l  <- return (Lam s b')
                    reb_xx l
    (App m n) -> do m' <- evl_xx m
                    case m' of (Lam s b) -> do n' <- \E{ar1e_xx} n
                                               s  <- evl_xx (subst n' s b)
                                               reb_xx s
                               _         -> do m'' <- id m'
                                               n'  <- \B{ar2e_xx} n
                                               a   <- return (App m'' n')
                                               reb_xx a
    where evl_xx :: Red
          evl_xx =  gen_eval_apply \R{lae_xx} evl_xx \E{ar1e_xx} id \B{ar2e_xx}
          reb_xx :: Red
          reb_xx =  gen_readback \R{lar_xx} \B{ar2r_xx}
  \end{code}

\item Inline the first and last occurrences of \texttt{reb\_xx} in
  \texttt{eval\_readback\_xx1}. The result is \texttt{eval\_readback\_xx2}
  below:
  \begin{code}
  eval_readback_xx2 :: Red
  eval_readback_xx2 t = case t of
    v@(Var _) -> return v
    (Lam s b) -> do b'  <- \R{lae_xx} b
                    b'' <- \R{lar_xx} b'
                    return (Lam s b'')
    (App m n) -> do m' <- evl_xx m
                    case m' of (Lam s b) -> do n' <- \E{ar1e_xx} n
                                               s  <- evl_xx (subst n' s b)
                                               reb_xx s
                               _         -> do m''  <- id m'
                                               n'   <- \B{ar2e_xx} n
                                               m''' <- reb_xx m''
                                               n''  <- \B{ar2r_xx} n'
                                               return (App m''' n'')
    where evl_xx :: Red
          evl_xx =  gen_eval_apply \R{lae_xx} evl_xx \E{ar1e_xx} id \B{ar2e_xx}
          reb_xx :: Red
          reb_xx =  gen_readback \R{lar_xx} \B{ar2r_xx}
  \end{code}
\newpage
\item \label{lwf:step-swap} In \texttt{eval\_readback\_xx2}, swap the lines
  \begin{code}
    n'   <- \B{ar2e\_xx} n
    m''' <- reb\_xx m''
  \end{code}
  \vspace{-10pt} This swapping is required by the LWF transformation and
  alters the evaluation order in neutrals: the operator is evaluated entirely
  and then the operand entirely, instead of alternatingly. One-step
  equivalence is preserved except for strategies where \texttt{\B{ar2e\_xx}}
  is not the identity. In that case one-step equivalence is modulo commuting
  redexes. We identify such strategies in Section~\ref{sec:conclusions}.

  After swapping, replace the binding \mbox{\texttt{m'' <- id m'}} by the
  binding \mbox{\texttt{m'' <- evl\_xx m'}} which has no effect because
  \texttt{m'} is the term delivered by \texttt{evl\_xx} on the neutral's
  operator \texttt{m}, and certainly \texttt{id}~$\circ$~\texttt{evl\_xx} =
  \texttt{evl\_xx}~$\circ$~\texttt{evl\_xx}.

  Finally, simplify the \texttt{do} notation by applying the right identity of
  monadic bind and also Kleisli composition. The result is
  \texttt{eval\_readback\_xx3} below:
  \begin{code}
  eval_readback_xx3 :: Red
  eval_readback_xx3 t = case t of
    v@(Var _) -> return v
    (Lam s b) -> do b' <- \R{(lar_xx <=< lae_xx)} b
                    return (Lam s b')
    (App m n) -> do m' <- evl_xx m
                    case m' of
                      (Lam s b) -> do n' <- \E{ar1e_xx} n
                                      (reb_xx <=< evl_xx) (subst n' s b)
                      _         -> do m''  <- (reb_xx <=< evl_xx) m'
                                      n'   <- \B{(ar2r_xx <=< ar2e_xx)} n
                                      return (App m'' n')
    where evl_xx :: Red
          evl_xx =  gen_eval_apply \R{lae_xx} evl_xx \E{ar1e_xx} id \B{ar2e_xx}
          reb_xx :: Red
          reb_xx =  gen_readback \R{lar_xx} \B{ar2r_xx}
  \end{code}

\item Extrude the body of \texttt{eval\_readback\_xx3} by adding a new binding
  named \texttt{template} that takes three parameters \texttt{\R{la}},
  \texttt{\E{ar1}} and \texttt{\B{ar2}}. The result is
  \texttt{eval\_readback\_xx4} below:
  \begin{code}
  eval_readback_xx4 :: Red
  eval_readback_xx4 =
    template \R{(lar_xx <=< lae_xx)} \E{ar1e_xx} \B{(ar2r_xx <=< ar2e_xx)}
    where
      evl_xx :: Red
      evl_xx =  gen_eval_apply \R{lae_xx} evl_xx \E{ar1e_xx} id \B{ar2e_xx}
      reb_xx :: Red
      reb_xx =  gen_readback \R{lar_xx} \B{ar2r_xx}
      template :: Red -> Red -> Red -> Red
      template \R{la} \E{ar1} \B{ar2} t = case t of
        v@(Var _) -> return v
        (Lam s b) -> do b' <- \R{la} b
                        return (Lam s b')
        (App m n) -> do m' <- evl_xx m
                        case m' of
                          (Lam s b) -> do n' <- \E{ar1} n
                                          (reb_xx <=< evl_xx) (subst n' s b)
                          _         -> do m''  <- (reb_xx <=< evl_xx) m'
                                          n'   <- \B{ar2} n
                                          return (App m'' n')
  \end{code}

\item Replace the two calls to \texttt{(reb\_xx <=< evl\_xx)} in
  \texttt{template} by a call to a new binding \texttt{this} that is bound to
  the top-level call to \texttt{template}. The result is
  \texttt{eval\_readback\_xx5} below:
  \begin{code}
  eval_readback_xx5 :: Red
  eval_readback_xx5 =  this
    where evl_xx :: Red
          evl_xx =  gen_eval_apply \R{lae_xx} evl_xx \E{ar1e_xx} id \B{ar2e_xx}
          reb_xx :: Red
          reb_xx =  gen_readback \R{lar_xx} \B{ar2r_xx}
          template :: Red -> Red -> Red -> Red
          template \R{la} \E{ar1} \B{ar2} t = case t of
            v@(Var _) -> return v
            (Lam s b) -> do b' <- \R{la} b
                            return (Lam s b')
            (App m n) -> do m' <- evl_xx m
                            case m' of
                              (Lam s b) -> do n' <- \E{ar1} n
                                              this (subst n' s b)
                              _         -> do m''  <- this m'
                                              n'   <- \B{ar2} n
                                              return (App m'' n')
          this :: Red
          this =  template \R{(lar_xx <=< lae_xx)} \E{ar1e_xx} \B{(ar2r_xx <=< ar2e_xx)}
  \end{code}

\item Promote \texttt{template} to a global function that keeps the local
  bindings for \texttt{evl\_xx}, \texttt{reb\_xx}, and \texttt{this}. The
  result is \texttt{eval\_readback\_xx6} below:
  \begin{code}
  template :: Red -> Red -> Red -> Red
  template \R{la} \E{ar1} \B{ar2} t = case t of
    v@(Var _) -> return v
    (Lam s b) -> do b' <- \R{la} b
                    return (Lam s b')
    (App m n) -> do m' <- evl_xx m
                    case m' of (Lam s b) -> do n' <- \E{ar1} n
                                               this (subst n' s b)
                               _         -> do m''  <- this m'
                                               n'   <- \B{ar2} n
                                               return (App m'' n')
    where evl_xx :: Red
          evl_xx =  gen_eval_apply \R{lae_xx} evl_xx \E{ar1e_xx} id \B{ar2e_xx}
          reb_xx :: Red
          reb_xx =  gen_readback \R{lar_xx} \B{ar2r_xx}
          this   :: Red
          this   =  template \R{(lar_xx <=< lae_xx)} \E{ar1e_xx} \B{(ar2r_xx <=< ar2e_xx)}

  eval_readback_xx6 :: Red
  eval_readback_xx6 =  this
    where evl_xx :: Red
          evl_xx =  gen_eval_apply \R{lae_xx} evl_xx \E{ar1e_xx} id \B{ar2e_xx}
          reb_xx :: Red
          reb_xx =  gen_readback \R{lar_xx} \B{ar2r_xx}
          this   :: Red
          this   =  template \R{(lar_xx <=< lae_xx)} \E{ar1e_xx} \B{(ar2r_xx <=< ar2e_xx)}
  \end{code}

\item Add formal parameters \texttt{op1} and \texttt{op2} to \texttt{template}
  so that it has the same parameters as \texttt{gen\_eval\_apply}, and rename
  \texttt{template} to \texttt{template1}. In the call to \texttt{template1}
  instantiate \texttt{op1} to \texttt{evl\_xx} and \texttt{op2} to
  \texttt{this}. Remove from \texttt{template1} the local bindings in the
  \texttt{where} clause which are now unused. The result is
  \texttt{eval\_readback\_xx7} below:
  \begin{code}
  template1 :: Red -> Red -> Red -> Red -> Red -> Red
  template1 \R{la} op1 \E{ar1} op2 \B{ar2} t = case t of
    v@(Var _) -> return v
    (Lam s b) -> do b' <- \R{la} b
                    return (Lam s b')
    (App m n) -> do m' <- op1 m
                    case m' of (Lam s b) -> do n' <- \E{ar1} n
                                               this (subst n' s b)
                               _         -> do m''  <- op2 m'
                                               n'   <- \B{ar2} n
                                               return (App m'' n')
    where this :: Red
          this =  template1 \R{la} op1 \E{ar1} op2 \B{ar2}

  eval_readback_xx7 :: Red
  eval_readback_xx7 =  this
    where
      evl_xx :: Red
      evl_xx =  gen_eval_apply \R{lae_xx} evl_xx \E{ar1e_xx} id \B{ar2e_xx}
      reb_xx :: Red
      reb_xx =  gen_readback \R{lar_xx} \B{ar2r_xx}
      this   :: Red
      this   =
        template1
          \R{(lar_xx <=< lae_xx)} evl_xx \E{ar1e_xx} this \B{(ar2r_xx <=< ar2e_xx)}
  \end{code}

\item Replace \texttt{template1} by \texttt{gen\_eval\_apply} because they
  have identical definitions. The result is \texttt{eval\_readback\_xx8}
  below:
  \begin{code}
  eval_readback_xx8 :: Red
  eval_readback_xx8 =  this
    where
      evl_xx :: Red
      evl_xx =  gen_eval_apply \R{lae_xx} evl_xx \E{ar1e_xx} id \B{ar2e_xx}
      reb_xx :: Red
      reb_xx =  gen_readback \R{lar_xx} \B{ar2r_xx}
      this   :: Red
      this   =
        gen_eval_apply
          \R{(lar_xx <=< lae_xx)} evl_xx \E{ar1e_xx} this \B{(ar2r_xx <=< ar2e_xx)}
  \end{code}
  The provisos in code for \texttt{eval\_readback\_xx}
  (page~\pageref{prov:U-ER-code}) constrain the possible values for
  \texttt{\R{lar\_xx}} which can only stand for:
  \begin{enumerate}
  \item The identity \texttt{id}. In this case, the parameter
    \R{\texttt{(lar\_xx <=< lae\_xx)}} either stands for \R{\texttt{(id <=<
        id)}} or for \R{\texttt{(id <=< evl\_xx)}}, each equivalent
    respectively to \R{\texttt{id}} and to \R{\texttt{evl\_xx}.}

  \item The value \texttt{reb\_xx} if \texttt{\R{lae\_xx}} stands for
    \texttt{evl\_xx}, or the value \texttt{(reb\_xx <=< evl\_xx)} if
    \texttt{\R{lae\_xx}} stands for \texttt{id}. Then, the parameter
    \R{\texttt{(lar\_xx <=< lae\_xx)}} stands for \texttt{(reb\_xx <=<
      evl\_xx)} or for \texttt{((reb\_xx <=< evl\_xx) <=< id)}, both
    equivalent to \texttt{(reb\_xx <=< evl\_xx)} and thus equivalent to a
    recursive call to \texttt{this}.
  \end{enumerate}
  It is therefore safe to replace \R{\texttt{lar\_xx}} in
  \texttt{eval\_readback\_xx8} by a new parameter \R{\texttt{laer\_xx}} that
  stands for
  \begin{itemize}
  \item \texttt{id} when both \R{\texttt{lar\_xx}} and \R{\texttt{lae\_xx}}
    are \texttt{id},
  \item \texttt{evl\_xx} when \R{\texttt{lar\_xx}} is \texttt{id} and
    \R{\texttt{lae\_xx}} is \texttt{evl\_xx},
  \item \texttt{this} when \R{\texttt{lar\_xx}} is either \texttt{reb\_xx} or
    \texttt{(reb\_xx <=< evl\_xx)},
  \end{itemize}
  which meets the provisos for \texttt{bal\_hybrid\_xx}
  (page~\pageref{prov:U-H-code}).

  The same discussion applies to \B{\texttt{(ar2r\_xx <=< ar2e\_xx)}}. It is
  therefore also safe to replace \B{\texttt{ar2e\_xx}} by a new parameter
  \B{\texttt{ar2er\_xx}}. \label{lwf:template-to-wtr}

\item By the discussion on the parameters in Step~(\ref{lwf:template-to-wtr}),
  replace \R{\texttt{lar\_xx}} and \B{\texttt{ar2r\_xx}} with
  \R{\texttt{laer\_xx}} and \B{\texttt{ar2er\_xx}} respectively. Remove the
  binding for \texttt{reb\_xx} which is now unused. The result is
  \texttt{eval\_readback\_xx9} below.
  \begin{code}
  eval_readback_xx9 :: Red
  eval_readback_xx9 =  this
    where evl_xx :: Red
          evl_xx =  gen_eval_apply \R{lae_xx}  evl_xx \E{ar1e_xx} id   \B{ar2e_xx}
          this   :: Red
          this   =  gen_eval_apply \R{laer_xx} evl_xx \E{ar1e_xx} this \B{ar2er_xx}
  \end{code}
  Compare it to \texttt{bal\_hybrid\_xx}. We copy-paste the definition here
  for convenience:
  \begin{code}
    bal_hybrid_xx  :: Red
    bal_hybrid_xx  =  hyb_xx
      where sub_xx :: Red
            sub_xx =  gen_eval_apply \R{las_xx} sub_xx \E{ar1s_xx} id     \B{ar2s_xx}
            hyb_xx :: Red
            hyb_xx =  gen_eval_apply \R{lah_xx} sub_xx \E{ar1s_xx} hyb_xx \B{ar2h_xx}
  \end{code}
  The code for \texttt{eval\_readback\_xx9} and \texttt{bal\_hybrid\_xx}
  coincide when:
  \begin{itemize}
  \item The identifiers \E{\texttt{ar1e\_xx}} and \E{\texttt{ar1s\_xx}} both
    stand for \texttt{id}, or for \texttt{eval\_xx}, the latter the same
    evaluator defined by the hybrid's subsidiary evaluator \texttt{sub\_xx}.

  \item Each of the remaining eval-readback parameters (\R{\texttt{lae\_xx}},
    \B{\texttt{ar2e\_xx}}, \R{\texttt{laer\_xx}} and \B{\texttt{ar2er\_xx}})
    has the same value as its homonymous hybrid parameters
    (\R{\texttt{las\_xx}}, \B{\texttt{ar2s\_xx}}, \R{\texttt{lah\_xx}} and
    \B{\texttt{ar2h\_xx}}) as explained in
    Step~(\ref{lwf:template-to-wtr}).

  \item The local variables \texttt{sub\_xx} and \texttt{hyb\_xx} are
    substituted respectively for \texttt{eval\_xx} and \texttt{this}. This is
    a mere syntactic substitution.
  \end{itemize}
  \label{lwf:cosmetic}
\end{enumerate}
The possible substitutions are collected in the following equations:
\begin{displaymath}
  \begin{array}{rcl}
    \R{\mathtt{lae\_xx}} &=& \R{\mathtt{las\_xx}} \\
    \E{\mathtt{ar1e\_xx}} &=& \E{\mathtt{ar1s\_xx}} \\
    \B{\mathtt{ar2e\_xx}} &=& \B{\mathtt{ar2s\_xx}} \\
    \R{\mathtt{lar\_xx}}\ \R{\textnormal{\texttt{<=<}}}\ \R{\mathtt{lae\_xx}}
      &=& \R{\mathtt{lah\_xx}} \\
    \B{\mathtt{ar2r\_xx}}\ \B{\textnormal{\texttt{<=<}}}\ \B{\mathtt{ar2e\_xx}}
      &=& \B{\mathtt{ar2h\_xx}}
  \end{array}
\end{displaymath}
The equations are not a part of the LWF steps, they collect the required
substitutions of arbitrary parameters in Step~\ref{lwf:template-to-wtr}. When
instantiating the arbitrary evaluators with concrete parameters, the equations
must be checked considering \texttt{reb\_xx~<=<~evl\_xx} and \texttt{hyb\_xx}
are one-step equivalent modulo commuting redexes (recall the examples in
Figure~\ref{fig:inst-equations}).

\section{Summary and conclusions}
\label{sec:conclusions}
Figure~\ref{fig:summary-end-figure-eval-apply} (page
\pageref{fig:summary-end-figure-eval-apply}) provides an extensive list of
uniform and hybrid eval-apply evaluators.
Figure~\ref{fig:summary-end-figure-eval-readback}
(page~\pageref{fig:summary-end-figure-eval-readback}) provides an extensive
list of eval-readback evaluators and their LWF-equivalent balanced hybrid
eval-apply evaluators where `mcr' abbreviates `modulo commuting redexes'.

\begin{figure}[htb]
\centering  \scriptsize
\begin{displaymath}
    \begin{array}{llllll}
      \text{Name} & \text{Systematic} & \text{Equivalent} &
         \text{Strictness} & \text{Classif.} & \text{Result} \\
      \hline \hline
      \text{Call-by-name}~(\bn) & \stgy{III} & & \text{non-strict} & \text{uniform}
                   & \WHNF \\
      & \stgy{IIS} & \IIH\hyb\III & \text{non-strict} & \text{uniform}
                   & \WNF \\
      \text{Head spine}~(\he) & \stgy{SII} & & \text{non-strict} & \text{uniform}
                   & \HNF \\
      & \stgy{SIS} & & \text{non-strict} & \text{uniform}
                   & \NF \\
      \hline
      & \stgy{ISI} & & \text{strict} & \text{uniform}
                   & \WHNF \\
      \text{Call-by-value}~(\bv) & \stgy{ISS} & & \text{strict} & \text{uniform}
                   & \WNF \\
      \text{Head applicative order}~(\ho) & \stgy{SSI} & & \text{strict}
                   & \text{uniform} & \HNF \\
      \text{Applicative order}~(\ao) & \stgy{SSS} & & \text{strict} & \text{uniform}
                   & \NF \\
      \hline \hline
      & \IIH\hyb\III & \IIS & \text{non-strict} &\text{uniform}
                   & \WNF \\
      & \SIH\hyb\III & & \text{non-strict} & \text{hybrid bal.}
                   & \WNF \\
      \text{Head reduction}~(\hr) & \HII\hyb\III & & \text{non-strict}
                   & \text{hybrid bal.} & \HNF \\
      & \HIS\hyb\III & & \text{non-strict} & \text{hybrid bal.}
                   & \NT{VHNF} \\
      \text{Normal order}~(\no) & \HIH\hyb\III & \HIH\hyb\IIS~\text{(mcr)}
                   & \text{non-strict} & \text{hybrid bal.}
                   & \NF \\
      \hline
      & \SIH\hyb\IIS & & \text{non-strict} & \text{hybrid bal.}
                   & \WNF \\
      & \HIS\hyb\IIS & & \text{non-strict} & \text{hybrid bal.}
                   & \NT{VHNF} \\
      & \HIH\hyb\IIS & \no~\text{(mcr)} & \text{non-strict} & \text{hybrid bal.}
                   & \NF \\
      \hline
      & \SIH\hyb\SII &
                   & \text{non-strict} & \text{hybrid bal.}
                   & \WNF \\
      & \HIS\hyb\SII & & \text{non-strict} & \text{hybrid bal.}
                   & \HNF \\
      \text{Hybrid normal order}~(\hn) & \HIH\hyb\SII & & \text{non-strict} &
                   \text{hybrid bal.}
                   & \NF \\
      \hline \hline
      & \ISH\hyb\ISI &  & \text{strict} & \text{hybrid bal.}
                   & \WNF \\
      & \SSH\hyb\ISI &
                   & \text{strict} & \text{hybrid bal.}
                   & \WNF \\
      & \HSI\hyb\ISI & & \text{strict} & \text{hybrid bal.}
                   & \HNF \\
      & \HSS\hyb\ISI & & \text{strict} & \text{hybrid bal.}
                   & \NT{VHNF} \\
      & \HSH\hyb\ISI & & \text{strict} & \text{hybrid bal.}
                   & \NF \\
      \hline
      & \SSH\hyb\ISS & & \text{strict} & \text{hybrid bal.}
                   & \WNF \\
      \text{Ahead machine}~(\am) & \HSS\hyb\ISS &
                   & \text{strict} & \text{hybrid bal.}
                   & \NT{VHNF} \\
      \text{Strict normalisation}~(\sn) & \HSH\hyb\ISS &
                   & \text{strict} & \text{hybrid bal.}
                   & \NF \\
      \hline
       & \SSH\hyb\SSI &
                   & \text{strict} & \text{hybrid bal.}
                   & \WNF \\
      & \HSS\hyb\SSI & & \text{strict} & \text{hybrid bal.}
                   & \HNF \\
      \text{Balanced spine applicative order}~(\bs) & \HSH\hyb\SSI
                   & & \text{strict} & \text{hybrid bal.}
                   & \NF \\
      \hline \hline
      & \IHH\hyb\ISI & & \text{strict} & \text{hybrid}
                   & \WNF \\
      & \SHH\hyb\ISI &
                   & \text{strict} & \text{hybrid}
                   & \WNF \\
      & \HHI\hyb\ISI & & \text{strict} & \text{hybrid}
                   & \HNF \\
      & \HHS\hyb\ISI &
                   & \text{strict} & \text{hybrid}
                   & \NT{VHNF} \\
      & \HHH\hyb\ISI &
                   & \text{strict} & \text{hybrid}
                   & \NF \\
      \hline
      & \SHH\hyb\ISS & & \text{strict} & \text{hybrid}
                   & \WNF \\
      & \HHS\hyb\ISS &
                   & \text{strict} & \text{hybrid}
                   & \NT{VHNF} \\
      \text{Hybrid applicative order}~(\ha) & \HHH\hyb\ISS &
                   & \text{strict} & \text{hybrid}
                   & \NF \\
      \hline
      & \SHH\hyb\SSI &
                   & \text{strict} & \text{hybrid}
                   & \WNF \\
      & \HHS\hyb\SSI & & \text{strict} & \text{hybrid}
                   & \HNF \\
      \text{Spine applicative order}~(\so) & \HHH\hyb\SSI &
                   & \text{strict} & \text{hybrid}
                   & \NF \\
      \hline\hline
  \end{array}
\end{displaymath}
\caption{List of uniform and hybrid eval-apply evaluators.}
\label{fig:summary-end-figure-eval-apply}
\end{figure}

\begin{figure}[htbp]
\centering  \scriptsize
\begin{displaymath}
    \begin{array}[t]{lllll}
      \text{Name} & \text{Systematic} & \text{LWF-equivalent eval-apply} &
                   \text{Intermediate} & \text{Result}\\
      \hline \hline
      & \I(\Rb\Ev)\circ\III & \IIH\hyb\III\ (=\IIS) & \WHNF & \WNF \\
      & \Ev(\Rb\Ev)\circ\III & \SIH\hyb\III & \WHNF & \WNF \\
      \text{Head reduction} & (\Rb\Ev)\I\circ\III & \hr & \WHNF & \HNF \\
      & (\Rb\Ev)\Ev\circ\III & \HIS\hyb\III & \WHNF & \NT{VHNF}\\
      \text{Normal order} & (\Rb\Ev)(\Rb\Ev)\circ\III & \no & \WHNF & \NF \\
      \hline
      & \Ev\Rb\circ\IIS & \SIH\hyb\IIS~\text{(mcr)} & \WNF & \WNF \\
      & (\Rb\Ev)\I\circ\IIS & \HIS\hyb\IIS~\text{(mcr)} & \WNF & \NT{VHNF} \\
      & (\Rb\Ev)\Rb\circ\IIS & \HIH\hyb\IIS~\text{(mcr)} & \WNF & \NF \\
      \hline
      & \I(\Rb\Ev)\circ\SII & \SIH\hyb\SII & \HNF & \WNF \\
      & \Rb\Ev\circ\SII & \HIS\hyb\SII & \HNF & \HNF \\
      \texttt{byName} & \Rb(\Rb\Ev)\circ\SII & \hn & \HNF & \NF \\
      \hline \hline
      & \I(\Rb\Ev)\circ\ISI & \ISH\hyb\ISI & \WHNF & \WNF \\
      & \Ev(\Rb\Ev)\circ\ISI & \SSH\hyb\ISI & \WHNF & \WNF \\
      & (\Rb\Ev)\I\circ\ISI & \HSI\hyb\ISI & \WHNF & \HNF \\
      & (\Rb\Ev)\Ev\circ\ISI & \HSS\hyb\ISI & \WHNF & \NT{VHNF} \\
      & (\Rb\Ev)(\Rb\Ev)\circ\ISI & \HSH\hyb\ISI & \WHNF & \NF \\
      \hline
      & \Ev\Rb\circ\ISS & \SSH\hyb\ISS~\text{(mcr)} & \WNF & \WNF \\
      \text{Ahead machine} & (\Rb\Ev)\I\circ\ISS & \am~\text{(mcr)} & \WNF & \NT{VHNF} \\
      \texttt{byValue} & (\Rb\Ev)\Rb\circ\ISS & \sn~\text{(mcr)} & \WNF & \NF \\
      \hline
      & \I(\Rb\Ev)\circ\SSI & \SSH\hyb\SSI & \HNF & \WNF \\
      & \Rb\Ev\circ\SSI & \HSS\hyb\SSI & \HNF & \HNF \\
      \text{Balanced spine applicative order} & \Rb(\Rb\Ev)\circ\SSI & \bs
        & \HNF & \NF \\
      \hline \hline
    \end{array}
  \end{displaymath}
  \caption{List of eval-readback evaluators and their one-step equivalent
    eval-apply evaluators.}
\label{fig:summary-end-figure-eval-readback}
\end{figure}

There are $8$ uniform eval-apply evaluators and $33$ (four times more) hybrid
eval-apply evaluators. Of the latter, one evaluator defines a uniform strategy
($\IIH\hyb\bn = \IIS$), and two evaluators ($\no$ and
$\stgy{HIH}\hyb\stgy{IIS}$) define equivalent strategies modulo commuting
redexes (mcr). We find $4$ uniform evaluators/strategies in the literature
($\bn$, $\he$, $\bv$, $\ao$). The remaining $4$ are novel ($\ho$, $\IIS$,
$\SIS$, $\ISI$) where the first two have interesting properties and uses. We
find $6$ hybrid evaluators in the literature that define $6$ hybrid strategies
($\hr$, $\no$, $\hn$, $\am$, $\sn$, $\ha$). The remaining $26$ evaluators are
novel. Some define strategies with interesting properties and uses. In
particular, $\bs$ and $\so$ can be used for eager evaluation of general
recursive functions using delimited CPS and thunking protecting-by-variable.
The other evaluators and strategies remain to be explored.

There are $22$ eval-readback evaluators that are LWF-equivalent to $22$
balanced hybrid eval-apply evaluators, $6$ of them mcr. (There are more hybrid
evaluators than eval-readback evaluators.) The equivalences mcr inform the
equivalences mcr within the eval-readback style. For instance,
$(\Rb\Ev)\Rb\circ\IIS$ is equivalent mcr to $\HIH\hyb\IIS$ and the latter to
$\no$, thus $(\Rb\Ev)\Rb\circ\IIS$ is equivalent mcr to
$(\Rb\Ev)(\Rb\Ev)\circ\bn$. One eval-readback evaluator $\I(\Rb\Ev)\circ\bn$
defines the uniform strategy $\IIS$. We find $4$ eval-readback evaluators in
the literature (head reduction, normal order, \texttt{byName} / hybrid normal
order, \texttt{byValue} / strong reduction). We have not found any of the
remaining $18$ eval-readback evaluators in the literature, but we have not
searched exhaustively.

Figure~\ref{fig:tarpit} illustrates the position of our one-step equivalence
proof within the wider context of correspondences between operational
semantics devices. The top diagram in the figure is based on the diagrams in
\cite[p.\,545--547]{DM09} and shows the addition of the space of big-step
eval-readback evaluators that LWF-transform into the subspace of big-step
balanced hybrid eval-apply evaluators. The so-called `functional
correspondence' \cite{Rey72,ABDM03,ADM05,DM09,DJZ11} begins at the space of
direct-style big-step eval-apply evaluators (called `reduction-free
normalisers' in \cite{DM09,Dan05}) which can be transformed into abstract
machines by CPS transformation and defunctionalisation steps. Both steps are
invertible by refunctionalisation and direct-style transformation.  The
abstract machines are transformed (compiled) to unstructured programs using
\GOTO\ statements to implement the control structures of the structured
abstract machines.

\begin{figure}
  \scriptsize
  \begin{center}
    \begin{tikzpicture}
    \node[draw,circle,minimum size=2.2cm,very thick] (evalreadback) {};
    \node[below of=evalreadback,node distance=43,%
      text width=1.9cm,align=center]{eval-readback evaluators};

    \node[draw,circle,xshift=3.2cm,minimum size=2.2cm,fill=lightgray]
    (evalapply) {};
    \node[below of=evalapply,node distance=43,text width=1.9cm,align=center]
    {eval-apply evaluators};

    \node[draw,dashed,circle,xshift=3.2cm,minimum size=1.5cm,text width=1.1cm,
    align=center, very thick, fill=white] (balanced) {balanced hybrid};

    \draw[-stealth,very thick]
    (tangent cs:node=evalreadback,point={(balanced.south)},solution=2) --
    (tangent cs:node=balanced,point={(evalreadback.south)});
    \draw[-stealth,very thick]
    (tangent cs:node=evalreadback,point={(balanced.north)},solution=1) --
    (tangent cs:node=balanced,point={(evalreadback.north)},solution=2)
    node[pos=0.36,above,rotate=-6.2] {{\tiny LWF transformation}};

    \node[draw,circle,xshift=6.4cm,minimum size=2.2cm] (cps) {};
    \node[below of=cps,node distance=43,text width=1.9cm,align=center]
    {evaluators in\\CPS};

    \node[draw,dashed,circle,xshift=6.4cm,minimum size=1.5cm, align=center]
    (cpssmall) {};

    \draw[-stealth]
    (tangent cs:node=evalapply,point={(cpssmall.south)},solution=2) --
    (tangent cs:node=cpssmall,point={(evalapply.south)});
    \draw[-stealth]
    (tangent cs:node=evalapply,point={(cpssmall.north)},solution=1) --
    (tangent cs:node=cpssmall,point={(evalapply.north)},solution=2)
    node[pos=0.36,above,rotate=-6.2] {{\tiny CPS transformation}};

    \node[draw,circle,xshift=9.6cm,minimum size=2.2cm] (abstractmachines) {};
    \node[below of=abstractmachines,node distance=43,text width=5cm,align=center]
    {abstract machines\\(iterative structured programs)};

    \node[draw,dashed,circle,xshift=9.6cm,minimum size=1.5cm, align=center]
    (abstractsmall) {};

    \draw[-stealth]
    (tangent cs:node=cps,point={(abstractsmall.south)},solution=2) --
    (tangent cs:node=abstractsmall,point={(cps.south)});
     \draw[-stealth]
    (tangent cs:node=cps,point={(abstractsmall.north)},solution=1) --
    (tangent cs:node=abstractsmall,point={(cps.north)},solution=2)
    node[pos=0.39,above,rotate=-6.2] {{\tiny defunctionalisation}};

    \node[draw,circle,xshift=12.8cm,minimum size=2.2cm] (unstructured) {};
    \node[below of=unstructured,node distance=43,text width=1.9cm,align=center]
    {unstructured programs};

    \node[draw,dashed,circle,xshift=12.8cm,minimum size=1.5cm, align=center]
    (unstructuredsmall) {};

    \draw[-stealth]
    (tangent cs:node=abstractmachines,point={(unstructuredsmall.south)},solution=2) --
    (tangent cs:node=unstructuredsmall,point={(abstractmachines.south)});
    \draw[-stealth]
    (tangent cs:node=abstractmachines,point={(unstructuredsmall.north)},solution=1) --
    (tangent cs:node=unstructuredsmall,point={(abstractmachines.north)},solution=2)
    node[pos=0.45,above,rotate=-6.2] {{\tiny compiler}};
  \end{tikzpicture}

  \vspace{1cm}

  \begin{tikzpicture}
      \matrix (m) [matrix of nodes, row sep=2em, column sep=2.6cm,]
      {
        \parbox{5.7em}{\centering \textbf{Big-step} Evaluators \\
          \color{dimgray}{Natural Semantics}} &
        \parbox{4.3em}{\centering Abstract Machines \\
          \color{dimgray}{State trans.}} &
        \parbox{6.4em}{\centering Small-step Evaluators \\
          \color{dimgray}{Reduction Semantics}} &
        \parbox{5em}{\centering Search Functions \\
          \color{dimgray}{SOS}\\\color{white}{padding}}  \\
      };
      \draw[implies-implies,double]  (m-1-1) -- node[above] {Functional Corresp.}
      (m-1-2);

      \draw[-implies,double]
      (m-1-3) -- node[above] {Syntactic Corresp.} (m-1-2);

      \draw[-implies,double]
      (m-1-4) -- node[above] {} (m-1-3);
    \end{tikzpicture}
\end{center}
\caption{\emph{Top:} Functional correspondence diagram based on
  \cite[p.\,547]{DM09} and prefixed with the LWF transformation of
  eval-readback evaluators into balanced hybrid eval-apply evaluators.
  \emph{Bottom:} Correspondence diagram between operational semantics devices
  in \cite[p.\,177]{GPN14}. The one-step equivalence proof presented in this
  paper extends the functional correspondence to the left for big-step
  eval-readback evaluators.}
  \label{fig:tarpit}
\end{figure}

The functional correspondence applies to arbitrary direct-style eval-apply
evaluators, including uniform and unbalanced hybrid evaluators, which are
located within the grey ring. The transformations in the diagram inject a
source space into a target subspace, \eg\ not every evaluator in CPS
functionally-corresponds with a source direct-style eval-apply evaluator, but
those that do can be transformed back and forth.

As elegantly argued in \cite[pp.\,545--547]{DM09}, straying from the image of
the transformation (staying within the rings of more expressive power) with
good reasons is a clear indication that a useful control structure or operator
may be missing in the source space. For instance, programmers write evaluators
in CPS with non-tail calls `thereby introducing the notion of delimited
control [\ldots]; this effect can be obtained in direct style [in the source
space] by adding [\ldots] delimited-control operators [\ldots]'
\cite[p.\,546]{DM09}. Similarly in the space of structured programs, the case
against the \GOTO\ statement in \cite{Dij68} is not that it is unconditionally
harmful but that `straying with good reason is a clear indication that a
useful control construct is missing in the source language. For example, C and
Pascal programmers condone the use of GOTO for error cases because these
languages lack an exception mechanism' \cite[p.\,546]{DM09}.

There may be good reasons for straying from the image of the LWF
transformation in the case of unbalanced hybrid strategies, such as hybrid
applicative order ($\ha$) and spine applicative order ($\so$). Readback is a
selective-iteration-of-eval function. An unbalanced strategy requires more
work than eval on operands of redexes beyond distributing eval on abstractions
and neutrals. There may be a missing control structure $X$ that enables the
equivalence between eval-readback-with-$X$ and arbitrary hybrid eval-apply
evaluators.

The bottom diagram in Figure~\ref{fig:tarpit} shows a wider correspondence
diagram between operational semantics devices \cite[p.\,177]{GPN14}. The
transformation steps employed by the correspondences are semantics-preserving
and the implementation of the devices in code define the same evaluation
strategy. The figure shows on the left the functional correspondence zoomed
out. On the right it shows the so-called `syntactic correspondence' to
abstract machines from small-step evaluators that implement context-based
reduction semantics \cite{Dan05,BD07,DM08,DJZ11}. The small-step evaluator
(called a `reduction-based normaliser' in \cite{Dan05}) is a loop that
iterates the following steps: (i) the decomposition of the term into a unique
context with the permissible redex within the hole, (ii) the contraction, and
(iii) the recomposition of the contractum into a term. The evaluator
transforms to an abstract machine by the following non-invertible steps:
refocusing (which optimises the iteration loop), LWF to fuse decomposition and
recomposition, and inlining of iteration. Notice the different use of LWF to
which we have already pointed to in the introduction.  There is a further
unnamed transformation from small-step structural operational semantics (SOS)
to small-step context-based reduction semantics witnessed by the `search
function' that implements the compatibility rules of the SOS that traverse a
term to locate the redex. An implementation of a search function transforms to
the small-step evaluator by the following non-invertible steps:
CPS-transformation, simplification, defunctionalisation, and turning the
search function into a decomposition function which, additionally to the next
redex, delivers the context. The simplification and defunctionalisation of the
search function reveals the continuation stack, which is not revealed by a
whole implementation of a SOS that searches the input term, contracts the
redex, and delivers the next reduct.

Syntactic correspondences have been shown for specific uniform and hybrid
strategies, weak-reducing and full-reducing, and connected through the
functional correspondence with big-step eval-apply evaluators
\cite{Mun08,GPN13,DKM13,GNMN13,DJ13,GPN14,BCZ17}. The further connection with
eval-readback evaluators is provided by our one-step equivalence proof. The
eval-readback evaluators are obtained from balanced hybrid eval-apply
evaluators by instantiating the equations
(Proposition~\ref{thm:one-step-equiv-thm}).

\section{Related work}
\label{sec:related-work}
We have discussed related work throughout the paper, especially in the
introduction and Sections~\ref{sec:prelim:cbv}, \ref{sec:strategies},
\ref{sec:uniform-vs-hybrid}, and Appendix~\ref{app:lambda-value}. Big-step
evaluators are ubiquitous in the literature on lambda calculus and functional
programming.  We have not found higher-order generic evaluators like those of
Section~\ref{sec:generic-reducer} in the literature, but we believe they must
have been considered before.

The survey and structuring of the strategy space presented in this paper
significantly emend and amend the seminal survey and structuring presented in
\cite{GNG10,GP14} which in turn extend the survey and structuring in
\cite{Ses02}. The LWF transformation presented in this paper on plain
arbitrary evaluators that instantiate generic evaluators is simpler and more
direct than the LWF transformation in \cite{GP14} on generic evaluators that
work with boolean triples and the hybridisation operation ($\hyb$).

A recent work \cite{BCD22} presents a strategy space (called a zoo) that
includes call-by-name, call-by-value (call-by-$\WNF$ and lambda-value's), head
and full-reducing versions, and non-deterministic and right-to-left
versions. The work focuses on small-step context-based reduction semantics to
define strategies, using defunctionalised continuations that include contexts
and other useful information. It connects the uniform-style property (grammars
with one non-terminal, \ie\ one continuation stack) with determinism, and it
structures the space by building strategies from (or decomposing strategies
into) combinations of basic fine-grained sets of contexts (weak contexts, head
contexts, left or right contexts, etc.) and subsets of specific terms
(intermediate and final results). A hybrid strategy as presented in this paper
is such that it has contexts in one (maybe more) of its building fine-grained
contexts not closed under inclusion. The idea of combining fine-grained
contexts is related to the notion of dependence \cite{GPN14}. Despite the
differences between our approach and \cite{BCD22}, both find and structure
many identical left-to-right deterministic strategies: all the well-known
strategies and the novel uniform ones presented in \cite{GNG10,GP14}, \eg\
$\IIS$ and $\ISI$ are respectively $\mathit{low}$ and $\mathit{cbwh}$ in
\cite{BCD22}. The zoo does not consider many hybrid strategies ($\am$, $\ha$,
$\so$, $\bs$, and the unnamed hybrids in
Figure~\ref{fig:summary-end-figure-eval-apply}), but both strategy spaces
might be unified or superseded by extending the catalogue of basic
fine-grained contexts with those required to define, through their
combination, the full hybrid space. We take this as a confirmation that the
notion of uniform/hybrid strategy is a fruitful structuring criterion.

Hybrid strategies abound. They are essential to completeness and `completeness
for' in the non-strict and strict spaces. For the non-strict strategies, the
underlying explanatory technical notion is `needed reduction'
(Appendix~\ref{app:completeness-leftmost-spine}) which can be extended to
`completeness for' other kinds of final result than normal forms. We summarise
the completeness for some uniform and balanced hybrid strategies in
Figure~\ref{fig:summary-survey} (page~\pageref{fig:summary-survey}):
\begin{itemize}
\item $\bn$, $\IIS$, $\hr$, and $\no$ are respectively complete for $\WHNF$,
  $\WNF$, $\HNF$, and $\NF$ by the Quasi-Leftmost Reduction Theorem
  (Appendix~\ref{app:completeness-leftmost-spine}).

\item $\he$ and $\hn$ are respectively complete for $\HNF$ and $\NF$, with
  redexes in $(\lambda x.\HNF)\Lambda$ by the Needed Reduction Theorem
  (Appendix~\ref{app:completeness-leftmost-spine}).\footnote{By the properties
    of $\HNF$, it turns out that restricting the redexes to $(\lambda
    x.\HNF)\Lambda$ does not alter the completeness of $\he$ and $\hn$ with
    unrestricted redexes.}
\end{itemize}
Completeness results can be given for strict strategies (which are incomplete
in the pure lambda calculus) using notions of reduction other than the
$\beta$-rule, namely, using calculi with notions of reduction with permissible
redexes where the operand and/or the operator has to be in some final
form. Concretely, we conjecture the following:
\begin{itemize}
\item $\bv$, $\am$, and $\sn$ are respectively complete for $\WNF$,
  $\NT{VHNF}$, and $\NF$ with redexes in $(\lambda x.\Lambda)\WNF$.

\item $\ho$ and $\bs$ are respectively complete for $\HNF$ and $\NF$ with
  redexes in $(\lambda x.\HNF)\HNF$.
\end{itemize}
The conjectures rest on these principles:
\begin{enumerate}
\item The strategies above are outermost when taking into account the
  corresponding notion of permissible redex, \ie\ an outermost non-permissible
  redex may become the next permissible redex to be contracted iff contracting
  some inner permissible redex makes the outer one a permissible redex.

\item All the notions of permissible redexes are stable under contraction,
  \ie\ a permissible redex remains so when any of its permissible sub-redexes
  is contracted.

\item Contraction of permissible redexes is confluent, \ie\ given a term with
  two permissible redexes, the terms resulting after the contraction of the
  two permissible redexes in different order reduce to a common term.
\end{enumerate}
We finally discuss the connection between hybridity, outermost,
completeness-for, `refocusing', and `backward overlaps' \cite{DJ13}. A set of
reduction rules is backward overlapping if a subterm of the left hand side of
one rule unifies with the right hand side of the same or a different rule.
Intuitively, this means that a potential redex (the term that will unify with
the left hand side of the first rule) will turn into an actual redex after
contracting some sub-redex (the redex whose contractum is the right hand side
of the second rule) of the former one. As pointed out in \cite{DJ13}, in the
presence of backward overlaps the `refocusing' techniques that deforest the
recomposition and decomposition of a contractum in a context-based reduction
semantics may miss the outer redex when this context-based reduction semantics
implements an outermost strategy (\ie\ refocusing will carry on decomposing
the term at the scope of the inner redex, instead of recomposing the term up
to the scope of the outer redex). In order to prevent this, \cite{DJ13}
propose a backtracking step that will recompose the term just enough to reach
the scope of the outer redex.

All the strategies discussed in the paragraph above about `completeness-for'
are outermost, and the notions of reduction induced by the corresponding
permissible redexes are backward overlapping. (An outer non-permissible redex
may turn into a permissible redex by contracting some inner permissible
redex.) Some of the strategies are uniform, because their permissible redexes
do not contain any permissible sub-redexes to be contracted according to the
strategy, and the problem of missing an outer redex by decomposing at the
scope of the inner redex never occurs. This is easy to see in the non-strict
case, where the operator $(\lambda x.\Lambda)$ in $\WHNF$ is irreducible for
$\bn$ and $\IIS$, and the operator $(\lambda x.\HNF)$ in $\HNF$ is irreducible
for $\he$, and where the operands of a permissible redex are never reduced.
It can also be seen in the strict case, where the operator $(\lambda
x.\Lambda)$ in $\WNF$ and the operand in $\WNF$ are irreducible for $\bv$, and
where the operator $(\lambda x.\HNF)$ in $\HNF$ and the operand in $\HNF$ are
irreducible for $\ho$. The remaining outermost and complete-for strategies are
inexorably hybrid because their permissible redexes may contain permissible
sub-redexes. In this case, a less-reducing subsidiary strategy must be called
at the scope of the potential redex (at the operator in the non-strict case,
and at the operator and the operand in the strict case) in order to find the
inner redexes that will turn the outermost potential redex into an actual one,
but without reducing any further. In \cite{GPN13}, a refocusing solution based
on the shape invariant of the continuation stack used in the decomposition
step is shown to be applicable to such hybrid strategies. This solution is
also applicable to the outermost strategy presented in \cite{DJ13}, which is
hybrid.

\section{Future work}
\label{sec:future-work}
We list a few of many possibilities for future work. To analyse the strategies
that stray from the image of the LWF transformation and their possible
relation with control operators mentioned in the conclusion. To investigate
one-step equivalence when the subsidiary evaluator implements a shallow hybrid
strategy (Section~\ref{sec:hy-systematic}). To investigate the relationship
between, on the one side, non-strict leftmost-outermost and spine strategies,
and on the other side, strict leftmost-innermost and spine-ish strategies,
especially in the wider context of `needed reduction'
(Appendix~\ref{app:completeness-leftmost-spine}).  To (dis)prove the duality
conjecture (Section~\ref{sec:summary-survey}). To transcribe the structuring
of the strategy space to the search-function side of the syntactic
correspondence in order to derive plain but arbitrary small-step evaluators
and abstract machines that connect to the functional correspondence. To
investigate the connection with the small-step framework in \cite{BCD22}. To
mechanise the LWF proof involving constraints on parameters. To develop the
functional correspondence from parametric eval-apply evaluators into
parametric abstract machines. To apply the principles and criteria presented
in this paper to structuring the strategy space of calculi with explicit
substitutions.

\section*{Acknowledgements}
Our interest in strategies and evaluators was stimulated by the textbooks and
papers cited in the introduction. The following were particularly influential:
\cite[Chp.\,9]{FH88}, \cite[Chp.\,13]{Rea89}, \cite[Chps.\,11--12]{Mac90}, and
\cite{McC60,Lan64,Rey72,Plo75,Ses02,GL02,ABDM03}. We thank Emilio J.\ Gallego
Arias for contributing to \cite{GNG10}. We thank the following people for
their feedback and encouragement at talks and presentations on the survey and
structuring material in \cite{GNG10,GP14}: Olivier Danvy, Jan Midtgaard, Ian
Zerny, Peter Sestoft, Ruy Ley-Wild, Noam Zeilberger, Richard Bird, Jeremy
Gibbons, Ralf Hinze, Geraint Jones, Bruno Oliveira, Meng Wang, and Nicolas
Wu. We specially thank Jeremy Gibbons for permission to name the Beta Cube
after him. We thank Tomasz Drab for a discussion on the strategy space in
\cite{BCD22} and the strategy space presented in this paper. We finally thank
the three anonymous reviewers of \emph{LMCS} for accepting to review such a
long paper and for their insightful criticism and suggestions which have led
us to make important corrections, clarifications, and additions.

\appendix
\section{Pure lambda calculus concepts and notation}
\label{app:prelim:lambda}
The set $\Lambda$ of pure lambda calculus terms consists of variables,
abstractions, and applications. Uppercase, possibly primed, letters (\eg~$B$,
$M$, $N$, $M'$, $M''$, etc.) range over $\Lambda$.  Lowercase, also possibly
primed, letters (\eg~$x$, $y$, $z$, $x'$, $x''$, etc.) range over a
countably-infinite set of variables $\Var$. Lowercase variables act also as
\emph{meta-variables} ranging over $\Var$ when used in term forms. For
example, the concrete term $\lambda y.zy$ with parameter $y$ is of the form
$\lambda x.B$ where $x$ is a meta-variable. We say $x$ is the parameter and
$B$ is the \emph{body}. A \emph{free variable} is not a parameter of an
abstraction.  Terms with at least one free variable are \emph{open
  terms}. Terms without free variables are \emph{closed terms}. In an
application $MN$ we say $M$ is the \emph{operator} and $N$ is the
\emph{operand}.  The application operation is denoted by tiny whitespace or by
juxtaposition.  Abstractions associate to the right. Applications associate to
the left and bind stronger than abstractions. A few example terms: $x$, $xy$,
$xx$, $\lambda x.x$, $\lambda x.\lambda y.yx$, $x(\lambda x.x)$, $xy(\lambda
x.x)$, $(\lambda y.y)(xy)$, and $(\lambda x.y)(xyz)(\lambda x.\lambda
y.xy)$. Uppercase boldface letters stand for concrete terms. For example,
$\CH{I}$ abbreviates the identity term $\lambda x.x$, and $\CH{Y}$ abbreviates
the non-strict fixed-point combinator $\lambda f.(\lambda x.f(xx))(\lambda
x.f(xx))$. For more fixed-point combinators and their explanation see
\cite[p.\,130]{Ros50}, \cite[p.\,178]{CF58}, \cite[p.\,36]{HS08}, and
\cite[p.\,131]{Bar84}.

Every term of $\Lambda$ has the form $\lambda x_1 \ldots \lambda x_n.H N_1
\cdots N_m$ with $n,m\geq0$. The head term $H$ is either a variable (the `head
variable') or a redex (the `head redex'). A \emph{redex} is an application of
the form $(\lambda x.B)N$. We write an arrow `$\rel$' to denote one-step
beta-reduction. We write $\cas{N}{x}{B}$ for the \emph{capture-avoiding
  substitution} of the operand $N$ for the free occurrences (if any) of
parameter $x$ in the abstraction body $B$. The resulting term is called a
\emph{contractum}. The function notation $\cas{N}{x}{B}$ that wraps the
subject of the substitution $B$ in parentheses comes from \cite{Plo75}.  This
notation is itself based on the original notation $[N/x]B$ without parenthesis
of \cite{CF58}. Both notations flow in English when read from left to right:
`substitute $N$ for free $x$ in $B$'. The action is substitute-for, not
replace-by. Other notations are, of course, valid.

\emph{Normal forms}, collected in set $\NF$ (Figure~\ref{fig:lam-sets}), are
terms without redexes. They consist of variables, and abstractions and
applications in normal form.  \emph{Neutrals} have been introduced in
Section~\ref{sec:the-setting}. Examples of neutrals are $xy$, $x(\lambda
x.y)$, and $x((\lambda x.y)z)(y(\lambda x.x))$. The last neutral has a neutral
$y(\lambda x.x)$ as its second operand. A neutral cannot evaluate to a redex
unless a suitable operand can be substituted for the head variable. For this,
it is necessary for the neutral to occur within an abstraction that has the
head variable for parameter. For example, in the neutral $xN$ the head
variable is free and cannot be substituted. In $(\lambda x.xN)(yM)$, the
abstraction body $xN$ is a neutral, but the substitution of $yM$ for $x$
delivers $yMN$ which is itself a neutral, not a redex. In contrast, in
$(\lambda x.xN)(\lambda y.y)$ the contractum $(\lambda y.y)N$ is a redex.

An evaluation \emph{strategy} is defined generically as a function subset of
the reduction relation, \eg~\cite[p.\,53]{Bar84} \cite[p.\,130]{Ter03}. For
instance, a one-step strategy $\st$ is defined as a total function on terms
such that either $M\equiv\st(M)$ for all $M\in\NT{F}$ where $\NT{F}$ is a set
of final forms, or $M\rel\st(M)$ for all $M\not\in\NT{F}$. An iterative
multiple-step definition is obtained by transitive closure:
$M\rel^{+}\st^n(M)$ for $n>0$. For each input term the strategy determines a
unique \emph{evaluation sequence} up to, if finite, a final result in
$\NT{F}$. The strategy \emph{converges} (opposite, \emph{diverges}) when the
evaluation sequence is finite (opposite, infinite). The term is said to
converge (opposite, diverge) under the strategy. The big-step notation,
$\st(M) = N$, can be translated in terms of one-step $\st$ as follows: either
$M\in\NT{F}$ and $\st(M)\equiv M \equiv N$, or $M\not\in\NT{F}$ and
$M\rel^{+}\st^n(M)$ and $\st^n(M)\equiv N$ for $n>0$ and $N\in\NT{F}$.

A generic definition of a strategy is uninformative because the particular
redex contracted in the reduction is unspecified. We illustrate the evaluation
sequence of a strategy by underlining the chosen redex: $\underline{(\lambda
  x.\lambda y.xy)(\lambda x.z)}x \rel \underline{(\lambda y.(\lambda x.z)y)x}
\rel \underline{(\lambda x.z)x} \rel z$. The second term in the sequence has
another redex, $(\lambda y.\underline{(\lambda x.z)y})x$, which is not chosen
by the strategy. In this example a final result is a normal form. Operational
semantics styles \cite{Plo81,Lan64,Kah87,FH92} define strategies precisely by
stipulating how the chosen redex is located.

In the pure lambda calculus it is undecidable whether an arbitrary term can be
evaluated to normal form. A strategy that aims at delivering normal forms is
\emph{full-reducing}. A strategy that delivers a normal form when the latter
exists is \emph{normalising}. A strategy that diverges for terms for which a
normal form exists is \emph{not normalising}. The archetypal term without a
normal form is $(\lambda x.xx)(\lambda x.xx)$ which is abbreviated as
$\OMEGA$. The word \emph{complete} (opposite, \emph{incomplete}) is also used
in the literature as a synonym of `normalising'. We have not been able to
trace the genesis of `complete'. We think it comes from the idea that the
strategy's evaluation sequence completes to (reaches) the final form, or
perhaps from extending to the strategy the `complete' property of the
reduction relation which obtains when it is confluent and terminating
\cite[p.\,13]{Ter03}. A strategy determines a unique evaluation sequence so,
if it converges in a final form, then it may be called complete by extension.
We prefer `complete' because of its generality: final forms need not be
exclusively normal forms. Following an anonymous reviewer's suggestion, we say
that a strategy is `complete for' the particular set of terms.

\section{Call-by-value in the lambda-value calculus}
\label{app:lambda-value}
The motivation for the lambda-value calculus of \cite{Plo75} was to prove that
the \SECD\ abstract machine of the ISWIM language \cite{Lan65} was correct
with respect to a calculus with open and divergent terms and a reduction
relation with blocking, and provide a standardisation theorem for a complete
strategy with blocking that delivers values as results. The blocking of
neutrals guarantees the confluence of the reduction relation.

The correctness of \SECD\ is proven with respect to the big-step evaluator
`$\mathrm{eval}_V$' \cite[p.\,130]{Plo75} and the small-step evaluator
`$\rel_V$' \cite[p.\,136]{Plo75} by (1) a standardisation theorem
\cite[Thm.\,3,\,p.\,137]{Plo75} that proves that `$\rel_{\mathrm{V}}$' is
complete for values, (2) by the one-step equivalence between
`$\mathrm{eval}_V$' and `$\rel_{\mathrm{V}}$' (proven on paper
\cite[Thm.\,4,\,p.\,142]{Plo75}, not by program transformation), and (3) by
the one-step equivalence of `$\mathrm{eval}_V$' and \SECD\ (also proven on
paper \cite[Thm.\,1,\,p.\,131]{Plo75}, not by program transformation). The
proofs of (2) and (3) assume \emph{closed input terms}. Indeed, the evaluators
do not correspond for open input terms because `$\mathrm{eval}_V$' ignores
neutrals whereas `$\rel_V$' evaluates the first operand of a neutral in its
4th clause. Stuck terms do not arise with closed input terms. Under this
assumption, `$\mathrm{eval}_V$' and `$\rel_V$' implement call-by-$\WNF$
(Section~\ref{sec:prelim:cbv}). With open terms, call-by-$\WNF$ differs from
`$\rel_V$' and cannot even be compared to `$\mathrm{eval}_V$' because the
latter does not consider neutrals.

Discussions about the pure version of lambda-value, its operational and
denotational theory, and its full-reducing strategies can be found among
others in \cite{EHR91,EHR92,PR99,RP04,AP12,CG14,GPN16,KMRDR21,AG22}. Normal
forms with stuck terms and a standardisation theorem for a full-reducing and
complete strategy were not considered in \cite{Plo75}. The issues have been
addressed among others in \cite{PR99,RP04,AP12,CG14,GPN16,KMRDR21,AG22}. A
standardisation theorem for a full-reducing complete strategy that delivers
lambda-value normal forms is given in \cite[p.\,70]{RP04}. A standardisation
theorem for a variant full-reducing strategy that evaluates blocks
left-to-right is given in \cite[Sec.\,7.1]{GPN16} along with a standard theory
based on the operational relevance (`solvability') and sequentiality of some
open neutrals and stuck normal forms. The relation between lambda-value and
the sequentiality of the linear lambda calculus was first studied in
\cite{MOTW95}. For the most complete up-to-date work at the time of writing on
lambda-value solvability see \cite{AG22}, and the references on
standardisation there cited.

\section{Overview of LWF}
\label{app:LWF}
The LWF optimisation is outlined in \cite[p.\,144]{OS07} using code written in
the maths font we have used in this paper for lambda calculus terms. To avoid
confusion, in this appendix we use the same code symbols as \cite{OS07} but in
typewriter font, which is appropriate because LWF is applied to code.

Fusion a is program-transformation optimisation that combines a composition
$\mathtt{f} \circ \mathtt{g}$ into a single function $\mathtt{h}$. LWF is a
general fusion optimisation where $\mathtt{f}$ and $\mathtt{g}$ can be general
recursive functions, of any algebraic type, and with any number of curried or
uncurried parameters. LWF is lightweight because it consists of a few simple,
terminating, and automatable steps that perform local syntactic
transformations without the need for heuristics. Thus, LWF can be embedded in
the inlining optimisation of a compiler.

Given two general recursive functions $\mathtt{f} = \mathtt{fix}\:
\mathtt{f}. \lambda \mathtt{x}. \mathtt{E}_{\mathtt{f}}$ and $\mathtt{g} =
\mathtt{fix}\: \mathtt{g}. \lambda \mathtt{x}.  \mathtt{E}_{\mathtt{g}}$, the
composition $\mathtt{f} \circ \mathtt{g}$ is fused by promoting $\mathtt{f}$
through $\mathtt{g}$'s fixed-point operator $\mathtt{fix}$ in the following
steps.
\begin{enumerate}
\item Inline $\mathtt{g}$ in $\mathtt{f} \circ \mathtt{g}$ to obtain
  $\mathtt{f} \circ (\lambda \mathtt{x}. \mathtt{E}_{\mathtt{g}})$ and then
  beta-convert that to $\lambda \mathtt{x}. \mathtt{f}\,
  \mathtt{E}_{\mathtt{g}}$.
\item Transform the application $\mathtt{f} \, \mathtt{E}_{\mathtt{g}}$ to
  $\mathtt{E}_{\mathtt{g}}'$ by distributing $\mathtt{f}$ to tail position in
  $\mathtt{E}_{\mathtt{g}}$.
\item Inline $\mathtt{f}$ once in $\mathtt{E}_{\mathtt{g}}'$ and simplify to
  obtain $\mathtt{E}_{\mathtt{f,g}}$.
\item Replace the occurrences of $\mathtt{f} \circ \mathtt{g}$ in
  $\mathtt{E}_{\mathtt{f,g}}$ by a new function name $\mathtt{h}$ and generate
    a new binding $\mathtt{h} = \mathtt{fix}\: \mathtt{h}. \lambda
    \mathtt{x}. \mathtt{E}_{\mathtt{h}}$.
\end{enumerate}
Because the first steps pull out $\mathtt{g}$'s case expression, the
transformation requires $\mathtt{f}$ to be strict (to pattern-match) on its
first argument.

The soundness of LWF is proven in \cite[Sec.\,3]{OS07} assuming a `standard
call-by-name' denotational semantics for a small core functional programming
language where $\mathtt{f}$ and $\mathtt{g}$ are least fixed points in the
space of continuous functions. We believe that semantics is assumed for
theoretical simplicity and reasons of space in a conference paper because LWF
is used on code written in \emph{both} lazy (non-strict, call-by-\emph{need},
\eg\ Haskell) and eager (strict, call-by-\emph{value}, \eg\ Standard ML)
functional programming languages.  Because a denotational soundness proof for
the latter would depend on the particular call-by-value theory of choice (\eg\
Appendix~\ref{app:lambda-value}), the authors left that proof as future work
and suggested one possible choice in \cite[Sec.\,9]{OS07}.

An alternative would have been to prove the call-by-value operational (rather
than denotational) soundness of the steps, but that would have also enlarged
the paper. At any rate, such correctness is assumed by LWF's syntactic
simplicity and by the empirical evidence provided by the correct results
obtained in the `expected typical cases' \cite[p.\,144]{OS07}, where the fused
program $\mathtt{h}$ is syntactically identical to a known existing correct
solution $\mathtt{h}'$.  The code examples with correctly fused results in
\cite{OS07} are written in strict Standard ML which certainly has no
call-by-name denotational semantics. Applications of LWF in strict Standard ML
presented in subsequent work \cite{Dan05,DM08,GPN14} confirm the soundness of
the transformation by producing fused programs which are syntactically
identical to the expected correct known ones.

The only caveat on soundness, discussed in \cite[Sec.\,6.3]{OS07}, occurs when
there is an expression in between a binding $\mathtt{x} =
\mathtt{g}\:\mathtt{e}$ and the application $\mathtt{f}\:\mathtt{x}$. These
two must be placed together to fuse them into $\mathtt{h}\:\mathtt{e}$, and
this may require swapping the in-between expression, changing the evaluation
order. For pure code, whether in lazy or eager functional programming
languages, the swapping may affect divergence order and performance in
particular cases \cite[p.\,152]{OS07}. When the swapped expressions have
side-effects, $\mathtt{f} \circ \mathtt{g}$ and $\mathtt{h}$ would differ in
side-effect order. It is for this reason that \cite[p.\,151,\,154]{OS07} state
that LWF is unsound for \emph{strict and impure languages}. However, the
problem affects \emph{impure code} in any language (the authors seem to
consider only the pure subset of Haskell without impure features
\cite[p.\,151]{OS07}). When there are no in-between bindings or the impact of
the change in evaluation order (divergence and side-effects) is factored in
and taken into account, the LWF transformation is unproblematically applied in
lazy and eager languages for strict and non-strict functions.

In our proof by LWF (Section~\ref{sec:lwf-steps}) a change of evaluation order
does occur in the swapping of bindings for neutral terms in
Step~(\ref{lwf:step-swap}). This change only affects the evaluation order of
redexes in operands $N_i$ of neutrals $x N_1 \cdots N_n$. Such redexes are
non-overlapping (commuting) and the change in evaluation order means the
one-step equivalence is `modulo commuting redexes' for some strategies. The
use of impure \texttt{IO} in the Haskell code is consistent with this because
the result term is printed as a side-effect, showing an equivalence modulo
commuting those redexes.

The LWF transformation does not involve eta-conversion, $(\lambda
\mathtt{x}.\: \mathtt{f}\:\mathtt{x}) \rel \mathtt{f}$, nor eta-expansion,
$\mathtt{f} \rel (\lambda \mathtt{x}.\: \mathtt{f}\:\mathtt{x})$, which is
fortunate, as the first is sensitive to types \cite{JG95} and the second
introduces the equation $\bot = \lambda \mathtt{x}. \bot$.

\section{Haskell's operational semantics}
\label{app:haskell-semantics}
Haskell was deliberately conceived as a research language that should avoid
`success'. This included avoiding standardisation and a formal semantics that
would either tie up the language or would have to keep up with its fast
evolution \cite{HHPW07}. However, there is a Haskell 2010 standard, \emph{The
  Haskell 2010 Language Report} (available online), that includes a static
semantics and also an operational semantics, as we clarify below. The latest
version of the language is known as `GHC Haskell', namely, the latest version
compiled by the Glorious Haskell Compiler which is the \emph{de facto}
compiler. The static and operational semantics of GHC Haskell are scattered in
research papers, some of them listed on the Haskell website and wiki pages.

Haskell 2010 has a standard operational semantics. More precisely, any version
of Haskell has an implementation-based operational semantics. Haskell has a
stable small core language, System FC, which is a direct-style (in opposition
to CPS-style) version of the well-known explicitly-typed System F with extra
support for type casts and coercions \cite{Eis20}. Haskell code is statically
type-checked and desugared into System FC code on which GHC performs the
optimisations and generates target code for the STG abstract machine
\cite{PJ92}. This code is, in turn, finally compiled to C or to native machine
code.\footnote{See Chp.\,10 of the
  \href{https://downloads.haskell.org/ghc/latest/docs/users_guide/}{GHC Users
    Guide} (visited: July 2023).} System FC has an operational semantics given
by its compilation scheme and execution on the STG machine. Whatever the
version of Haskell used, what needs to be explained is the static semantics
and the translation to System FC. These are provided for Haskell 2010 by
\emph{The Haskell 2010 Language Report} and the GHC implementation of the 2010
standard. In this sense, Haskell 2010 has a well-defined operational
semantics.

With the exception of the `arbitrary rank polymorphism' (\texttt{RankNTypes})
extension, the Haskell code in this paper conforms to the Haskell 2010
standard. Type classes for ad-hoc polymorphism are supported by the standard,
with the \texttt{Monad} type class and the \texttt{IO} instance shipped with
the Haskell 2010 standard library. Type classes have a well-known and simple
dictionary semantics \cite{HHPJW96}. The \texttt{RankNTypes} extension enables
GHC to type-check the type of our plain and generic evaluators which have a
\texttt{forall} bounded on a monad type parameter. This extension has no
run-time effect as it only allows the code to type-\emph{check}. The Haskell
2010 standard uses a type-inference algorithm that cannot infer
rank-polymorphic types whose inhabitants are representable in System FC. Such
types are innocuous, type-checkeable, and translatable to System FC. At any
rate, in our case the extension is entirely optional and can be dropped if
desired.

The use of strict monads for strict evaluation does require some
elaboration. Haskell evaluates expressions to a weak head normal form. The
2010 standard provides features for strict evaluation to head normal form:
\texttt{seq}, strict function application (\texttt{\$!}), and strict value
constructor application (or strict data-types), all of them discussed in
Sections 4.2.1 and 6.2 of \emph{The Haskell 2010 Language Report}. These
features are partly reconcilable with parametricity theorems and the
correctness of some fusion transformations \cite{JV04}. GHC also offers
extensions for evaluation to normal form (\texttt{deepseq} and fully evaluated
data-types).  However, we have used a non-strict data-type \texttt{Term} and
the strict \texttt{IO} monad to fully evaluate Haskell expressions (evaluator
calls) to print terms.

\section{Completeness of leftmost and spine strategies}
\label{app:completeness-leftmost-spine}
The completeness and `completeness for' of the leftmost strategies is
supported by the Standardisation Theorem~\cite{CF58} which states that a
strategy that chooses the leftmost redex (the leftmost-outermost redex when
considering the abstract syntax tree of the term) is complete. Strict
strategies are thus incomplete because they evaluate the operand of a
leftmost-outermost redex.

The Standardisation Theorem is relaxed by the Quasi-Leftmost-Reduction Theorem
\cite[p.\,40]{HS08} which states that a strategy that \emph{eventually}
chooses the leftmost-outermost redex is complete. Such strategies are called
\emph{quasi-leftmost}. They subsume the leftmost and spine strategies, whose
completeness and completeness-for is supported by the theorem.

The Standardisation and Quasi-Leftmost-Reduction theorems are further relaxed
by the `needed reduction' theorem which states that a strategy that eventually
chooses the `needed' redexes is complete \cite[p.\,208]{BKKS87}. A needed
redex is a redex that is reduced in every reduction sequence to normal form
when the latter exists. Needed reduction eventually reduces all the needed
redexes (or its residuals). Intuitively, the critical point is in applications
$MN$ where the reduction of $M$ or $N$ may diverge. The operator $M$ should be
reduced enough to deliver an abstraction. The operand $N$ should not be
reduced because it could be discarded. Only terms that eventually appear in
operator position need be evaluated, and just enough (to head normal form) to
contract the redex.

\bibliographystyle{alphaurl}
\bibliography{bc}

\end{document}